\documentclass[10pt,letterpaper]{article}
\usepackage[T1]{fontenc}
\usepackage[utf8]{inputenc}
\usepackage[letterpaper,margin=0.75in,columnsep=0.25in]{geometry}
\usepackage{graphicx}
\graphicspath{{./}}
\usepackage{amsmath,amsfonts,amssymb}
\usepackage{amsthm}
\usepackage{subcaption}
\usepackage{microtype}
\usepackage{chapterbib}
\usepackage[numbers,sort&compress,round]{natbib}

\usepackage{xcolor}
\usepackage[normalem]{ulem} 

\newif\ifrevisions

\revisionsfalse       

\ifrevisions

    \newcommand{\deleted}[1]{%
        \textcolor{blue}{\sout{#1}}%
    }

    \newcommand{\changed}[2]{%
        \textcolor{blue}{\sout{#1}}%
        \textcolor{blue}{#2}%
    }

    \newcommand{\movedfigure}[1]{%
        \textcolor{blue}{[Figure moved from #1.]}%
    }

    \newcommand{\revisionnote}[1]{%
        \textcolor{blue}{[#1]}%
    }

\else

    \newcommand{\deleted}[1]{}
    \newcommand{\changed}[2]{#2}
    \newcommand{\movedfigure}[1]{}
    
    \newcommand{\revisionnote}[1]{}

\fi

\usepackage{hyperref}
\hypersetup{
  colorlinks=true,
  linkcolor=blue,
  citecolor=blue,
  urlcolor=blue,
  pdftitle={Intermittent Cauchy walks enable optimal 3D search across target shapes and sizes},
  pdfauthor={Matteo Stromieri, Emanuele Natale, Amos Korman}
}
\DeclareMathOperator{\Prob}{\mathbb{P}}
\newcommand{\bigO}[1]{\mathcal{O}\!\left(#1\right)}
\newcommand{\eres}[3]{\mathcal{R}^{#3}(#1 \leftrightarrow  #2)}
\newcommand{\torus}{\mathbb{T}_n}
\newcommand{\lmax}{\ell_{max}}
\newcommand{\norm}[1]{\left\| #1 \right\|}
\newcommand{\E}{\mathbb{E}}
\newcommand{\R}{\mathbb{R}}
\newcommand{\tdetect}[1][cauchy]{t_{\text{detect}}^{X^{#1}}(S)}
\newcommand{\mdetect}[1][cauchy]{m_{\text{detect}}^{X^{#1}}(S)}
\newcommand{\Var}{\operatorname{Var}}

\newcommand{\mhit}{m_\text{hit}}
\newcommand{\mmix}{m_\text{mix}}

\newcommand{\dirform}[2]{\mathcal{E}_{#1}(#2)}

\newtheorem{theorem}{Theorem}[section]
\newtheorem{fact}[theorem]{Fact}
\newtheorem{claim}[theorem]{Claim}
\newtheorem{corollary}[theorem]{Corollary}
\newtheorem{definition}[theorem]{Definition}
\newtheorem{lemma}[theorem]{Lemma}
\newtheorem{observation}[theorem]{Observation}
\begin{document}
\twocolumn[
\begin{center}
{\LARGE\bfseries Intermittent Cauchy walks enable optimal 3D search across target shapes and sizes\par}
\vspace{0.8em}
{\normalsize Matteo Stromieri\textsuperscript{a,b}\quad Emanuele Natale\textsuperscript{c}\quad Amos Korman\textsuperscript{a,*}\par}
\vspace{0.5em}
{\small
\textsuperscript{a}Department of Computer Science, University of Haifa, Haifa, Israel\par
\textsuperscript{b}COATI, Université Côte d’Azur, Sophia Antipolis, France\par
\textsuperscript{c}CNRS, COATI, Université Côte d’Azur, Sophia Antipolis, France\par
\textsuperscript{*}Correspondence: amoskorman@cs.haifa.ac.il\par}
\end{center}
]
\begin{abstract}
Efficiently locating rare targets requires balancing the exploration of novel space with the exploitation of local vicinities. While this tradeoff is well-documented in simplified environments, its behavior in realistic three-dimensional environments remains unclear. Here we analyze intermittent L\'evy walks in three dimensions and show that search efficiency becomes strongly dependent on target geometry, with volume, surface area, and elongation governing performance across the L\'evy spectrum. Against this geometric complexity, the Cauchy walk ($\mu = 2$) emerges as a robust, near-optimal solution across a wide range of target geometries. Specifically, we prove that for convex targets of surface area $\Delta$ in a domain of volume $n$, the optimal expected detection time scales as $n/\Delta$, and that the Cauchy walk achieves this scaling up to polylogarithmic factors, while other L\'evy strategies are suboptimal in different geometric regimes. Our findings provide a theoretical framework for biological movement patterns, including marine predator foraging and immune surveillance, and inform the design of efficient search algorithms for engineered systems such as autonomous nanoscale explorers.
\end{abstract}
\section*{Introduction}
Finding rare targets in large, complex environments is a fundamental challenge faced by organisms and engineered systems alike. Marine predators search vast three-dimensional environments for schooling prey and resource patches \cite{sims2008scaling,humphries2010environmental}. At microscopic scales, microzooplankton and larval fish forage for motile prey within heterogeneous, turbulent water columns, where resources are structured in filaments and thin layers rather than isolated points \cite{bartumeus2003helical}. Similarly, in vertebrate tissues, immune cells, such as T-cells and dendritic cells, patrol complex microenvironments to detect infected or malignant cells \cite{harris2012generalized,krummel2016t}. Metastatic cancer cells likewise navigate tissue environments, potentially facilitating the search for suitable sites for metastatic seeding \cite{huda2018levy}. Beyond biological systems, engineered agents face analogous challenges, from tracking chemical plumes fragmented by turbulent flows \cite{pasternak2009levy} to synthetic nanorobots navigating physiological environments to locate cancer cells \cite{hortelao2018targeting}. In these settings, targets are rarely isolated points but extended spatial regions whose size and morphology vary.

In the absence of directional cues or prior information about the target, the search process imposes a fundamental exploration–exploitation tradeoff: searchers must relocate to distant regions to explore new areas of space, yet maintain sufficient local dwell time to detect and exploit targets once nearby \cite{benichou2011intermittent,lomholt2008levy}. While this tradeoff is well understood in simplified or low-dimensional settings, it remains unclear how it is shaped by the structure of realistic three-dimensional environments. In particular, it is not understood how environmental structure influences which exploration–exploitation balance is most effective.

A widely used mathematical framework for modeling this exploration–exploitation balance is the L\'evy walk \cite{viswanathan2008levy,zaburdaev2015levy}. In a L\'evy walk, the lengths $\ell$ of steps are drawn from a heavy-tailed distribution, $p(\ell)\propto \ell^{-\mu}$. The exponent $\mu$ serves as a single  parameter controlling the movement statistics: a smaller $\mu$ yields rarer but longer ballistic flights (favoring broad exploration), whereas a larger $\mu$ produces shorter steps and diffusive behavior (favoring local exploitation). Unlike L\'evy flights, which assume instantaneous jumps, L\'evy walks have finite velocity and are therefore physically realizable. The special case $\mu \approx  2$, often termed the Cauchy regime, lies at the critical boundary between these strongly exploratory and highly localized strategies.

Heavy-tailed displacement patterns with exponents in this Cauchy regime have been reported across diverse systems, including marine searchers
\cite{humphries2010environmental,sims2008scaling,bartumeus2003helical}, spider monkey foragers \cite{boyer2006scale}, patrolling immune cells \cite{harris2012generalized}, and migrating cancer cells \cite{huda2018levy},
see also reviews in \cite{viswanathan2008levy,zaburdaev2015levy} and the references therein.
These empirical observations are frequently interpreted through the lens of the L\'evy Flight Foraging Hypothesis (LFFH), which posits that natural selection favors movement patterns with $\mu\approx 2$ due to their presumed optimality when foraging for sparse targets in unpredictable environments \cite{viswanathan2008levy}.

The theoretical foundation of this ``Cauchy optimality'' narrative traces back to Viswanathan et al., who argued that under continuous detection—in which a searcher perceives targets at all points during its trajectory—this strategy maximizes encounter rates with sparse, randomly distributed, replenishable targets (that reappear once visited) \cite{viswanathan1999optimizing}. This seminal work from 1999 catalyzed an extensive body of empirical and theoretical literature, much of which relied heavily on numerical simulations or scaling approximations \cite{zaburdaev2015levy}. For example, early simulations under the same model reported an efficiency maximum near $\mu=2$ in two and three dimensions~\cite{viswanathan2001levy}.
In parallel, the original optimality claim in the non-destructive,
revisitable-target setting has generated considerable debate concerning both
its biological assumptions and its theoretical basis. In particular, Levernier et
al.~questioned whether inverse-square L\'evy walks are uniquely optimal in
dimensions \(d\geq2\) \cite{levernier2020inverse}. Their conclusions were
contested in a subsequent Comment and defended in a Reply
\cite{buldyrev2021comment,levernier2021reply}. More broadly, this debate sharpens a fundamental question: under what
structural conditions is inverse-square L\'evy motion a superior search or
foraging strategy?

A significant rigorous advance in $d=2$ was provided by Guinard and Korman \cite{GuinardKorman2021}, who reframed the LFFH within the context of intermittent search \cite{benichou2011intermittent,kramer2001behavioral,Lomholt2008}, in which movement alternates between (i) fast relocation phases, during which detection is weak or effectively absent, and (ii) slower sensing phases (or pauses), during which detection is possible.
They proved that in finite two-dimensional domains, the intermittent Cauchy walk is nearly optimal for finding sparse targets whose size may vary over orders of magnitude \cite{GuinardKorman2021}. Crucially, their result is not that $\mu=2$ is best for a particular target size; rather, $\mu=2$ is uniquely robust across target scales, whereas other $\mu\neq 2$ strategies can be efficient for either large targets or small ones, but not for both. Intuitively, diffusive walks can be efficient for small targets, but they reach distant regions too slowly for their detection time to decrease at a rate that scales with target size. Conversely, ballistic walks spend excessive time in the relocation phase, making them inefficient when the target is small. This provides a compelling evolutionary rationale: if detection is weak and target size is unpredictable, movement robustness becomes a decisive fitness advantage.

Here, we extend this intermittent search framework to three dimensions, the primary spatial setting of the diverse biological and engineered systems introduced above. We show that moving from two to three dimensions fundamentally changes the nature of the search problem: efficiency becomes sensitive to target geometry, so that different target shapes favor strategies with distinct exploration–exploitation balances.

Against this increased geometric complexity, our analysis proves that the intermittent Cauchy strategy acts as a uniquely robust, near-optimal solution across a wide range of target geometries. For convex and approximately convex targets, we prove that it achieves nearly optimal, scale-invariant detection without tuning $\mu$ to the target. In particular, for a convex target of surface area $\Delta$ embedded in a domain of volume $n$, we show that the optimal expected detection time scales as $n/\Delta$, and that intermittent Cauchy search achieves this scaling up to lower-order (polylogarithmic) factors. In contrast, all other L\'evy strategies suffer significant efficiency losses as they encounter different geometric regimes: ballistic walks $(\mu < 2)$ are penalized by targets with a large surface-to-volume ratio (e.g., small spheres as well as flat or elongated targets of any size), while diffusive walks $(\mu > 2)$ degrade when encountering large spherical or disc-like targets (see Fig.~\ref{fig:summary}).

Our analysis leads to testable predictions: generalist intermittent 3D foragers encountering prey that varies widely in size and morphology should be favored near $\mu\approx 2$, while specialists targeting particular geometries may benefit from systematic deviations.
Together, our results provide a geometry-sensitive framework for understanding how exploration–exploitation tradeoffs are shaped by environmental structure, and when robust search strategies, such as the Cauchy walk, should emerge in complex systems.

\begin{figure}[!t]
    \centering
    \includegraphics[width=0.7\linewidth]{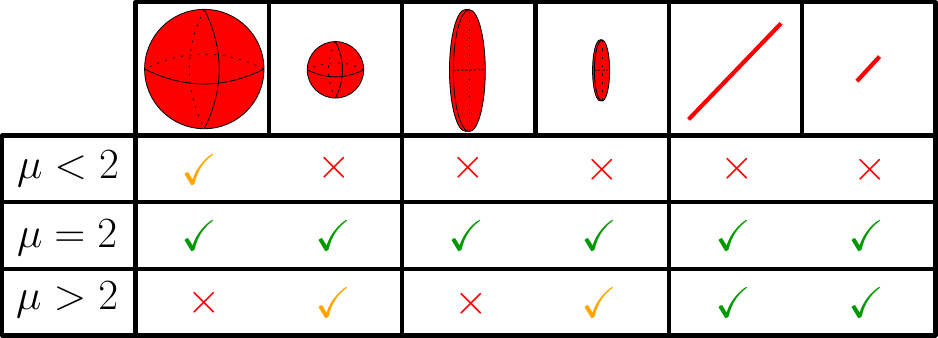}
    \caption{{\bf Search efficiency across L\'evy regimes and target geometries.} Summary of the performance of ballistic ($\mu < 2$), Cauchy ($\mu = 2$), and diffusive ($\mu > 2$) intermittent L\'evy walks across three fundamental target geometries—balls, discs, and lines—at both small and large scales. Performance is measured relative to the mean detection time of the optimal search strategy tuned to the specific geometry.
    In our rigorous proofs, efficiency ($\checkmark$) indicates that the performance ratio is tight up to lower-order terms, and inefficiency ($\times$) corresponds to a polynomial gap in system parameters. Green and Red denote efficiency and inefficiency, respectively, as established by rigorous mathematical proofs (and corroborated by simulations); Yellow indicates efficiency supported solely by numerical simulations.}
    \label{fig:summary}
\end{figure}

\section*{Results}\label{sec:results}

We consider a searcher attempting to locate a single (not necessarily connected) target within a finite three-dimensional domain, modeled as a cubic torus \[\torus = \left[-\frac{\sqrt[3]{n}}{{2}}, \frac{\sqrt[3]{n}}{{2}}\right]^3,\] of volume~$n$, with periodic boundary conditions.

The searcher starts from a random point of the torus and moves at constant speed according to a random walk strategy~$X$, alternating between constant-duration pauses and ballistic steps whose lengths~$\ell$ follow a prescribed distribution~$p$ in uniformly random directions. For $\mu \in (1,3]$, the distribution of the (truncated) L\'evy walk~$X^\mu$ (see details in Methods Section)
is:
\[
p(\ell) \sim \frac{1}{\ell^{\mu}}, \quad \ell \le \frac{\sqrt[3]{n}}{2}.
\]
The searcher possesses a detection radius $d \geq 1$; in detection mode, it successfully detects a given set of points if it is located within distance $d$ of them. Hereafter, we define this entire detectable region—the collection of points within distance $d$ of the desired set—as the target, denoted $S$ (see Fig.~\ref{fig:1patch} and Methods Section).
 We focus on  {\em intermittent search} \cite{benichou2011intermittent}, where detection is only possible during the short pauses between steps, and not while moving during a step.
  Thus, detection occurs if, at the end of a ballistic step, the searcher is located inside~$S$ (see Fig.~\ref{fig:1patch}).

The {\em detection time} of process~$X$ with respect to~$S$, denoted~$t_{\text{detect}}^{X}(S)$, is the expected time until~$X$ first detects~$S$, averaged over the randomness of both the process~$X$ and the initial position of the searcher.

To evaluate performance with respect to a target $S$, we employ a competitive analysis approach to define the search overhead.
Specifically,
for a given target $S$, let $X^\star$ denote the optimal strategy, that is, the one that minimizes  detection time of $S$. The {\em overhead} of a strategy $X$ with respect to $S$ is defined as the ratio between its detection time $t_{\text{detect}}^{X}(S)$ and the optimal detection time $t_{\text{detect}}^{X^\star}(S)$. Consequently, efficiency is considered degraded (or increased) if this overhead increases (or decreases). In our mathematical framework, a search strategy is classified as efficient with respect to a target if its overhead is at most poly-logarithmic in  the system size $n$. Conversely, a strategy is called inefficient if the overhead scales polynomially with $n$.

\begin{figure*}[!t]
\centering

\begin{minipage}[c]{0.46\textwidth}
    \centering
    \subcaptionbox{Patch example\label{fig:1patch}}
    {\includegraphics[width=\linewidth]{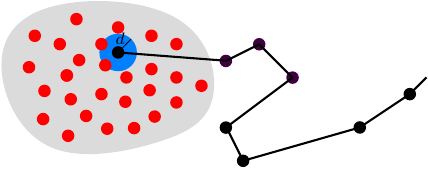}}

    \vspace{5pt}

    \subcaptionbox{Example for $\Delta_P$ and $\Delta_B$\label{fig:1delta}}
    {\includegraphics[width=\linewidth]{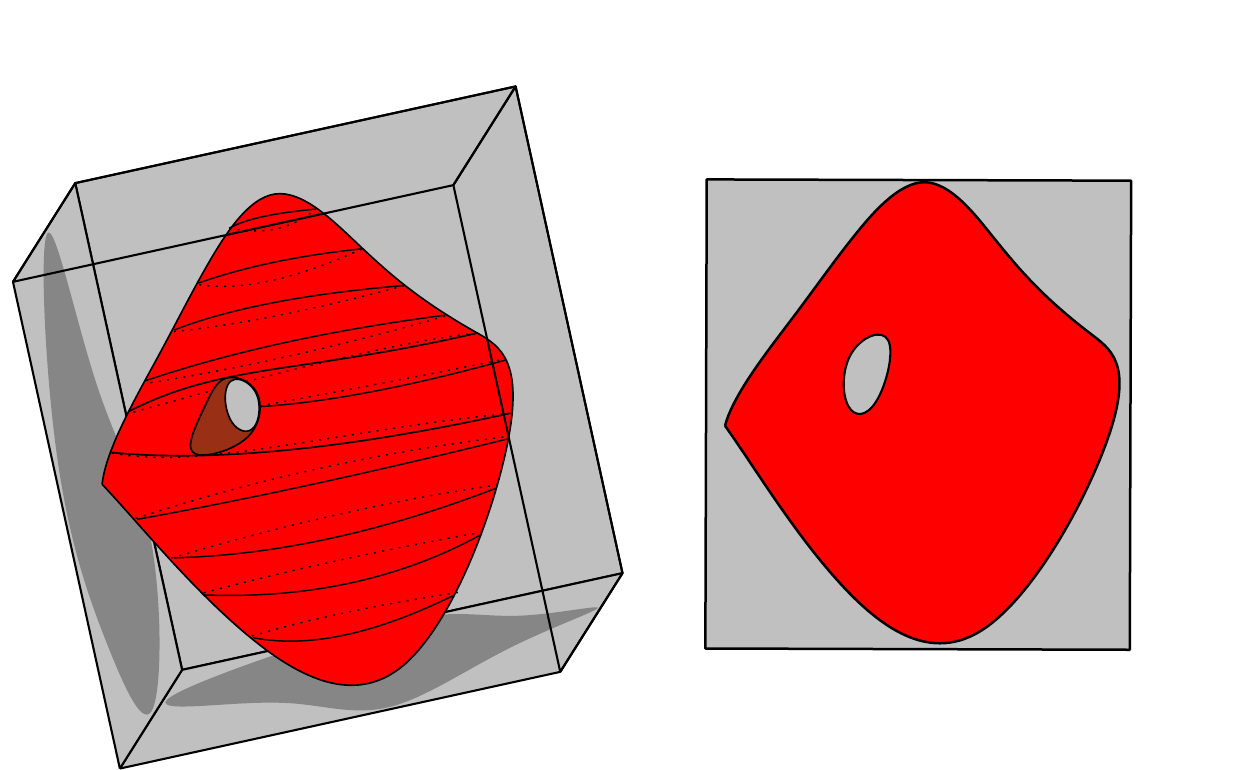}}
\end{minipage}
\hfill
\begin{minipage}[c]{0.53\textwidth}
    \centering
    \subcaptionbox{Tiling\label{fig:1tiling}}
    {\includegraphics[
        width=\linewidth,
        height=0.5\textheight,
        keepaspectratio
    ]{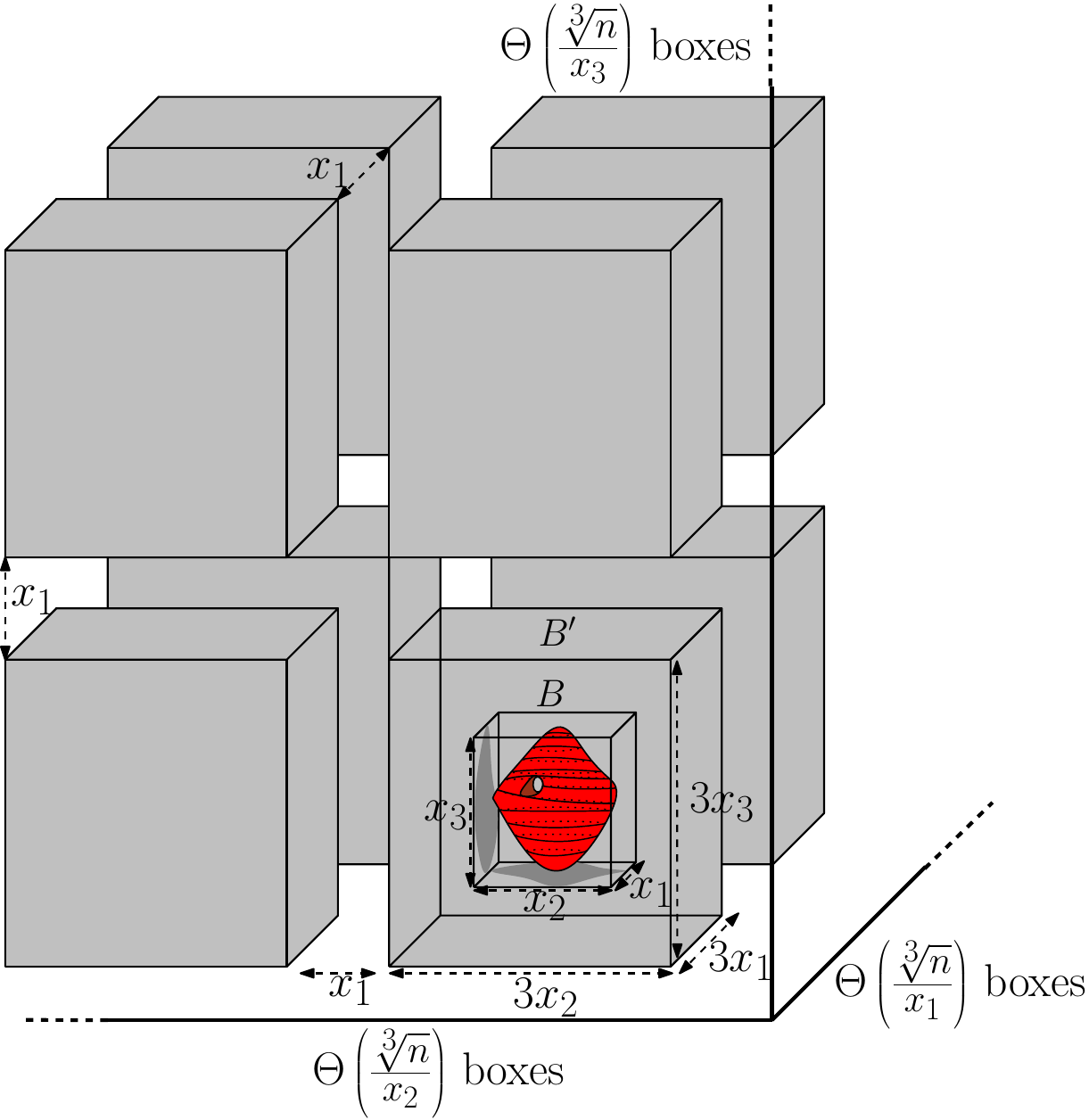}}
\end{minipage}
\caption{{\bf Intermittent search, Bounding box, and Universal lower bound construction.} Fig.~\ref{fig:1patch} provides a visualization of an agent performing an intermittent L\'evy walk (for simplicity of presentation, it is two-dimensional). The points to be detected (e.g., food patches) are represented by red circles, and the surrounding grey area indicates the detectable region, which we refer to as the target $S$; similarly, the blue area around the searcher represents its own detection radius. In this specific instance, the agent successfully detects the prey via this range. Notably, because the search is intermittent, the agent cannot detect targets while in transit; detection occurs only upon stopping at the locations indicated by the black dots. Fig.~\ref{fig:1delta} shows a target (red shape on the left) together with its minimum surface area bounding box \(B = \texttt{Box}(S)\) (grey box on the left), which is not necessarily axis-aligned. The panel on the right illustrates the projection of the shape onto the largest face of the box; In this example, $\Delta_P$ equals the red area, while $\Delta_B$ equals the red area plus the grey area. Figure~\ref{fig:1tiling} depicts the tiling argument used to derive the universal lower bound in~\eqref{universal-lower}. For clarity, the boxes are shown as axis-aligned; in the general case, however, they need not be aligned with the coordinate axes. }
\label{fig1}
\end{figure*}

\begin{figure*}[!t]
\centering

\begin{subfigure}[t]{0.46\textwidth}
  \centering
  \includegraphics[width=\linewidth]{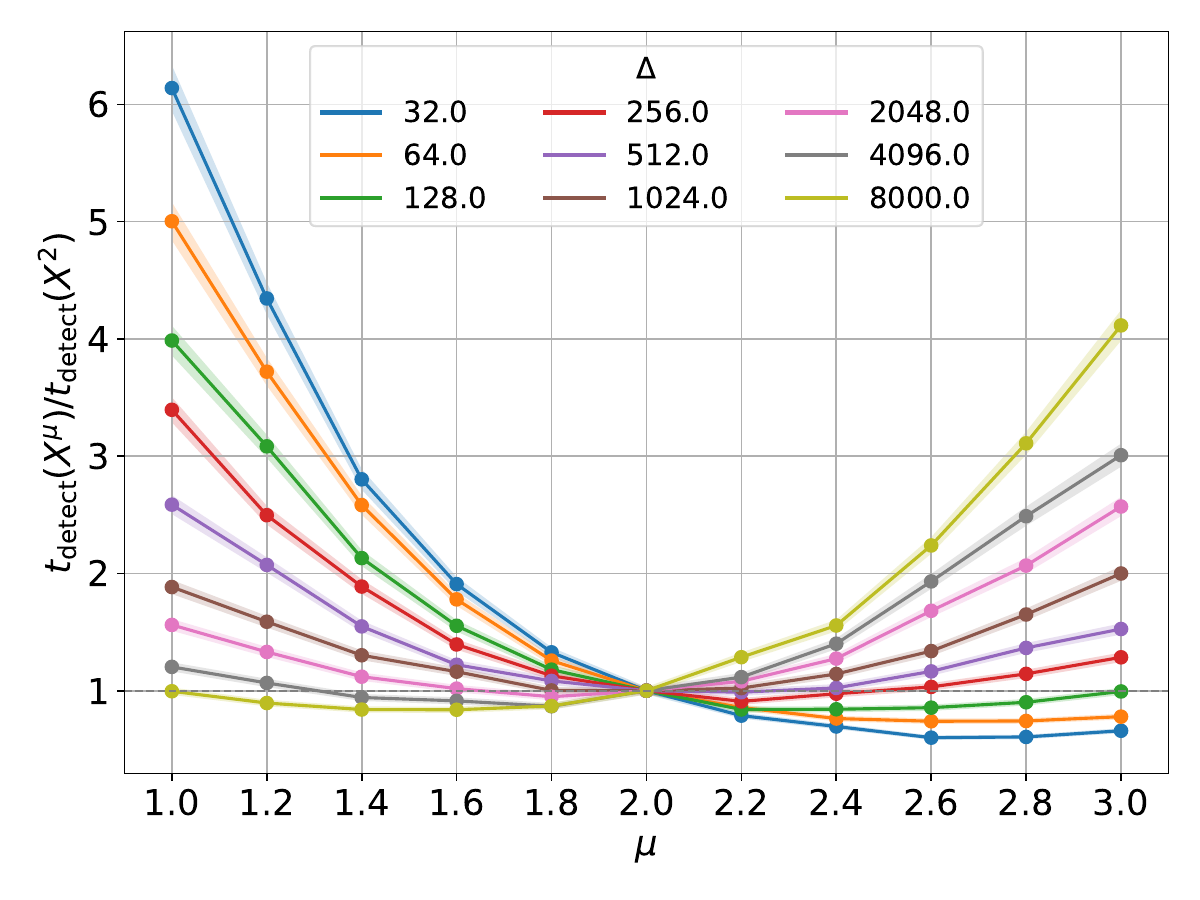}
  \caption{Relative detection time (Ball)}
  \label{fig:2relative_detection_time_ball}
\end{subfigure}
\hfill
\begin{subfigure}[t]{0.46\linewidth}
    \centering
    \includegraphics[width=\linewidth]{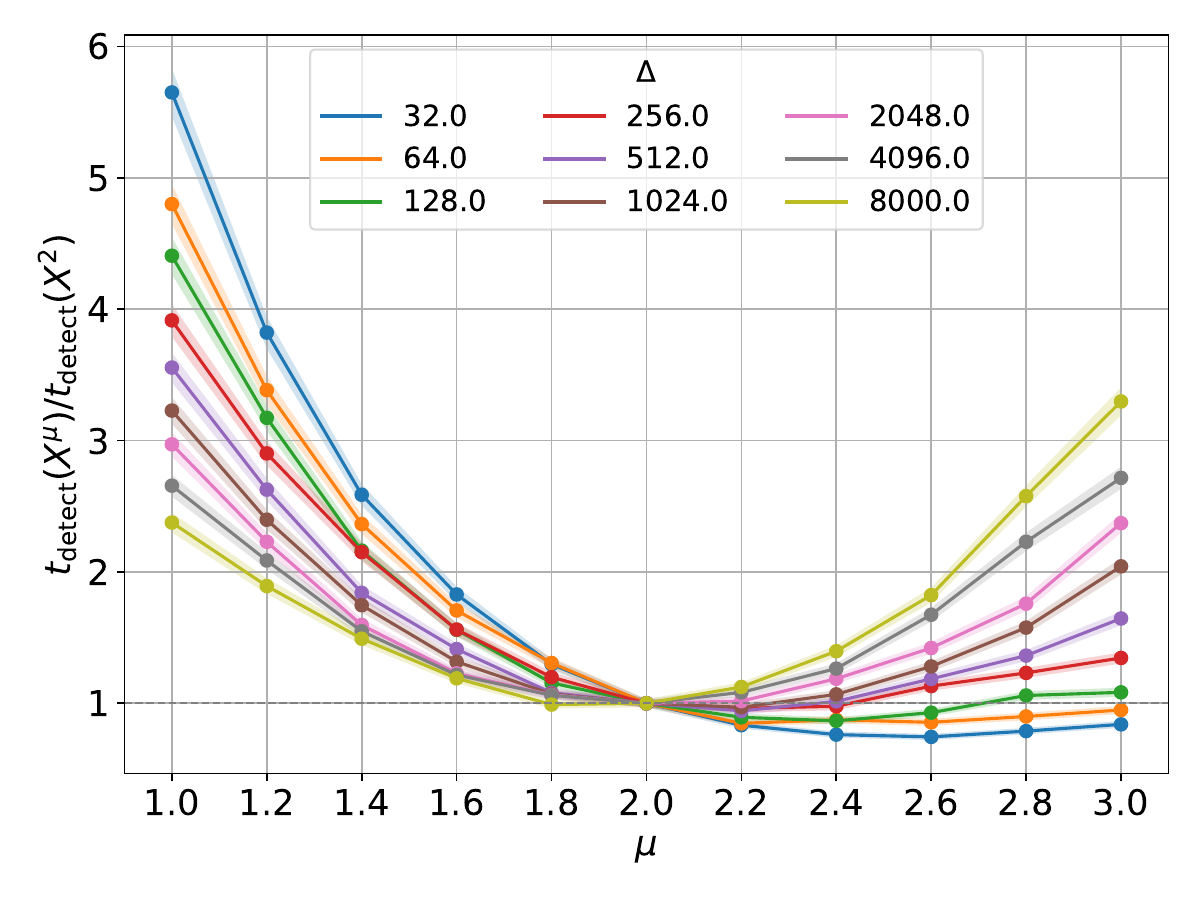}
    \caption{Relative detection time (Disk)}
    \label{fig:2relative_detection_time_disk}
\end{subfigure}

\vspace{6pt}
\begin{subfigure}[t]{0.46\textwidth}
  \centering
  \includegraphics[width=\linewidth]{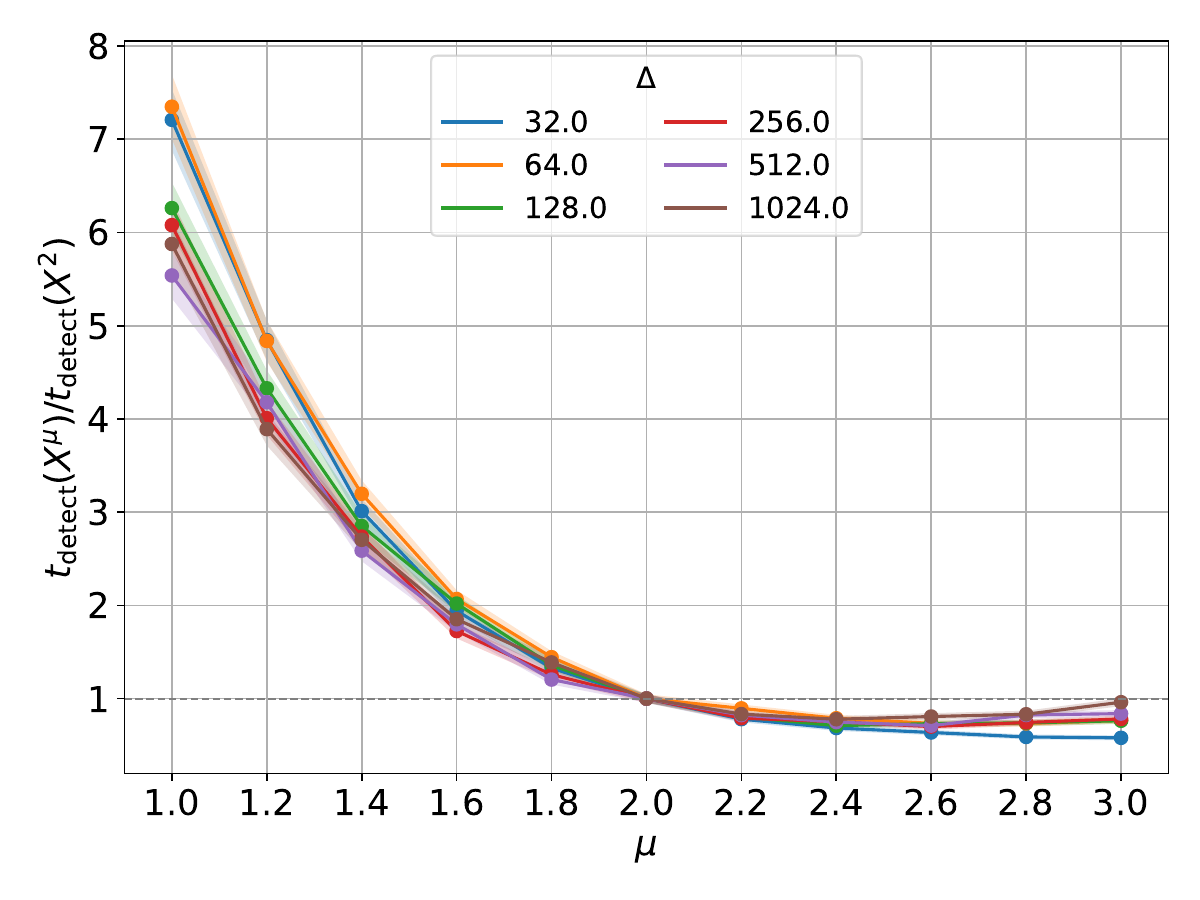}
  \caption{Relative detection time (Line)}
  \label{fig:2relative_detection_time_line}
\end{subfigure}
\hfill
\begin{subfigure}[t]{0.46\textwidth}
  \centering
  \includegraphics[width=\linewidth]{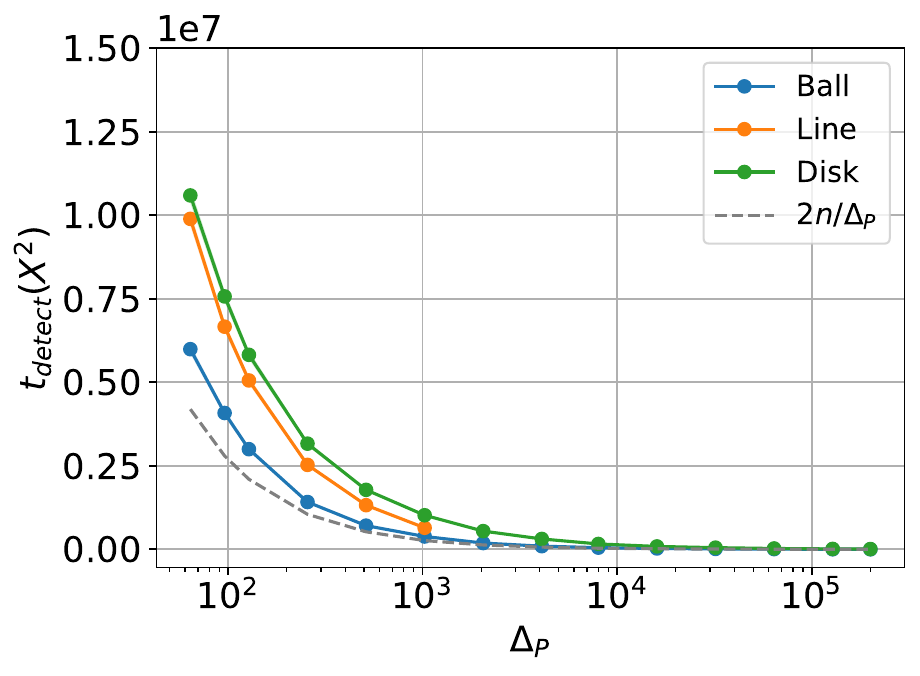}
  \caption{Cauchy searching different shapes}
  \label{fig:1upper_bound}
\end{subfigure}

\caption{{\bf Comparing the detection times of different L\'evy walks.} All simulations correspond to the 3D torus $\torus$ with volume $n=512^3$.
Fig.~\ref{fig:2relative_detection_time_ball} displays detection-time ratios normalized by the detection time of the Cauchy walk ($\mu = 2$) for ball-shaped targets with surface area $\Delta$ (color-coded), emphasizing the relative performance of different exponents. Figs. \ref{fig:2relative_detection_time_disk} and \ref{fig:2relative_detection_time_line} display analogous results for disk- and line-shaped targets, using the same simulation setup. Fig.~\ref{fig:1upper_bound} presents a comparison of the detection times of the Cauchy walk for \textit{Line}, \textit{Disk}, and \textit{Ball} targets, all sharing the same projected surface area $\Delta_P$, alongside the corresponding universal lower bound (dashed line). Takeaway: $\mu = 2$ stays near baseline across geometries; $\mu \neq 2$ exhibits geometry-dependent
slowdowns. Shaded bands represent $\pm \text{ SEM}$ ($N >2000$ independent runs); where bands are not visible, the band is smaller than the line width.}

\end{figure*}

\paragraph*{Geometric Framework}
Our analysis suggests that the influence of target shape on search efficiency is substantially more complex in three dimensions than in two. While diameter serves as a natural size descriptor in 2D settings \cite{GuinardKorman2021}, we were unable to identify a single geometric parameter that reliably captures optimal detection across the diverse range of 3D target shapes. Nevertheless, surface area, and related geometric variants, emerge as effective predictors of optimal search time across a broad class of shapes, particularly those that are not highly folded.

Furthermore, different L\'evy exponents are sensitive to distinct geometric features. Intuitively, a L\'evy walk with $\mu \approx 1$ samples the torus almost uniformly in each step; this makes the target's volume the primary determinant of detection time. Conversely, for $\mu \approx 3$, the walk resembles Brownian motion. Because steps are small and local, the walker cannot `jump' into the target's interior but must instead intercept its boundary. Hence, detection time is fundamentally driven by boundary crossing, making surface area a natural candidate for governing detection performance.
As we will demonstrate later, this intuition is insufficient to encompass the full complexity of geometric sensitivity, e.g., the observation that Brownian-like walkers detect elongated targets significantly faster than spherical ones of equivalent surface area.

Combining mathematical analysis with simulations, we demonstrate that geometric sensitivity shifts continuously with the L\'evy exponent: volume dominates detection time for $\mu \approx 1$, gradually yielding to surface area, which becomes the sole governing factor at the critical point $\mu = 2$. Beyond this point, for $\mu > 2$, performance becomes sensitive to both surface area and elongation.

The following definitions will be used to formally relate target morphology to search efficiency. Given a target $S$, let $\texttt{Box}(S)$ denote the smallest 3D box containing $S$.
\begin{itemize}
    \item \textbf{Largest face area ($\Delta_B$):} This is defined as the surface area of the largest face of $\texttt{Box}(S)$ (Fig.~\ref{fig:1delta}, left side).
    \item \textbf{Projected surface area ($\Delta_P$):} This is defined by considering the orthogonal projections of $S$ onto each of the six faces of $\texttt{Box}(S)$, and taking the largest surface area among these six projections (Fig.~\ref{fig:1delta}, right side).
\end{itemize}

Note that $\Delta_B\geq \Delta_P$ for every shape $S$. A shape is called {\em approximately convex} if $\Delta_P \geq  \Delta_B/36$. The  choice of the factor $1/36$ allows this family to contain all convex shapes (e.g., spheres, ellipsoids, discs, cylinders), possibly with small surface irregularities such as holes or bumps, see proof in SI, Section 5.

Let $x_1, x_2, x_3$ denote the side dimensions of the bounding box $B = \texttt{Box}(S)$, where, without loss of generality, $x_1 \leq x_2 \leq x_3$. We refer to these quantities  as the target's {\em thickness}, {\em width}, and {\em length}, respectively.
Note that the target's length satisfies
$x_3 = \Delta_B^\delta$, for some $\delta \in \left[\tfrac{1}{2}, 1\right]$, which we refer to as the {\em elongation} of $S$. For example, $\delta \approx \frac{1}{2}$ for a ball-shaped or disc-shaped target, while $\delta \approx 1$ for a long line-shaped target  (see details in SI, Section 3).

\paragraph*{Universal lower bound}\label{par:univ-lb}

  We proved mathematically (SI, Section 2) that, regardless of the search strategy, the detection time of
 $S$ is always at least
\begin{align}\label{universal-lower}
\Omega\left(\frac{n}{\Delta_B}\right).
\end{align}
We refer to this bound as a {\em universal lower bound}, since it holds for any search strategy, and thus applies in particular to all L\'evy walks.

To provide geometric intuition for the derivation of ~\eqref{universal-lower}, note that the volume of the bounding box $B = \texttt{Box}(S)$ is given by $V = x_1 x_2 x_3$, and the largest face area is $\Delta_B=x_2 x_3$. Analogously, define an enlarged box $B'$ with sides $3x_1, 3x_2, 3x_3$ and volume $V'=\Theta(V)$. Consider $\Omega(n/V')=\Omega(n/V)$ different placements of $B'$ in the torus, so that every two displacements are of distance at least $x_1$ apart (see Fig.~\ref{fig:1tiling}). Before the search starts, instead of choosing the starting point of the searcher at random, we fix it at the origin and choose the position of $S$ randomly in the torus. With constant probability,  $S$ is fully contained in one box $B'$, and all such boxes have equal probability to contain it.
 To find $S$, the searcher needs to visit at least half of these box-displacements on average, and since each new visit takes at least $x_1$ time (the minimal distance between displacements), the total expected time is at least $\Omega(n x_1/V)=\Omega(n/x_2x_3)=\Omega(n/\Delta_B)$. This establishes ~\eqref{universal-lower}.

Next, we seek to identify whether any intermittent L\'evy walk strategy can detect
$S$ in time approaching the universal lower bound. Our analysis of the interplay between target shape and the exponent $\mu$ reveals a phase transition at $\mu=2$.

\paragraph*{Ballistic regime ($\mu<2$) is penalized by targets with a large surface to volume ratio}

Consider a L\'evy walk with exponent $\mu \in (1,2)$
and a target $S$ of volume~$V$.
Let $\varepsilon=2-\mu$ denote the `distance' from the Cauchy exponent.
We first proved (SI, Section 3) a general lower bound of $\Omega(n/V)$ for the expected number of steps before finding the target, which holds for any intermittent L\'evy walk. Applying a variant of Wald's identity (SI, Claim 1.1), we can transform this lower bound into one for the detection time by scaling it by the average step length of the L\'evy process. In the regime $\mu \in (1,2)$, the mean step length is $\Theta(n^{\varepsilon/3})$ (SI, Fact~7.1), leading to the following lower bound on the detection time:
\begin{align}\label{eq:lb-smaller}
t_{\text{detect}}^{X^\mu}(S) = \Omega\left(\frac{n^{1+\varepsilon/3}}{V}\right).
\end{align}
Consider a target $S$ whose volume $V$ is proportional to its largest face area $\Delta_B$. For example, a ball of constant diameter, a disc of an arbitrary diameter, or a line of arbitrary length (see Methods Section). For such a target,~\eqref{eq:lb-smaller} implies that detection time is a factor of $\Omega(n^{\varepsilon/3})$ slower than the universal lower bound. This gap widens as $\mu$ decreases further from~$2$, with the deterioration driven by the polynomial degree $\varepsilon/3$.
Due to~\eqref{eq:scale-inv-surface} (discussed later), detection times for such targets are slower than the ones of the Cauchy walk by the same polynomial factor, up to a lower order term. Since the optimal strategy finds $S$ at least as fast as Cauchy, the overhead (defined w.r.t the optimal strategy) is polynomial in $n$.

These theoretical results are consistent with our numerical simulations.
Fig.~\ref{fig:2relative_detection_time_ball} shows the detection time as a function of the L\'evy exponent $\mu$ for ball-shaped targets of varying surface areas $\Delta$, normalized by the detection time of the Cauchy walk ($\mu = 2$) for each corresponding $\Delta$. In the regime $\mu<2$, small ball targets are found more slowly by such L\'evy processes than by the Cauchy searcher. Figs.~\ref{fig:2relative_detection_time_disk} and~\ref{fig:2relative_detection_time_line} further show that when $\mu < 2$, detection time of disc-targets (Fig.~\ref{fig:2relative_detection_time_disk}) as well as line-targets (Fig.~\ref{fig:2relative_detection_time_line}), increases as $\mu$ decreases, regardless of the target sizes.
Finally, Fig.~\ref{fig:2relative_detection_time_ball} suggests that large-scale balls are detected efficiently in the ballistic regime ($\mu < 2$). This performance may be attributed to the diminishing surface area to volume ratio as the radius of the ball increases.

 These results indicate that increasingly ballistic strategies become less effective at detecting targets with a large surface area to volume ratio, regardless of their size,
 and confirm that the transition from ballistic to diffusive search at $\mu = 2$ is critical for efficient detection of such geometries.

\paragraph*{Localized exploitation ($\mu>2$) is penalized for large spherical and large disc-like geometries}

Next, we examined L\'evy walks in the diffusive regime, $\mu \in (2,3]$. We prove that the overhead for spherical or disc-like targets is at least a polynomial factor $\Omega(\Delta^{\frac{\mu-2}{2}})$, where the penalty exponent increases as $\mu$ deviates from 2.

Specifically, we proved (SI, Corollary 3.7) that detection time of a L\'evy walk $X^\mu$ with $\mu\in(2,3]$ with respect to a ball (or a disc) target $S$ of surface area $\Delta$ obeys the following asymptotic lower bounds:
\begin{align*}
t_{detect}^{X^\mu}(S)=
\begin{cases}\Omega\left(n\Delta^{\frac{\mu-2}{2} - 1}\right) \text{ if } \mu \in (2,3) \\
\Omega\left(\frac{n}{\Delta^{1/2}\log \Delta}\right) \text{ if } \mu = 3\end{cases}.
\end{align*}
These lower bounds are based on a tiling argument analogous to the one used in the proof of the universal lower bound (see Fig.~\ref{fig:1tiling}), combined with a more nuanced lower bound on the time needed to move between boxes. Specifically, in the diffusive regime ($\mu > 2$), traversing a distance $\ell$ takes polynomial time in $\ell$, where the exponent of the polynomial is $\mu - 1$ (SI, Claim 3.5).

Observe that in this regime of $\mu$, detection time for a ball or a disc target with surface area $\Delta$ is larger than the universal lower bound by a polynomial factor in $\Delta$, specifically $\Delta^{\frac{\mu-2}{2}}$.
For large targets, such that $\Delta$ is polynomial in $n$, this factor is also polynomial in $n$.
Due to~\eqref{eq:scale-inv-surface} (discussed later), the detection time of the L\'evy process is also slower than the Cauchy walk by the same polynomial factor, up to the poly-logarithmic factor in~$n$ associated with the Cauchy walk.
This implies a polynomial in $n$  overhead with respect to such targets.

These theoretical findings are corroborated by our numerical simulations, which demonstrate that for spherical (Fig.~\ref{fig:2relative_detection_time_ball}) and disc-like (Fig.~\ref{fig:2relative_detection_time_disk}) targets with sufficiently large surface area, the search efficiency of diffusive strategies ($\mu > 2$) relative to the Cauchy strategy decreases as either the exponent $\mu$ or the target's surface area increases.

These results indicate that overly localized strategies suffer from reduced efficiency with respect to large spherical and disc-like targets, reinforcing the conclusion that the transition from diffusive to ballistic search at $\mu = 2$ is critical for the efficient detection of such targets.

\paragraph*{Scale-invariance optimality of the Cauchy Walk}\label{par: scale-inv}

Intermittency renders long relocations ``blind'', meaning detection relies entirely on where the step stops. To have any chance of hitting the target, an intermittent searcher must first select a winning direction; the probability of choosing this ``good angle'' is dictated strictly by the target's 2D silhouette rather than its total surface area, as a blind trajectory simply bypasses internal folds. Consequently, the success of a relocation is controlled by the target's cross-sectional exposure along the approach vector. This highlights the projected surface area $\Delta_P$ parameter, which acts as a proxy for the maximal target area that can be ``observed'' across all possible viewing directions.

We proved (SI, Theorem 4.4) that for any (not necessarily connected) target $S$ with projected surface area $\Delta_P$, the detection time of the Cauchy walk $X^{\text{cauchy}}$ satisfies:
\begin{align}\label{eq:scale-inv-surface}
    t_{\text{detect}}^{X^{\text{cauchy}}}(S) = \bigO{\frac{n \log^3 n}{\Delta_P}}.
\end{align}

Informally, the proof of ~\eqref{eq:scale-inv-surface} (see SI, Section 4) uses the projected surface area $\Delta_P(S)$ to identify a subset of $S$ whose properties render its hitting time more analytically tractable. More specifically, leveraging the geometrical properties of the projection $S$ onto the bounding box, the proof constructs a subset $S^\star\subseteq S$ which consists of $\Theta(\Delta_P)$ balls of radius 1, with the property that within any given ball of radius $r>1$, there are at most $\bigO{r^2}$ such balls. These properties are particularly helpful for upper-bounding the expected number of returns to $S^\star$, a key step in establishing an upper bound on the hitting time of $S^\star$. Since $S^\star\subseteq S$, this upper bound also serves as an upper bound on the hitting time of $S$.

Recall that the largest face area $\Delta_B$ and the projected surface area $\Delta_P$ of an approximately convex target are proportional. Combined with~\eqref{universal-lower} and~\eqref{eq:scale-inv-surface}, this implies that a Cauchy walk detects such a target in optimal time, up to a poly-logarithmic factor in~$n$.
In particular, for a convex target with surface area $\Delta$, the detection time is at most the optimal bound of $n/\Delta$, times a poly-logarithmic factor in~$n$.

We further showed (SI, Section 5.C) that, unlike the convex case, for arbitrary  shapes, the projected surface area $\Delta_P$ in~\eqref{eq:scale-inv-surface} cannot be replaced by either the target’s surface area $\Delta$ or the largest face surface area  $\Delta_B$.

Simulation results support our asymptotic mathematical findings.
Fig.~\ref{fig:1upper_bound} depicts the detection times of the Cauchy walk on the torus for both lines, discs, and balls, with equal projected surface area~$\Delta_P$ (see Methods Section). These plots demonstrate that Cauchy's detection time is proportional to the universal lower bound, suggesting that the poly-logarithmic factor predicted by~\eqref{eq:scale-inv-surface} is negligible in practice.

Figs.~\ref{fig:2relative_detection_time_ball},~\ref{fig:2relative_detection_time_disk} and~\ref{fig:2relative_detection_time_line} demonstrate the overall superiority of the Cauchy walk as a robust, geometry-agnostic search strategy. While searchers in the exploration-dominated ($\mu<2$) or exploitation-dominated ($\mu > 2$) regimes perform well for specific target scales and shapes, they suffer substantial efficiency losses when faced with mismatched target sizes and geometries. In contrast, the Cauchy walk ($\mu = 2$) identifies a critical point of geometric invariance, maintaining nearly optimal detection times across a spectrum of target geometries: from spheres (Fig.~\ref{fig:2relative_detection_time_ball}) and flat discs (Fig.~\ref{fig:2relative_detection_time_disk}) to highly elongated targets (Fig.~\ref{fig:2relative_detection_time_line}).

\begin{figure*}[!t]
\centering

\begin{subfigure}[t]{0.48\linewidth}
        \centering
        \includegraphics[width=\linewidth]{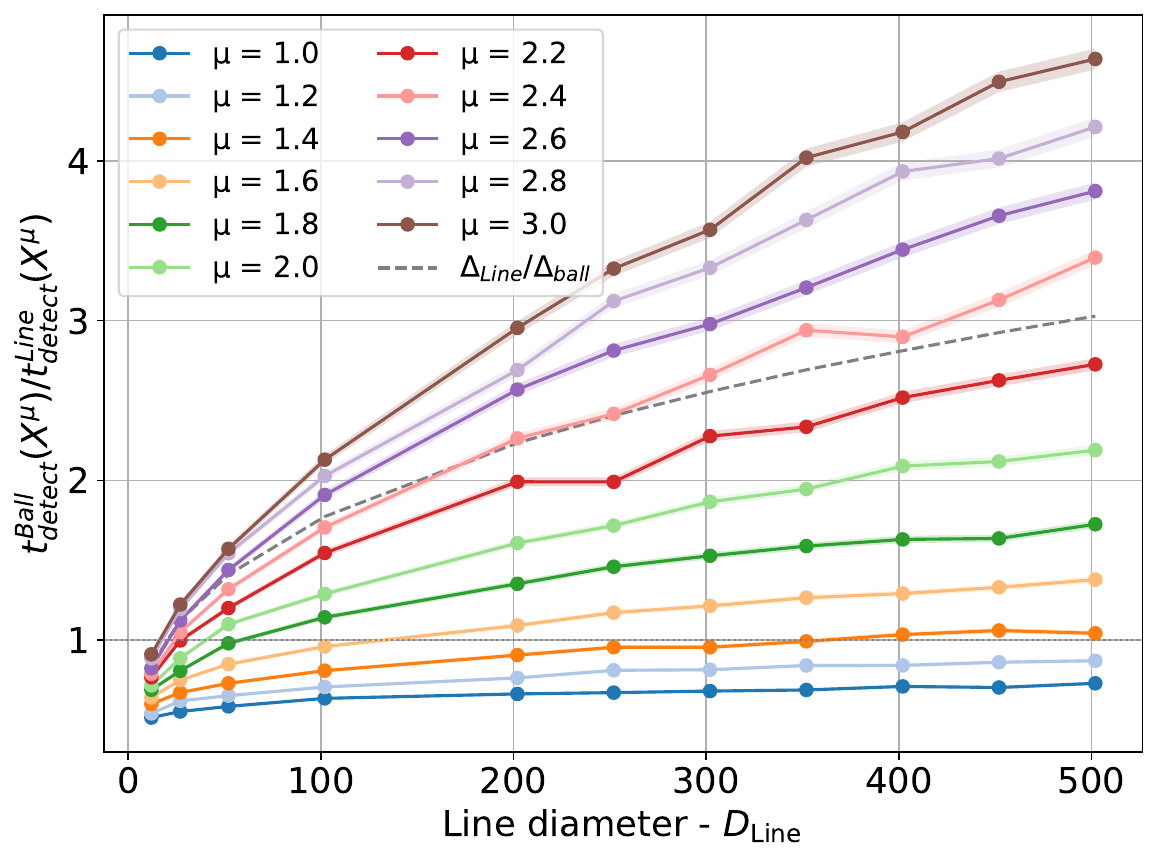}
        \caption{Detection time ratio with fixed volume}
        \label{fig:3detection_time_ratio_volume}
    \end{subfigure}
    \hfill
    \begin{subfigure}[t]{0.48\linewidth}
        \centering
        \includegraphics[width=\linewidth]{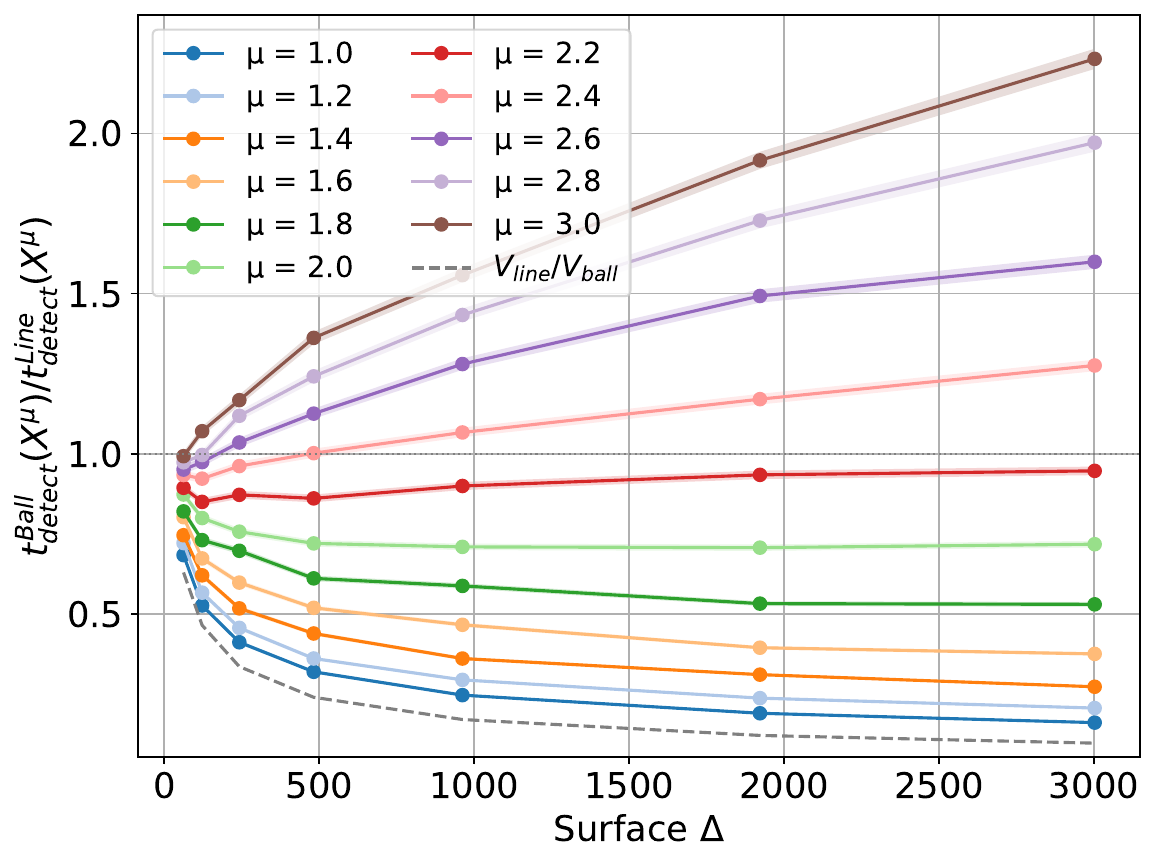}
        \caption{Detection time ratio with fixed surface area}
        \label{fig:3detection_time_ratio_surface}
    \end{subfigure}

\vspace{6pt}

\begin{subfigure}[b]{0.48\textwidth}
    \centering
    \includegraphics[width=\linewidth]{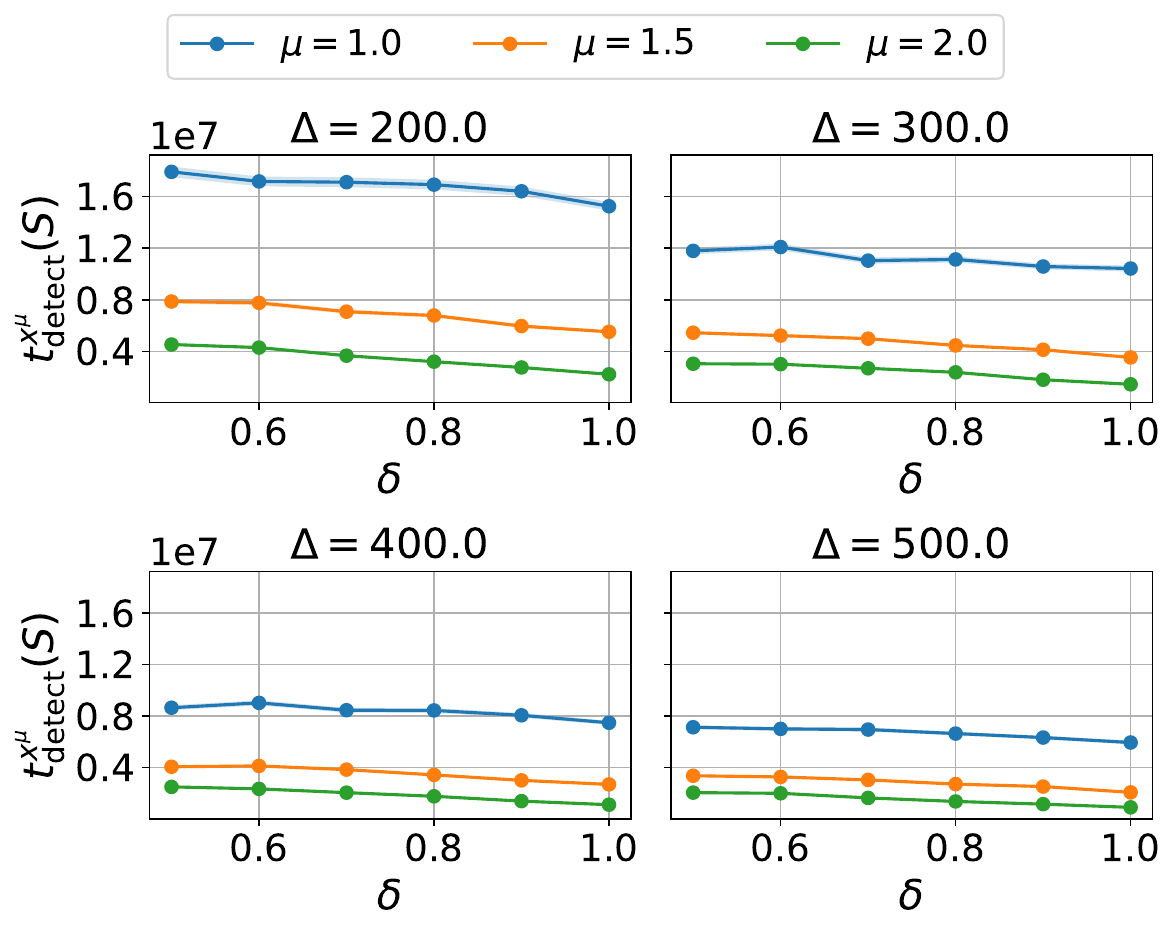}
    \caption{Varying $\delta$, $\mu \in \{1, 1.5, 2\}$}
    \label{fig:3varying_delta_small_mu}
\end{subfigure}\hfill
\begin{subfigure}[b]{0.48\textwidth}
    \centering
    \includegraphics[width=\linewidth]{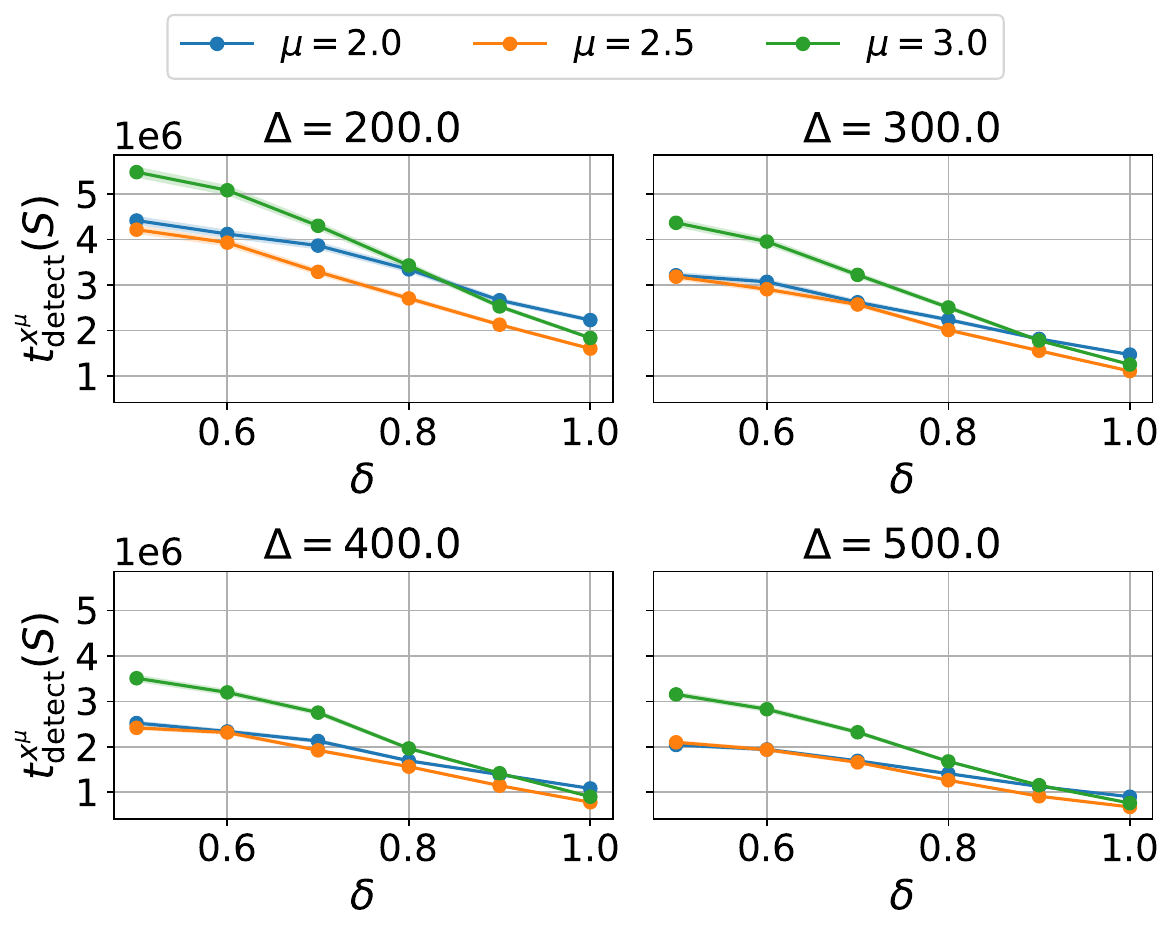}
    \caption{Varying $\delta$, $\mu \in \{2, 2.5, 3\}$}
    \label{fig:3varying_delta}
\end{subfigure}

\caption{{\bf Detection-time phase transition induced by varying $\mu$ and dependency on elongation.} All simulations are performed on a three-dimensional torus $\torus$ with volume $n = 512^3$. Fig.~\ref{fig:3detection_time_ratio_volume} shows the ratio of detection times for a ball and a line of equal volume, for several values of $\mu$. The $x$-axis corresponds to the line length $D_{\mathrm{Line}}$, while the diameter of the corresponding ball is chosen so that the two volumes are equal. The dashed grey line shows the line-to-ball surface-area ratio, $\Delta_{\mathrm{Line}}/\Delta_{\mathrm{Ball}}$, which provides a theoretically motivated scaling benchmark for the  detection-time ratio when $\mu=2$. Similarly, Fig.~\ref{fig:3detection_time_ratio_surface} presents the ratio of detection times for a ball and a line with the same surface area $\Delta$, shown on the $x$-axis. In this figure, the dashed grey line shows the line-to-ball volume ratio, $V_{\mathrm{Line}}/V_{\mathrm{Ball}}$, which provides a theoretically motivated scaling benchmark for the detection-time ratio in the limit $\mu\to1$.
Fig.~\ref{fig:3varying_delta_small_mu} presents detection times for L\'evy walks with $\mu \in \{1, 1.5, 2\}$, where the targets are rectangles with side lengths $\Delta^\delta$ and $\Delta^{1-\delta}$. Similarly, Fig. \ref{fig:3varying_delta} shows the result for the same set of simulations when $\mu \in \{2, 2.5, 3\}$. Takeaway: the governing geometric control shifts from volume to surface, and then to surface and elongation
as $\mu$ increases. Shaded bands represent $\pm \text{ SEM}$ ($N >2000$ independent runs); where bands are not visible, the band is smaller than the line width.
    \movedfigure{Results, after the discussion of elongation in the diffusive regime}}

\end{figure*}

\paragraph*{Different L\'evy walks are sensitive to different geometrical properties}
Integrating mathematical analysis with numerical simulations, we show that geometric sensitivity evolves continuously with the L\'evy exponent: while volume dictates detection time in the ballistic regime ($\mu \approx 1$), it gradually cedes dominance to surface area, which emerges as the primary governing parameter at the critical point $\mu = 2$. Conversely, in the diffusive regime ($\mu > 2$), detection performance becomes multifaceted, determined by a coupling of surface area and target elongation.

Consider first a L\'evy walker in the regime $\mu\in(1,2)$.~\eqref{eq:lb-smaller} provides a lower bound on its detection time concerning any target $S$ with volume $V$. Recall that this lower bound is composed of a lower bound on the number of steps $m$ before finding $S$, namely, $\Omega(\frac{n}{V})$, times the average step length, which is  $\Theta(n^{\varepsilon/3})$, where $\varepsilon=2-\mu$.
For $\mu\approx 1$, the process is close to sampling the torus uniformly in each step, implying that the number of steps $m$ scales like $n/V$. \changed{Together, these bounds yield a tight bound of $\Theta(\frac{n^{1+\varepsilon/3}}{V})$ for the detection time,  highlighting volume as the target's critical geometric parameter, governing detection time in scenarios where $\mu \approx 1$.}{This heuristic suggests that, for $\mu \approx 1$, detection times may be close to the scale $n^{1+\varepsilon/3}/V$ appearing in the lower bound in Eq. (2),  highlighting volume as the target’s critical geometric parameter, governing detection time in scenarios where $\mu \approx 1$.}

Fig.~\ref{fig:3detection_time_ratio_volume} illustrates the detection time ratio between a ball and a line of equivalent volume, for different exponents $\mu$. The relatively flat profiles observed for $\mu \leq 1.4$ suggest that, in this regime, detection is primarily volume-driven. Furthermore, Fig.~\ref{fig:3detection_time_ratio_surface} presents the ratio of detection times for a ball and a line with equal surface area $\Delta$. Here, the dashed line represents the volume ratio of the two targets; the proximity of the plots to this line at small $\mu$ values further supports volume as a dominant geometric parameter in this regime. Moreover,
Figs.~\ref{fig:3detection_time_ratio_volume} and~\ref{fig:3detection_time_ratio_surface} demonstrate a relatively smooth transition from volume to surface area, where the latter is represented by the dashed line in Fig.~\ref{fig:3detection_time_ratio_volume}, and the horizontal line $y=1$ in Fig.~\ref{fig:3detection_time_ratio_surface}.

For the Cauchy walk ($\mu=2)$,~\eqref{eq:scale-inv-surface} shows that the detection time of a convex target with surface area $\Delta$ scales like the universal lower bound $n/\Delta$, up to a poly-logarithmic factor in $n$. This highlights surface area as the critical geometric parameter governing optimal detection time, and, in particular, the most relevant geometric parameter for the case $\mu=2$. Accordingly, as predicted by combining~\eqref{universal-lower} and~\eqref{eq:scale-inv-surface}, Fig.~\ref{fig:1upper_bound} demonstrates that despite the substantial difference in volume between the corresponding line, disk, and ball shapes, their detection times by the Cauchy walk remain comparable due to their projected surface area similarity.

Next, we considered L\'evy walks with exponent $\mu \in (2,3]$, and showed that the detection time is strongly affected by both surface area and elongation. Indeed, these searchers are exceptionally efficient at locating highly elongated targets; however, their performance degrades rapidly as the target becomes wider, that is, as the second-largest dimension of the bounding box, increases.

Specifically,
we proved (SI, Theorem 3.6) that detection time of a target $S$ with elongation $\delta$ obeys the following asymptotic lower bounds:
\begin{align}\label{eq:lb-larger}
t_{\text{detect}}^{X^\mu}(S) =
\begin{cases}
 \Omega\!\left(n \Delta_B^{\varepsilon(1-\delta)-1}\right),
   & \text{if }  0 < \varepsilon < 1, \\[6pt]
 \Omega\!\left(\dfrac{n}{\Delta_B^\delta \log \Delta_B}\right),
   & \text{if } \varepsilon = 1,
\end{cases}
\end{align}
 with $\mu = 2 + \varepsilon$. In this regime of $\mu$, detection time for approximately convex targets with a largest face area of $\Delta_B$ (or equivalently, a  projected surface area) is slower than the Cauchy walk by a polynomial factor $\Delta_B^c$, where $c = \varepsilon(1-\delta)$, up to a poly-logarithmic factor in~$n$. The exponent $c$ scales positively with $\mu$ and inversely with $\delta$.

While the lower bounds defined in~\eqref{eq:lb-larger} are maximized for ball or disc  shapes ($\delta=1/2$), they become closer to the universal lower bound of $\Omega(n/\Delta_B)$ as  the elongation parameter $\delta$ approaches 1.
Consistent with this, Fig.~\ref{fig:3varying_delta} illustrates how   detection time of rectangles (see definition in Methods Section) increases as $\delta$ decreases (i.e., the target becomes wider), while maintaining the same surface area.

Next, we demonstrate that elongated targets are located with high efficiency by L\'evy searchers in the diffusive regime ($2 < \mu \leq 3$). We provide analytical evidence by analyzing a discretized version of the process (see Methods Section), proving that the detection time for a line target of length $L$ is $O((n/L) \log n)$ (see SI, Section 6). For the discrete grid, it is easy to adjust the argument behind the universal lower bound (\eqref{universal-lower}) to obtain an $\Omega(n/L)$ lower bound for the detection time of the line. Hence, such L\'evy walks incur only a logarithmic overhead. As $n \to \infty$, this discrete approximation can be expected to converge to the continuous domain setting with detection radius being 1, suggesting that these scaling results hold also for the continuous setting.

The proof is based on the following intuition. Assume the target is a line of length~$L$, oriented parallel to the $x$-axis, as illustrated in Fig.~\ref{fig:line_target_intuition_target}. Consider the L\'evy walk projected onto the $yz$-plane, as illustrated in Fig.~\ref{fig:line_target_intuition_projection}, and let $v$ be the point representing the projection of the target line onto this plane. This projected process behaves as a L\'evy walk with the same exponent on a 2D domain of area $n^{2/3}$; consequently, the expected time for the projected walk to hit $v$ is $\mathcal{T}:=\bigO{n^{2/3} \log n}$
(SI, Theorem 6.2).
At the moment of impact with $v$, the 3D walk is effectively mixed on the 3D torus
(SI, Lemma 6.9),
implying that all points along the line segment orthogonal to $v$ (which contains the target and has a total length of $n^{1/3}$) are roughly equally likely to be occupied by the walker. This yields a conditional detection probability of $\Theta(L/n^{1/3})$. Hence, by time $\mathcal{T}$, detection occurs with probability $p:=\Theta(L/n^{1/3})$. Summing these independent geometric trials yields a total expected detection time of $\mathcal{T}/p=\bigO{(n/L) \log n}$. Accordingly, our simulations (Fig.~\ref{fig:2relative_detection_time_line}) confirm that this accelerated detection of line targets is comparable to that of the Cauchy walk, and remains robust across this entire diffusive range.

The previous argument does not translate directly to a flattened target. Consider a square slab of constant thickness and surface area $\Delta$, where the two lateral dimensions are $\Theta(\sqrt{\Delta})$. If one dimension is aligned parallel to the $x$-axis, its projection onto the $yz$-plane is a segment of length $\Theta(\sqrt{\Delta})$. Compared to a linear target of length $\Delta$, the projection onto the $x$-axis reduces from $\Delta$ to $\sqrt{\Delta}$. For a parallel argument to hold, the hitting time of the projected target on the $yz$-plane would need to decrease from $\mathcal{T}$ to approximately $\mathcal{T}/\sqrt{\Delta}$ to compensate for this reduction in longitudinal length. However, for a simple random walk on the plane, the hitting time is bounded below by the time required to traverse a distance of $n^{1/3}$, which is $\Omega(n^{2/3})$, even for a large target. Consequently, a direct application of this argument would not reduce the time beyond $\Omega(n/\sqrt{\Delta})$. This demonstrates that while the argument effectively exploits a one-dimensional extent, it cannot directly recover the gain from an area split across two dimensions.

\begin{figure}[!t]
    \centering
    \begin{subfigure}{\linewidth}
        \centering
        \includegraphics[width=0.85\linewidth]{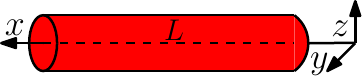}
        \caption{Line target}
        \label{fig:line_target_intuition_target}
    \end{subfigure}

    \vspace{4pt}

    \begin{subfigure}{\linewidth}
        \centering
        \includegraphics[width=0.68\linewidth]{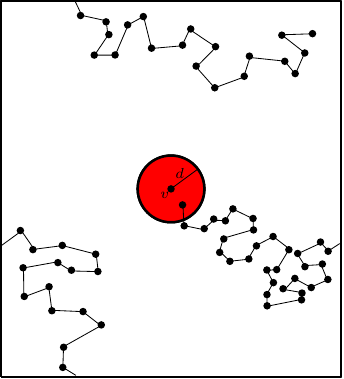}
        \caption{Projection onto the $yz$-plane}
        \label{fig:line_target_intuition_projection}
    \end{subfigure}

    \caption{{\bf Projection-and-mixing intuition for detecting a line target by a simple random walk.} Fig.~\ref{fig:line_target_intuition_target}: A line target of length~$L$, aligned with the $x$-axis, together with its detectable neighborhood. Fig.~\ref{fig:line_target_intuition_projection}: Under projection onto the $yz$-plane, the line collapses to the point $v$, with the red disk representing the detection radius~$d$ around it. Black dots denote relocation endpoints of the projected intermittent walk; detection is possible only at these endpoints; because of the projection, the projected step lengths vary. The projected simple random walk encounters the projected target after an expected $\mathcal{T}=\bigO{n^{2/3}\log n}$ steps. At such an encounter, mixing makes the $x$-coordinate of the walk approximately uniform along a cycle of length $n^{1/3}$, giving conditional detection probability $\Theta(L/n^{1/3})$. Thus, $\Theta(n^{1/3}/L)$ projected encounters are required on average, yielding a total expected detection time of $\bigO{(n/L)\log n}$.}
    \label{fig:line_target_intuition}
\end{figure}

Finally, consistent with the lower bounds in~\eqref{eq:lb-larger},
Figs.~\ref{fig:3detection_time_ratio_volume} and~\ref{fig:3detection_time_ratio_surface} demonstrate that long lines are detected significantly faster than balls
of equal volume (Fig.~\ref{fig:3detection_time_ratio_volume}) or surface area (Fig.~\ref{fig:3detection_time_ratio_surface}). This indicates that surface area, while dominant at $\mu=2$, is no longer a sufficient descriptor once $\mu$ enters the diffusive regime.

These results suggest that, in addition to the surface area, elongation plays a critical role for the detection time in the diffusive regime $\mu\in (2,3]$.

\paragraph*{Application of the framework as a predictive tool}

To demonstrate the predictive power of our theoretical framework in real-world three-dimensional motion, we analyzed trajectories from the Leukocyte Tracking Database (LTDB) \cite{pizzagalli2018leukocyte}, which contains manually curated 3D tracks of immune cells recorded under several biological conditions. We considered natural killer (NK) cells and neutrophils in popliteal lymph nodes following influenza vaccination \cite{farsakoglu2019influenza,pizzagalli2019characterization}, together with T cells in popliteal lymph nodes from a humanized-mouse model of HIV infection \cite{murooka2012hiv}.

Figure~\ref{fig:empirical_grounding} depicts the probability of
encountering (imaginary) ball, disk, and line targets matched by projected area
$\Delta_P$ before the track-specific empirical path-length budget is
exhausted (see Methods). For all three cell groups, detection probability increased with $\Delta_P$ and, at a fixed projected area, was generally higher for balls than for disks or lines. Our theoretical findings could be used to hypothesize that if these patterns indeed represent L\'evy walks, they should be linked to the ballistic regime ($\mu < 2$), which is consistent with this specific geometric ordering. Notably, this classification is independent of the detection mechanism and applies even in the absence of a target search. Although our methodology simulates hitting times for an imaginary target assuming intermittent search, under a proper fitting, these trajectories should yield the same L\'evy exponent. Consequently, our conclusion that the motion resides in the ballistic regime remains robust regardless of the specific detection process.

To verify this, we fitted the step-length model in~\eqref{eq: pdf-step-length} separately to each cell group. The resulting effective exponents were $\widehat\mu=1.072$ for T cells, $\widehat\mu=1.274$ for neutrophils, and $\widehat\mu=1.673$ for NK cells, placing all three estimates in the ballistic regime. Consistent with the empirical trajectories, the fitted processes reproduced both the increase in detection probability with projected target area and the overall ordering among target geometries.

This preliminary analysis is intended mainly to demonstrate how our theory can be applied to empirical trajectories. The fitted exponents depend on the crossover scale $c$ (defined in SI, Section~8), which we set using the published leukocyte arrest-speed threshold of $2\,\mu\mathrm m\,\mathrm{min}^{-1}$~\cite{pizzagalli2019characterization}. Inferring L\'evy-like motion from finite trajectories can also depend on how steps and fitting ranges are defined and on which alternative movement models are considered~\cite{james2011assessing}. We therefore interpret the fitted exponents cautiously.

\begin{figure}[!t]
    \centering
    \includegraphics[width=\linewidth]{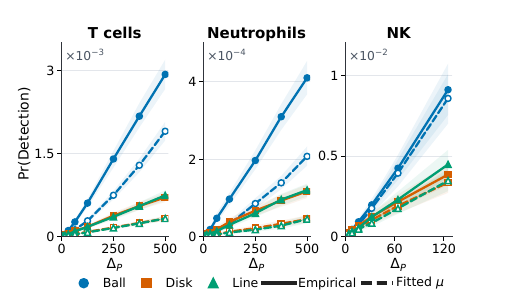}
    \caption{{\bf Application of the framework as a predictive tool.}
    Results of the analysis described in the \emph{Empirical-trajectory analysis} section. The finite endpoint-detection probabilities for ball (blue circles), disk (orange squares), and line (green triangles) targets are matched by projected area $\Delta_P$. Each replay was limited by the total travel distance of its corresponding recorded track and ended when this finite path-length budget was exhausted. Solid curves represent the empirical-trajectory construction, in which the direction of every displacement vector was randomized independently while preserving its magnitude, its order in the observed sequence, and the track's endpoint count. Detection was checked at every displacement-vector endpoint. Dashed curves represent displacement vectors drawn from the corresponding fitted model under the same track-specific budget. Points show equal-track means from 10,000 uniformly initialized replays per track and target. Shaded regions denote pointwise 95\% percentile-bootstrap confidence intervals obtained from 1,000 resamples of complete source tracks, which quantify variation among recorded tracks.}
    \label{fig:empirical_grounding}
\end{figure}

\section*{Discussion}

L\'evy walks are among the most extensively studied models of animal movement, yet the conditions under which they optimize foraging remain a subject of intense debate \cite{james2011assessing,plank2008optimal,reynolds2015liberating,humphries2014optimal}. This debate intensified as the original theoretical framework for the L\'evy Flight Foraging Hypothesis (LFFH) \cite{viswanathan1999optimizing,viswanathan2008levy} was called into question for dimensions greater than one \cite{levernier2020inverse,buldyrev2021comment,levernier2021reply}. While recent theory established the near-optimality of Cauchy walks ($\mu \approx 2$) in 2D intermittent search \cite{GuinardKorman2021}, the 3D case has remained theoretically unresolved. Since three-dimensional environments constitute the physical reality for systems ranging from patrolling T-cells to marine predators, understanding search in this setting is essential. Our work addresses this gap by revealing how search efficiency in 3D is fundamentally shaped by target geometry, providing a unified framework that connects empirical observations with rigorous search theory.

\paragraph*{The geometric complexity of 3D search}

A central finding of this work is that the transition from 2D to 3D search is a qualitative leap rather than an incremental change. Depending on the exponent $\mu$, detection efficiency in 3D intermittent search is governed by an interplay between volume, surface area, and elongation. This geometric complexity creates specific vulnerabilities for specialized strategies. Exploration-dominated strategies ($\mu < 2$) efficiently detect large, spherical targets but perform poorly for highly elongated or flattened structures. Conversely, exploitation-dominated strategies ($\mu > 2$) excel at intercepting line-like targets but become inefficient when searching for compact patches. Notably, both non-Cauchy regimes struggle to detect large planar, disc-like geometries, which can act as a geometric refuge from non-Cauchy searchers.

Our analysis also reveals a dimensional shift in how targets can ``hide'' from searchers. In 2D, elongated targets can retain partial protection from exploitation-dominated L\'evy searchers ($\mu > 2$) \cite{GuinardKorman2021}, whereas in 3D this protection collapses, and elongated shapes become highly detectable across the range $\mu > 2$. This suggests that the ``arms race'' between searcher movement and target morphology operates under fundamentally different geometric rules in three-dimensional environments.

\paragraph*{The Cauchy walk as a geometric equalizer}

Against this added complexity of 3D search, the intermittent Cauchy walk ($\mu = 2$) emerges as a unique geometric equalizer. We prove that it is the only intermittent L\'evy strategy robust to simultaneous variations in both scale and morphology, achieving nearly optimal detection across a broad range of convex shapes and sizes. Although other L\'evy exponents may offer marginal advantages for a specific and known geometry, the Cauchy strategy remains consistently nearly optimal without the need for parameter tuning.

\paragraph*{A transition in geometric sensitivity} Our results also identify a clear transition in geometric sensitivity throughout the L\'evy spectrum. To illustrate this effect, consider a large convex target of fixed volume that maximizes surface area through either elongation or flattening. For exploitation-dominated searchers ($\mu > 2$), these configurations differ sharply in detectability: elongated targets are detected rapidly, whereas disc-like shapes provide a geometric refuge through slower detection. Strikingly, this refuge disappears for the Cauchy strategy, which detects both shapes with comparable efficiency due to their similar surface areas.
\paragraph*{Evolutionary predictions}

Our findings indicate that many previous empirical studies warrant re-evaluation. Our results suggest that movement statistics could be interpreted through the lens of intermittent search and target geometry. We predict that generalist intermittent foragers should favor Cauchy-like patterns to maintain efficiency across diverse target geometries and scales. In contrast, specialists targeting specific morphologies may benefit from systematic deviations, for example, $\mu > 2$ for elongated structures or $\mu < 2$ for large dense aggregates. These results provide a basis for reinterpreting empirical movement patterns across systems ranging from immune surveillance \cite{harris2012generalized} and cancer cell migration \cite{huda2018levy}, to marine foraging \cite{sims2008scaling,humphries2010environmental}.

The immune-cell trajectory analysis also illustrates a possible use of our theory. If an organism's motion can be described as a L\'evy walk, its relative success in finding imaginary ball-like and line-like targets when assuming intermittent detection can provide an informed hypothesis about whether its movement is ballistic, Cauchy, or diffusive, regardless of whether the detection is indeed intermittent or not.

\paragraph*{Scope and outlook}

Our model isolates the role of target geometry by considering a single searcher, a static target, a periodic domain, and intermittent detection at relocation endpoints. It therefore abstracts away interactions among searchers, target motion, environmental heterogeneity, and sensory detection during relocations.

Relaxing these assumptions suggests several directions for future work. With multiple searchers, trajectory overlap and coordination may alter both detection-time scaling and the preferred exponent; in a central-place model on an infinite two-dimensional plane, the optimal exponent depends on the number of walkers and target's distance from the starting point~\cite{clementi2021search}. In nonperiodic environments, walls, obstacles, and dead ends may truncate or redirect long relocations, making performance sensitive to the boundary rule~\cite{dybiec2017levy,garcia2021intrinsic}. Previous work suggests that search performance is affected by both the exponent of the L\'evy searcher and the movement of the target~\cite{reynolds2006optimal,james2011assessing}. Regarding target geometry, our approximate-convexity results apply to all convex targets, including those perturbed by small surface irregularities. For highly folded structures, such as some DNA conformations, however, $\Delta_P$ may be much smaller than $\Delta_B$. Our Cauchy upper bound still applies in this regime, but optimality remains open.

Finally, we stress that Cauchy-like movement is not universal in biology; diffusive, persistent, and context-dependent patterns have been reported for T cells and marine predators~\cite{banigan2015heterogeneous,humphries2010environmental}.

\paragraph*{Conclusion}
By resolving the tension between recent mathematical critiques and empirical observations, our results place the L\'evy Flight Foraging Hypothesis on a firmer, geometry-sensitive footing. While environmental and sensory factors inevitably influence biological behavior, our work uncovers general geometric principles governing search in three-dimensional environments, showing that the intermittent Cauchy walk provides a unique robust solution to these constraints.

The persistence of this strategy across both two and three dimensions suggests a broader universality: under realistic conditions of uncertainty and weak detection, it represents a stable solution for navigation in complex environments. This offers a principled explanation for the frequent observation of $\mu \approx 2$ movement patterns in nature, where robustness mitigates the risk of catastrophic search failure.

More broadly, our findings suggest that the coevolution of movement strategies and environmental structure is a key organizing principle of spatial ecology, with implications for both biological systems and the design of autonomous search algorithms.

\section*{Methods}\label{sec:methods}
\paragraph*{Model}
The torus $\torus$ is identified with the cubic domain
$[-\sqrt[3]{n}/2,\, \sqrt[3]{n}/2]^3$ of volume $n$, embedded in the Euclidean space
$\mathbb{R}^3$, with periodic boundary conditions, that is, opposite faces of the cube are identified.
Let $\mu \in (1,3]$ and $\ell_{\max} > 1$ denote the maximal admissible step length (possibly $\ell_{\max} = \infty$). In this study, a \emph{L\'evy walk} $Z^{\mu}$ on $\mathbb{R}^3$ (or $X^{\mu}$ on $\torus$) with
maximal step length $\ell_{\max}$ is defined as a random walk process whose step lengths are distributed according to:
\begin{equation}\label{eq: pdf-step-length}
p(\ell) =
\begin{cases}
a, & \text{if } \ell \le 1,\\[4pt]
a\, \ell^{-\mu}, & \text{if } 1 < \ell < \ell_{\max},\\[4pt]
0, & \text{if } \ell \ge \ell_{\max},
\end{cases}
\end{equation}
where the normalization constant
\[
a = \left( 1 + \int_{1}^{\ell_{\max}} \ell^{-\mu} \, d\ell \right)^{-1}
\]
ensures that $\int_{0}^{\infty} p(\ell)\, d\ell = 1$. When the process is confined to the torus $\torus$, we set $\ell_{\max} = \sqrt[3]{n}/2$. In the special case where $\mu = 2$, the corresponding process is referred to as a \emph{Cauchy walk}, and its realization on the torus is denoted by
$X^{\mathrm{cauchy}}$. Note that~\eqref{eq: pdf-step-length} is one of infinitely many L\'evy-walk models with the same tail exponent $\mu$, just as many random walks lead to Brownian motion through the Central Limit Theorem. Such models can differ at finite times and in their higher moments~\cite{zaburdaev2015levy}.

The search process starts at a random location of the torus, chosen uniformly at random.
In each step $m$ of the L\'evy walk, the process chooses a direction uniformly at random, and a step length $\ell$ according to the distribution $p(\ell)$. The result is a vector $V(m)$, called {\em relocation} vector. The process then moves over the torus from its current position $X(m-1)$ to $X(m)=X(m-1)+V(m)$, modulo the torus, i.e., respecting its boundary conditions.

The {\em detection radius} of the searcher is $d\geq 1$.

We consider a subset of points $S'$ in the torus, called the {\em raw target}, which may be connected or consist of several disjoint components. For the purpose of our analysis, we define the detectable region~$S \subset \mathbb{T}_n$ as the set of all points within the detection radius~$d$ of some point in $S'$. We hereafter use the term target to refer to this detectable region $S$, rather than the raw target $S'$.
Throughout, we assume that the target's volume is well defined, and that its surface is piecewise smooth so that its surface area, as well as the surface area of its projections onto planes, are well defined. We further note that all our results concerning different parameters of the target, e.g., its surface area, correspond to $S$ and not to the raw target $S'$.

We focus on
\emph{intermittent search}, where it is assumed that the target $S$ cannot be detected while the agent is in motion, that is, during a transition from
$X(m)$ to $X(m+1)$. In this setting, the target is detected at step $m$ if the
agent's position at the end of the step satisfies $X(m) \in S$.

\paragraph*{Geometric objects}
We considered four types of raw target geometries: disks, balls, lines, and rectangles. Recall that because of the detection radius $d \geq 1$, the searcher effectively spots the target by interacting with its \emph{detectable region}. As a consequence, a disk target is thickened into a region of total thickness $2d$, while a {\em line} target of length $L$ is expanded into a cylinder-like region of length $L+2d$ and diameter (or width)~$2d$. Similarly, rectangles are 2D rectangular objects extruded with thickness equal to the detection radius $d$ in the third dimension.

A shape is called {\em approximately convex} if $\Delta_P \geq \frac{1}{36} \Delta_B$. As mentioned, the choice of this constant enables this family to contain all convex shapes, see proof in SI, Section 5. In fact, for convex targets we have $\Delta\leq 6 \Delta_P$ and $\Delta \geq \Delta_B/6$, where $\Delta$ is the surface area of the target (SI, Theorem 5.2).

\paragraph*{Numerical simulations}
We implemented the numerical simulations in the C programming language. The agent was initialized at a random position within $\torus$ and tasked with searching for a target region $S \subset \torus$, whose center of mass was fixed at the origin.

Throughout all numerical simulations, we set $d = 1$ and $\lmax = \sqrt[3]{n}/2$. Unlike \cite{GuinardKorman2021}, we did not approximate the step-length distribution using its discrete counterpart. Instead, we employed the Inverse Transform Sampling method to generate floating-point values directly from the continuous probability density function defined in~\eqref{eq: pdf-step-length}.

Each parameter-and-target configuration was simulated at least 2000 times.

\paragraph*{Empirical-trajectory analysis}
We analyzed manually curated three-dimensional leukocyte trajectories from the Leukocyte Tracking Database \cite{pizzagalli2018leukocyte}. The results of such analysis are summarized in Fig.~\ref{fig:empirical_grounding}. We constructed directed displacement vectors from frame-to-frame displacements and fitted the step-length model in~\eqref{eq: pdf-step-length} separately to each selected cell group. Each recorded trajectory has a fixed path-length budget, given by the sum of its displacement-vector magnitudes. We used each trajectory to construct a \emph{replay}: the observed displacement magnitudes and their order were retained, while their directions were independently randomized and the starting point was chosen randomly. Because the trajectories were recorded in a slab-like imaging volume, this randomization was used to avoid carrying the resulting directional anisotropy into the replays. During each replay, the target was placed in the center of the torus, its detection was checked after every displacement endpoint, and the replay ended when either its empirical path-length budget was exhausted or detection occurred.
The crossover scale, dataset selection, preprocessing, model fitting, replay construction, and uncertainty quantification are described in the SI section 8.

\paragraph*{Discrete setting}
In the discrete approximation of the L\'evy walk we consider a 3D torus graph with $n$ nodes, denoted $\torus$, with the searcher moving over it according to a truncated L\'evy walk $X$ of parameter $\mu$ following the general construction of \cite{clementi2021search}. Specifically, in the beginning, the searcher's initial location is chosen uniformly at random among all nodes in the torus. If the walker is at node $u$ at time $t$, its next location is determined as follows: first, a length $1 \le \ell \le \ell_{\max}$ is sampled from the power-law distribution $p(\ell) \propto \ell^{-\mu}$. Next, a node $v$ is chosen uniformly at random from the set of nodes at Manhattan distance $\ell$ from $u$.
When the transition happens, the actual path followed by the agent to reach the designated node is any shortest path between the source and the destination\footnote{In \cite{clementi2021search}, the authors considered a continuous detection setting and therefore defined a probability distribution over the set of shortest paths. However, since we work in the intermittent detection setting, introducing such a distribution is unnecessary.}. Hence, the time spent in this step is $\ell$. If the agent stays still (because of laziness), the time spent on this step is $0$.

Generally, as in the continuous (non-discrete) setting, the reader should keep in mind the distinction between the number of steps complexity (which includes steps where the agent stays still) and the time complexity, which takes into account the duration of each step. In particular, for a given target $S$, we denote by the {\em hitting step} $\mhit^X(S)$,  the number of steps until the searcher detects $S$, and by $\mmix^X(\varepsilon)$ the {\em mixing step} as the minimum number of steps needed for the total variation distance to the stationary distribution to drop below $\varepsilon$. We use this notion instead of the usual mixing and hitting time to emphasize that it counts steps rather than time. In particular, we write $\mmix^X := \mmix^X(1/4)$ and note that $\mhit^X$ plays the analogous role to $m_\mathrm{detect}^X$ in the continuous space.

\paragraph*{Asymptotic notation}\hfill\break
We employ standard Bachmann--Landau notation to capture asymptotic behavior as the torus size $n \to \infty$. For non-negative functions $f,g: \mathbb{N} \to [0,\infty)$, we write $f(n) = O(g(n))$ if there exist $c > 0$ and $n_0 \in \mathbb{N}$ such that $f(n) \leq c g(n)$ for all $n \geq n_0$. For example, $n^2 = \bigO{n^3}$ since $n^2 \leq n^3$ for $n \geq 1$. Likewise, $f(n) = \Omega(g(n))$ if $c g(n) \leq f(n)$ for all $n \geq n_0$, for some $c > 0$, $n_0 \in \mathbb{N}$; e.g., $n^3 = \Omega(n^2)$. The notation $f(n) = \Theta(g(n))$ indicates both bounds hold simultaneously, as in $3n^2 + 2n = \Theta(n^2)$.

\section*{Data, Materials, and Software Availability}
The empirical trajectories analyzed in this study are from the publicly available Leukocyte Tracking Database (LTDB) \cite{pizzagalli2018leukocyte}, archived at \url{https://doi.org/10.6084/m9.figshare.c.3827890}. The
simulation and analysis source code is available at
\url{https://github.com/pragmaticscientist/3D-levy-walks}. All results presented in this paper can be fully reproduced using this code.

\section*{Acknowledgments}
The authors thank Dr. Yuwen Wang for helpful discussions.
Portions of the manuscript text were edited with the assistance of AI language models (ChatGPT and Gemini) to improve clarity and grammar. Generative AI tools (Perplexity, Claude Code, Codex, and ChatGPT) were also used to assist with writing and revising mathematical proofs in the supplementary material and with writing, refactoring, and debugging code for the analysis pipelines. The authors verified all AI-assisted outputs and remain fully responsible for the integrity and accuracy of the work. This work was supported by the Israel Science Foundation (https://www.isf.org.il) Grant 1574/24.

\section*{Author Contributions}
Designed research: MS, EN, AK; Performed research: MS, AK; Simulations: MS; Visualisation: MS; Wrote the paper: AK, Writing—SI: MS.

\section*{Competing Interests}
The authors declare no competing interests.

\bibliographystyle{unsrtnat}
\bibliography{references}

\clearpage
\onecolumn
\setcounter{equation}{0}
\renewcommand{\theHequation}{SI.\arabic{equation}}
\setcounter{theorem}{0}
\renewcommand{\theHtheorem}{SI.\thesection.\arabic{theorem}}
\renewcommand{\thesubsection}{\Alph{subsection}}
\begin{center}
{\Large\bfseries Supplementary Materials\par}
\vspace{0.8em}
{\large\bfseries Intermittent Cauchy walks enable optimal 3D search across target shapes and sizes\par}
\vspace{0.7em}
Matteo Stromieri, Emanuele Natale, and Amos Korman\par
\vspace{0.25em}
\small Correspondence: amoskorman@cs.haifa.ac.il
\end{center}
\vspace{1.25em}
\section*{Supporting Information Text}
\section{Preliminary results}
In this section, we state several other results that were established in~\cite{GuinardKorman2021} for the two-dimensional scenario, and can be directly extended to the three-dimensional case with minor adaptations to the proofs.

\subsection{Time versus number of steps}

In the two-dimensional scenario, it was established in~\cite{GuinardKorman2021} that the detection time equals the expected number of steps until detection $m_{\mathrm{detect}}^X(S)$ multiplied by the mean step duration~$\tau$. The same proof also holds in the three-dimensional case, requiring only the substitution of two-dimensional relocation vectors with three-dimensional ones. By the same argument, the result generalizes to arbitrary spatial dimensions. Therefore, the proof is omitted here for brevity.

Formally, the 3D analogue of Claim~3 in the Supplementary Material of~\cite{GuinardKorman2021} is

\begin{claim}\label{claim:timeVSmoves}
For any intermittent random walk $X$ on $\torus$ and any subset $S \subseteq \torus$,
\begin{align*}
t_{\mathrm{detect}}^X(S) = \mathbb{E}\!\left[m_{\mathrm{detect}}^X(S)\right] \cdot \tau,
\end{align*}
where $\tau = \mathbb{E}[\|V(1)\|]$ denotes the expected step length.
\end{claim}
\subsection{Monotonicity}
A function $f$ on $\mathbb{R}^3$ is said to be \emph{radial} if there exists a function
$\tilde{f} : \mathbb{R}^{+} \to \mathbb{R}$ such that $f(x) = \tilde{f}(\|x\|)$ for every $x \in \mathbb{R}^3$. In this case, $f$ is called \emph{non-increasing} if $\tilde{f}$ is non-increasing. The proof of the following claim follows exactly the same steps as Claim~4 in the Supplementary Material of~\cite{GuinardKorman2021} and is therefore omitted here.

\begin{claim}\label{cor:monotonicity-variables}
Let $X$ and $Y$ be independent random variables taking values in $\mathbb{R}^3$,
with probability density functions $f$ and $g$, respectively.
Let $h$ denote the probability density function of $X + Y$.
If both $f$ and $g$ are radial and non-increasing,
then $h$ is also radial and non-increasing.
\end{claim}

\begin{corollary}[Monotonicity]\label{cor:monotonicity}
Let $Z$ be a random walk process on $\mathbb{R}^3$, starting at $Z(0)=0$, with a step-length distribution $p$. If $p$ is non-increasing, then for any $m\geq 1$ the distribution $p^{Z(m)}$ of $Z(m)$ is radial and non-increasing. In particular, for any $x,x'$ points in $\mathbb{R}^3$ with $\norm{x'}\leq \norm{x}$, we have $p^{Z(m)}(x)\leq p^{Z(m)}(x')$.  Furthermore, for any $x\in \mathbb{R}^3$ and any $m\geq 1$, $p^{Z(m)}(x) \leq \frac{3}{4\pi\norm{ x}^3}.$
\end{corollary}
\begin{proof}
The fact that $p^{Z(m)}$ is radial and non-increasing follows from Claim \ref{cor:monotonicity-variables} by induction on $m$. Indeed, the step-length vectors $V(1), V(2),\dots$ are independent and, by hypothesis, admit a radial, non-increasing p.d.f. Hence so does $Z(m)=V(1)+V(2)+\dots+V(m)$.
Next, consider $x\in \mathbb{R}^3$ and the ball $B$ of radius $\norm{x}$ centered at $0$. Note that for every point $y\in B$, we have  $\norm{y}\leq \norm{x}$, and hence, by the fact that $Z(m)$ is radial and non-increasing, we have $p^{Z(m)}(x)\leq p^{Z(m)}(y)$.
The upper bound on $p^{Z(m)}(x)$ follows from:
\begin{equation}
1\geq\int_{B}p^{Z(m)}(y)\,dy\geq p^{Z(m)}(x) \lvert B\rvert = p^{Z(m)}(x) \cdot \frac{4\pi}{3}\norm{x}^3,
\end{equation}
completing the proof.
\end{proof}

\subsection{Coordinate projection preserves the L\'evy exponent}
In this section, we show that the projection of a L\'evy walk onto any single axis is itself a L\'evy walk, with the same exponent. While a related proof was presented in~\cite{GuinardKorman2021} for the two-dimensional case, our approach generalizes the result to three dimensions and is more naturally extendable to arbitrary dimensions with only minor modifications.
In what follows, we denote by $Z_1^\mu$ the projection of the L\'evy walk $Z^\mu$ onto an arbitrary coordinate axis.

\begin{theorem}\label{thm:projection}
The projection
    $Z_1^{\mu}$ of $Z^{\mu}$ is a L\'evy walk on $\mathbb{R}$ with parameter $\mu$, in the sense that the p.d.f.~of the step-lengths of $X_1^{\mu}$ is $p(\ell)\sim 1/\ell^\mu$, for $\ell\in [1,\frac{\ell_{max}}{2}]$. Furthermore, the variance of $X_1^{\mu}$ is
\begin{align*}
\sigma'^2= \begin{cases}
\Theta({\ell^{3-\mu}_{max}}) \text{ if } \mu\in(1,3) \\
  \Theta(\log \ell_{max}) \text{ if } \mu=3\end{cases}.
\end{align*}
\end{theorem}

\begin{proof}
It is clear that $Z_1^\mu$ is also a random walk that moves incrementally, with the increments between $Z_1^\mu(m)$ and $Z_1^\mu(m+1)$ being the projection $Z_1(m+1)$ of the chosen 3-dimensional vector $V(m+1)=Z^\mu(m+1)-Z^\mu(m)$. These projections are i.i.d.~variables as the vectors $(V(m))_m$ are i.i.d.~variables, and their signs are $\pm$ with equal probability. Next, we verify that $l_1:=\lvert V_1(1) \rvert $ has a L\'evy distribution with parameter $\mu$.

Let $V$ be one step-length drawn according to a L\'evy distribution $p^{\mu}$ in the 3-dimensional torus. Recall that the step-length is distributed as:
\begin{align*} p^{\mu}(\ell)=\begin{cases} a_{\mu} \text{ if } \ell \leq 1 \\
a_{\mu}\ell^{-\mu} \text{ if } \ell\in [1,\lmax) \\
0 \text{ if } \ell \geq \lmax
\end{cases}, \end{align*}
where $a_{\mu}$ is the normalization factor, with $a_{\mu}=\frac{1}{1+\int_{1}^{\ell_{max}} \ell^{-\mu} d\ell }=\frac{1}{1+\frac{1-\ell_{max}^{1-\mu}}{\mu-1}}$.
Hence the distribution of $V=(V_1,V_2, V_3)\in \mathbb{R}^3$ is
\begin{align} p^V(v)= \frac{1}{4\pi}\frac{1}{\norm{ v}^2 } p^{\mu}(\norm{ v})=\begin{cases}\frac{a_{\mu}}{4\pi}\norm{ v}^{-2} \text{ if } \norm{ v}  \leq 1\\ \frac{a_{\mu}}{4\pi}\norm{ v} ^{-\mu-2}\text{ if } \norm{ v}\in [1,\ell_{max}) \\
0 \text{ if } \norm{v} \geq \lmax\end{cases}. \label{eq:law-step-R2}
\end{align}
Recall that $l_1:=\lvert V_1(1) \rvert $ is the distribution of the length of a projected step onto the first coordinate. Let $v=(v_1,v_2, v_3)$ denote a realization of the random vector $V$. For $v_1\in (0,\ell_{max})$, we have
\begin{align*}
    p^{l_1}(v_1) = 2\int_I p^V(v)\, dv = \frac{a_\mu}{2\pi} \int_I \left[ \mathbf{1}(\|v\| < 1) \frac{1}{\|v\|^2} + \mathbf{1}(\|v\| \geq 1) \frac{1}{\|v\|^{\mu + 2}} \right] dv,
\end{align*}
 where $I \subset \mathbb{R}^3$ comprises all vectors $v$ whose total length is constrained by a maximum step size $\ell_{max}$ and the first entry is equal to $v_1$ (i.e., $I$ is a disc in the plane orthogonal to $v_1$, centered at $v_1$, and of radius $\sqrt{\ell_{max}^2 - v_1^2}$). Formally, we define this set as:$$I = \left\{ v \in \mathbb{R}^3 : P_x(v) = v_1 \mbox{~and~} \| P_{y,z}(v) \| \leq \sqrt{\ell_{max}^2 - v_1^2} \right \}.$$
Here, $P_x$ denotes the projection onto the $x$-axis, while $P_{y,z}$ represents the projection onto the $yz$-plane. The inequality governing the $yz$-component is derived directly from the Euclidean norm; because the total magnitude $\|v\|$ cannot exceed $\ell_{max}$, the magnitude of the lateral projection is geometrically restricted by the remaining length available after accounting for the longitudinal component $v_1$. Now assume $|v_1| \geq 1$ and observe that this implies $\norm{v} \geq 1$. Hence,
\begin{align}
     p^{l_1}(v_1) = \frac{a_\mu}{2\pi} \int_I \frac{1}{\|v\|^{\mu + 2}} dv = \frac{a_\mu}{2\pi} \int_I \frac{1}{{v_1^{\mu+2}\left[1+(\frac{v_2}{v_1})^2+(\frac{v_3}{v_1})^2\right]}^{\frac{\mu + 2}{2}}} dv . \label{eq:surface_integral}
\end{align}
Note that \eqref{eq:surface_integral} is a surface integral. Thus, we can evaluate it by parameterizing the surface using the mapping $\mathbf{r}(w_1, w_2) = (v_1, v_1 w_1, v_1 w_2)$ and applying the standard formula for a surface integral over a parameterized region $D$:$$\int_I f(v) \, dA(x) = \int_D f(\mathbf{r}(w_1, w_2)) \left\| \frac{\partial \mathbf{r}}{\partial w_1} \times \frac{\partial \mathbf{r}}{\partial w_2} \right\| dw_1 \, dw_2.$$

where $dA(x)$ denotes an infinitesimal surface element at the point $x$ on the surface $I$ and $D = \mathbf{r}^{-1}(I)$ is the domain of the parameters $w_1$ and $w_2$.
By substituting our specific parameterization into this expression, we can transform the integral into the $w$-coordinate system to simplify the subsequent calculations. Thus, $\frac{\partial r}{\partial w_1} = v_1(0, 1, 0)^{\top}$, $\frac{\partial r}{\partial w_2} = v_1(0, 0, 1)^{\top}$ and $\norm{\frac{\partial r}{\partial w_1} \times \frac{\partial r}{\partial w_2}} = \norm{(v_1^2, 0, 0)^{\top}}=v_1^2$. To conclude this step, we observe that the change of variables $v_2 = v_1w_1$ and $v_3 = v_1w_2$ allows us to represent the set $I$ as the image of a simplified domain $D$ under the mapping $\mathbf{r}$. Specifically, $I = \mathbf{r}(D)$, where $D$ is defined as the disk:

$$D = \left\{ w \in \mathbb{R}^2 : \|w\| \leq \sqrt{\frac{\ell_{max}^2}{v_1^2} - 1} \right\}.$$

In fact, for any $w \in D$, the parameterization $\mathbf{r}(w) = (v_1, v_1w_1, v_1w_2)$ satisfies the required conditions for $D$. First, the longitudinal component is fixed at $v_1$. Second, the total magnitude satisfies the bound $\| \mathbf{r}(w) \|^2 = v_1^2(1 + w_1^2 + w_2^2) = v_1^2(1 + \|w\|^2) \leq \ell_{max}^2$, ensuring that the vector remains within the maximum step length. Therefore, we can rewrite the previous integral as follows:
\begin{align*}
    p^{l_1}(v_1)  = \frac{a_\mu}{2\pi} \int_{D} \frac{v_1^2}{{v_1^{\mu+2}(1 + \norm{y}^2)}^{\frac{\mu+2}{2}}} dy = \frac{a_\mu}{2\pi v_1^\mu} \int_{D} \frac{1}{{(1 + \norm{y}^2)}^{\frac{\mu+2}{2}}} dy.
\end{align*}
Now we focus on the remaining integral and show that it is upper bounded by a constant. Indeed,
\begin{align*}
    \int_{D} \frac{1}{{(1 + \norm{y}^2)}^{\frac{\mu+2}{2}}} dy \leq \int_{\mathbb{R}^2} \frac{1}{{(1 + \norm{y}^2)}^{\frac{\mu+2}{2}}} dy = \bigO{1}.
\end{align*}
In fact, since $\mu > 1$, we have that $\frac{\mu + 2}{2} \geq 1$. To see why the integral evaluates to a constant, we change to polar coordinates. Let $y = (r\cos\theta, r\sin\theta)$, so that $\|y\| = r$ and the area element is $dy = r \, dr \, d\theta$. The integral over $\mathbb{R}^2$ becomes:
\begin{align*}
    \int_{\mathbb{R}^2} \frac{1}{(1 + \|y\|^2)^{\frac{\mu+2}{2}}} dy &= \int_{0}^{2\pi} \int_{0}^{\infty} \frac{r}{(1 + r^2)^{\frac{\mu+2}{2}}} dr \, d\theta \\
    &= 2\pi \int_{0}^{\infty} r(1 + r^2)^{-\frac{\mu+2}{2}} dr.
\end{align*}
By applying the substitution $u = 1 + r^2$ (with $du = 2r \, dr$), the integral simplifies to:
\begin{align*}
    \pi \int_{1}^{\infty} u^{-\frac{\mu+2}{2}} du &= \pi \left[ \frac{u^{1 - \frac{\mu+2}{2}}}{1 - \frac{\mu+2}{2}} \right]_{1}^{\infty} \\
    &= \pi \left[ \frac{u^{-\mu/2}}{-\mu/2} \right]_{1}^{\infty}.
\end{align*}
Since $\mu > 1$, the term $u^{-\mu/2}$ vanishes as $u \to \infty$. Evaluating at the limits gives:
\begin{align*}
    \pi \left( 0 - \frac{1}{-\mu/2} \right) = \frac{2\pi}{\mu}.
\end{align*}
Because $\mu$ is a fixed parameter, the result $\frac{2\pi}{\mu}$ is independent of $v_1$, and therefore the expression is indeed $\bigO{1}$.
Hence, for $|v_1| \geq 1$, \eqref{eq:surface_integral} can be rewritten as
\[p^{l_1}(v_1) = \bigO{1/v_1^\mu}.\] The final stage of the proof establishes that $p^{l_1}(v_1) = \Omega(1/v_1^\mu)$ holds within the restricted domain $|v_1| \leq \frac{\ell_{\text{max}}}{2}$. Under this assumption -- specifically that $v_1$ lies in this central sub-interval -- we observe that
\begin{align*}
    p^{l_1}(v_1) = \frac{a_\mu}{2\pi v^\mu} \int_{D} \frac{1}{{(1 + \norm{y}^2)}^{\frac{\mu+2}{2}}} dy \geq \frac{a_\mu}{2\pi v^\mu} \int_{D'} \frac{1}{{(1 + \norm{y}^2)}^{\frac{\mu+2}{2}}} dy
\end{align*}
where $D' = \left\{y \in \mathbb{R}^2: \norm{y} \leq 1\right \}$ and the inequality follows from $D' \subseteq D$ for $|v_1| \leq \frac{\ell_{\text{max}}}{2}$. To conclude, observe that $\int_{D'} \frac{1}{{(1 + \norm{y}^2)}^{\frac{\mu+2}{2}}} dy$ is a positive constant. Thus, we have that $p^{l_1}(v_1) = \Theta(1/v_1^{\mu})$ for $v_1 \in [1, \lmax/2]$ and $p^{l_1}(v_1) = \bigO{1/v_1^{\mu}}$ for $v_1 \in [\lmax/2, \lmax]$. To prove the claim about the variance of the projected random walk, we followed exactly the same steps as in \cite{GuinardKorman2021} (Supplementary Material, Theorem 6); hence, we omit the rest of the proof.
\end{proof}

\section{Universal lower bound}\label{sec:universal}

In this section, we establish Theorem \ref{thm:universal-lower-box}, which states that the expected detection time of any compact target $S$ is $\Omega(n/\Delta_B)$. Here, $\Delta_B = \Delta_B(S)$ denotes the {\em largest face area} of $S$, namely, the surface area of largest face of the box $B=\texttt{Box}(S)$ that minimizes surface area while containing $S$ (red shape in Fig. 2 in the main text). This lower bound is universal for any search strategy; consequently, it applies to any L\'evy walk regardless of the exponent $\mu$, even under continuous detection.

\begin{theorem}\label{thm:universal-lower-box}
    Let $X$ be any search process on the torus. Consider any compact target~$S$, then the expected time to detect $S$ is $\Omega(n/\Delta_B)$ where $\Delta_B$ denotes the largest face area of the bounding box of $S$.
\end{theorem}

\begin{proof}
For simplicity, assume $B$ is axis-aligned with side lengths $x_1 \le x_2 \le x_3$ along the main axes, where $x_1 \ge 2$ because the detection radius satisfies $d \ge 1$ (see Methods Section). Consider an enlarged box $B'$ with side lengths $3x_1$, $3x_2$, $3x_3$ and tile the torus with copies of $B'$ (See Figure 2 in the main text for a visual representation), placing adjacent boxes exactly $x_1$ apart along each axis. This yields $M$ boxes, with
$$
M = \Theta\left(\frac{n}{x_1 x_2 x_3}\right).
$$
Recall that our model has a fixed target $S$ and initial searcher position $X(0)$ uniform on $\torus$. This is equivalent (by translation invariance) to having a fixed initial position at the origin, i.e.,$X(0)=0$, and a uniform center of mass $x^*$. Thus, the probability that the target is fully contained in the $i$-th box $B'_i$ is $\mathbb{P}(S \subseteq B'_i) = x_1 x_2 x_3 / n = \Theta(1/M)$, indeed the subset of $B'$ where $x^*$ could lie is a box with sides $x_1$, $x_2$ and $x_3$ (it is a copy of $B$ placed in the center of $B'$). So, the probability that the target belongs to some box is a constant:
$$\mathbb{P}(S \subseteq \bigcup_i^M B'_i) = \sum_i^Mx_1 x_2 x_3 / n = \Theta(1).$$
Hence, by law of total expectation, it suffices to lower bound the expected time to find $S$, conditioning on $S$ being fully contained within some box $B'$.

Moreover, conditioning on this event, each box has equal probability to contain the target, since all boxes are identical, up to translation in space. The target box $B^\star$ containing the target is therefore uniformly distributed over $\{B'_1,\dots,B'_M\}$.

Fix a search strategy $X$, and consider a single realization of it, which induces an ordering of the boxes according to the times at which they are first visited; we denote this sequence by $\texttt{Seq}(X) = (b_1,\dots,b_M)$; note that different realizations of a randomized strategy may produce different sequences. Since the target box is chosen uniformly at random, its expected position in any such sequence is $(M+1)/2$. Averaging over all choices of $B^\star$ and all realizations of $X$, it follows that, in expectation, $X$ must visit $\Omega(M)$ distinct boxes before entering $B^\star$. Moreover, traveling between two distinct boxes requires time $\Omega(x_1)$. It follows that
\[
\mathbb{E}[t_{\mathrm{detect}}^X(S) \mid S \subseteq \bigcup_i^M B'_i]
= \Omega(x_1 M)
= \Omega\!\left(\frac{n}{x_2 x_3}\right).
\]
Recalling that $\Delta_B = x_2 x_3$ completes the proof.
\end{proof}

\section{Lower bounds for non-Cauchy L\'evy walks}\label{sec:lower-bound-levy}
In this section, we derive lower bounds on the detection time of L\'evy walks that do not follow the Cauchy distribution. Specifically, for $1 < \mu < 2$, we show that the corresponding intermittent L\'evy walks are inefficient at locating targets of small surface area. Conversely, for $2 < \mu \leq 3$, we demonstrate that these L\'evy walks perform poorly when detecting  targets of large surface area, except when they are line-shaped.

To this end, we introduce the quantity $T_d$, the first time the process is expected to leave the ball of radius $d$. In the intermittent setting, it denotes the time at which the endpoint of a step exits the ball of radius $d$ centered at the starting point. Note that $T_d$ provides a lower bound on the time required to detect a target located at distance $d$ or greater.

Recall that $B=\texttt{Box}(S)$ denotes the bounding box of $S$ of minimum surface area, with side lengths $x_1 \le x_2 \le x_3$, where we assume without loss of generality that $B$ is aligned with the coordinate axes (See Figure 2 in the main text for a visual representation).

\begin{theorem}\label{th: lower_bound_delta}
        Let $X$ be any search process on the torus. Consider any compact target $S$ with a bounding box $B=\texttt{Box}(S)$ of largest surface area $\Delta_B$ and $\delta \in [\frac{1}{2},1]$ such that $x_3 = \Delta_B^{\delta}$ and $x_2 = \Delta_B^{1-\delta}$, then the expected time to detect $S$ is \[\Omega\left(\frac{n}{\Delta_B^{2-\delta}}T_{\Delta_B^{1-\delta}}\right).\]
\end{theorem}
\begin{proof}
Consider an enlarged box $B'$ with side lengths $3x_1$, $3x_2$, and $3x_3$. Tile the torus with copies of $B'$ as in the proof of Theorem~\ref{thm:universal-lower-box}, positioning them so that any two adjacent boxes are separated by a distance of exactly $x_2=\Delta_B^{1-\delta}$. This yields
\begin{align*}
    \Theta\left(\frac{n}{(x_1+\Delta_B^{1-\delta})(x_2+\Delta_B^{1-\delta})(x_3+\Delta_B^{1-\delta})}\right) = \Theta\left(\frac{n}{\Delta_B^{1-\delta}x_2x_3}\right) = \Theta\left(\frac{n}{\Delta_B^{2-\delta}}\right)
\end{align*}
such boxes. Next, by the same reasoning used in the proof of Theorem~\ref{thm:universal-lower-box}, even if the searcher is aware of this tiling, it must examine at least half of these boxes on average before finding the one that contains the target. Moving from one box to the next requires time $T_{\Delta_B^{1-\delta}}$, so the expected detection time is
\[
    \Omega\left(\frac{n}{\Delta_B^{2-\delta}}T_{\Delta_B^{1-\delta}}\right),
\]
concluding the theorem.
\end{proof}

\subsection{Lower bounds for L\'evy walks with \texorpdfstring{$\mathbf{1 < \mu < 2}$}{1 < mu < 2}}\label{SM:lower_small_mu}

Recall that $m_{\mathrm{detect}}^{X^{\mu}}(S)$ denotes the number of steps needed by the L\'evy walk with exponent $\mu$ to hit the target. The following claim provides a lower bound for the number of steps, which holds for any exponent $\mu\in(1,3]$. The proof is deferred to Section~\ref{sec:lower_number_steps}.

\begin{claim}\label{claim_steps}
Consider an intermittent L\'evy walk $X^\mu$ with $\mu\in(1,3]$ and a target $S$. Then,
\[
\mathbb{E}[m_{\mathrm{detect}}^{X^{\mu}}(S)] = \Omega\left(\frac{n}{V}\right),
\]
where $V$ denotes the volume of $S$.
\end{claim}

Due to a variant of Wald's identity (see Claim~\ref{claim:timeVSmoves}), multiplying this lower bound by the average step length of a given L\'evy walk yields a lower bound on the detection time for this L\'evy process. As given by ~\eqref{obs:tau}, for $\mu \in (1,2)$, the average step length is
 $\Theta(n^{\varepsilon/3})$, implying the following lower bound on the detection time:

\begin{theorem}
Consider a target $S$ of volume $V$. The detection time of an intermittent L\'evy walk with $\mu \in (1,2)$ is:
\[t_{\text{detect}}^{X^\mu}(S) = \Omega\left(\frac{n^{1+\varepsilon/3}}{V}\right).\]
\end{theorem}

\subsubsection{Proof of Claim \ref{claim_steps}}\label{sec:lower_number_steps}

\begin{proof}
Observe that when $V=\Theta(n)$ the claim is straightforward, thus we prove it for $V < {n}/{4}$. Recall that the  searcher starts from a random point of the torus, which is equivalent to letting it start at the origin, and locate the target's center of mass uniformly at random in the torus. Denote the target with its center of mass at the point $x$ by $S(x)$.
Similarly, define $R_t:=X(t)-S(0)$, the reflection of the target centered at $X(t)$. Observe that if, at time $t$, the L\'evy walk is at a location $X(t) \notin S$, then $R_t$ coincides with the set of positions at which the center of mass of the target cannot be located. Indeed, we have
    \begin{align*}
    X(t) \in S(x)
    &\iff \exists\, w \in S(x) \text{ such that } w = X(t) \\
    &\iff \exists\, w \in S(0) \text{ such that } x + w = X(t) \\
    &\iff \exists\, w \in S(0) \text{ such that } x = X(t) - w \\
    &\iff x \in R_t .
    \end{align*}
    Now, define the set
    \begin{align}\label{eq:ub_vol}
        F_t := \bigcup_{j=0}^t R_j,
    \end{align}
    which is obtained by taking the union of all regions excluded at each time step \(j \le t\), and therefore contains every point that cannot be the target’s center of mass after \(t\) steps without hitting the target. Observe that, for each $j$,
    \[
    \operatorname{Vol}\bigl(R_j\bigr) \le V.
    \]
    Indeed, $R_t$ results from a reflection and a shift of $S$, and hence has the same volume. Thus, we can compute the following upper bound on the volume of $F_t$:
    \[
    \operatorname{Vol}(F_t)
    = \operatorname{Vol}\!\left(\bigcup_{j=0}^t R_t\right)
    \le \sum_{j=0}^t \operatorname{Vol}\bigl(R_t\bigr)
    \le (t+1)V .
    \]
    Next, we show that when $t < \frac{n}{2V}$,
    \[
    \mathbb{P}(m_{\mathrm{detect}}^{X^{\mu}}(S) = t \mid m_{\mathrm{detect}}^{X^{\mu}}(S) \geq t) = \mathcal{O}\!\Big(\frac{V}{n}\Big).
    \]

     Recall that $m_{\mathrm{detect}}^{X^\mu}(S)$ denotes the number of steps for the L\'evy walk $X^\mu$ before it hits the target $S$, then $m_{\mathrm{detect}}^{X^{\mu}}(S) \geq t$ indicates that the target was not found in the first $t-1$ steps. Moreover, since the center of mass of the target is chosen uniformly at the start of the search, it follows that after $t-1$ unsuccessful steps the conditional distribution of $S$ at step $t$ is uniform over $\mathbb{T} \setminus F_{t-1}$, which has volume at least
    \[
    n-tV\geq n - \frac{n}{2} = \frac{n}{2}.
    \]
    Moreover, ~\eqref{eq:ub_vol} guarantees that the set of possible locations for the target center has volume at most $V$.
    Therefore, the probability can be bounded as

    \[
    \mathbb{P}(m_{\mathrm{detect}}^{X^{\mu}}(S) = t \mid m_{\mathrm{detect}}^{X^{\mu}}(S) \geq t) \leq \frac{2V}{n}.
    \]
     This shows that, for $t < n/(2V)$, the probability of hitting the target in a single step is at most twice the fraction of the state space occupied by the target. We now lower bound the expectation of $m_{\mathrm{detect}}^{X^{\mu}}(S)$. By definition,
\[
\mathbb{E}[m_{\mathrm{detect}}^{X^{\mu}}(S)] = \sum_{t=1}^{\infty} \mathbb{P}(m_{\mathrm{detect}}^{X^{\mu}}(S) \ge t) \geq \sum_{t=1}^{\frac{n}{4V}} \mathbb{P}(m_{\mathrm{detect}}^{X^{\mu}}(S) \ge t).
\]
Observe that the probability that the target has not been found by time $t$ can be bounded below as
\[
\mathbb{P}(m_{\mathrm{detect}}^{X^{\mu}}(S) \ge t) = \prod_{j=0}^{t-1} \Big(1 - \mathbb{P}(m_{\mathrm{detect}}^{X^{\mu}}(S)  = j+1 \mid m_{\mathrm{detect}}^{X^{\mu}}(S) > j)\Big)
\ge  \Big(1 - \frac{2V}{n}\Big)^t.
\]

This tail lower bound means the distribution has at least geometric decay, implying
\[\mathbb{E}[m_{\mathrm{detect}}^{X^{\mu}}(S)] \ge \sum_{t=1}^{\frac{n}{4V}} (1-2V/n)^t = \Omega\left(\frac{n}{V}\right)\]
which completes the proof.
\end{proof}

\subsection{Lower bounds for L\'evy walks with \texorpdfstring{$2 < \mu \leq 3$}{2 < mu <= 3}}

The following claim concerns a general random walk process on the 3D torus and thus applies not only to our main subject of interest—the L\'evy walk—but also to other search processes.
\begin{claim}\label{claim:max-length}
Consider a random walk process $X$ on the torus $\torus$ and let $\sigma'$ denote the standard deviation of the length of the projected steps onto one of the coordinates.
\begin{itemize}
    \item For every step $m \in \mathbb{N}$, $\mathbb{E}[{\max_{s\leq m} \norm{X(s)-X(0)}}] = O(\sqrt{m}\sigma')$.
    \item Let $m_d$ be the number of steps needed to go to distance at least $d<\sqrt[3]{n}/2$, in other words, $m_d$ is the first step $m$ for which $\norm{X(m)-X(0)}\geq d$. We have $\E(m_d)=\Omega(d^2/\sigma'^2)$.
    \item If the process is intermittent and $\tau$ denotes the average length of a jump, then the expected time before reaching distance $d<\sqrt[3]{n}/2$ is $T_d=\Omega(d^2/\sigma'^2 \tau)$.
\end{itemize}
In particular, if the process is intermittent and $L$ is the maximal length in the support of the step-length distribution, then the expected time needed to go to a distance $\Omega(\sqrt{n})$ is $\Omega( \frac{n}{L})$.
\end{claim}

\begin{proof}
The proof of the claim is similar to the corresponding proof in \cite{GuinardKorman2021} (Claim 10 in the Supplementary Material), except for minor adaptations to the 3D setting. We nevertheless include the proof here.

 Given the search process $X$ on the torus, let $Z$ be the process on $\R^3$, with $Z(0)=X(0)$ and evolving with the same relocation vectors as $X$ (except that $Z$ does not wrap around the torus due to the boundary conditions). Since the distance between $Z(m)$ and $Z(0)$, in $\R^3$, is always at least that of $X(m)$ and $X(0)$, in $\torus$, the number of steps needed to go to distance $d$ in $\torus$ is at least as high as in $\R^3$. Hence, we may analyze the process $Z$ instead of $X$.

 Define $d^Z_{max}(m)$ as the maximal distance (from the initial point) that the process $Z$ reached from step~$0$ up to step $m$, i.e.,
 \begin{align*}
    d^Z_{max}(m)=\max_{s\leq m} \norm{Z(0)- Z (s)}.
 \end{align*}
 Now write $Z=(Z_1,Z_2,Z_3)$, let $p'$ be the p.d.f. of the projected step-lengths (i.e., the p.d.f. of the step-lengths of $Z_i$, for any $i\in\{1,2,3\}$), and let $\tau'$ and $\sigma'$ be respectively its mean and standard deviation. Next, let $d^Z_{i,max}(m)$ be the maximal distance reached by the projection on coordinate $i=1,2,3$. Since steps are independent, the standard deviation of $Z_i(s)$, for $s\leq m$, is $\sqrt{s}\sigma'\leq \sqrt{m}\sigma'$. By Kolmogorov's inequality, we have for any $\lambda>0$,
 \[\Pr(d^Z_{i,max}(m) \geq \lambda \sqrt{m} \sigma' )\leq \frac{1}{\lambda^2}.\] Furthermore, since $d^Z_{max}(m)\leq \sqrt{3}\max\{d^Z_{1,max}(m),d^Z_{2,max}(m),d^Z_{3,max}(m)\}$, we have by a union bound argument, that for any $\lambda>0$,
\begin{align*}
    \Pr(d^Z_{max}(m)\geq \lambda\sqrt{m}\sigma') \leq  \sum_{i=1,2,3}\Pr\left (d^Z_{i,max}(m)\geq \frac{\lambda}{\sqrt{3}}\sqrt{m}\sigma' \right )\leq \frac{9}{\lambda^2}.\end{align*}
	Hence,
	\begin{align} \E(d^Z_{max}(m))&=\int_{s=0}^\infty \Pr\left(d^Z_{max}(m)\geq s\right)ds = \sum_{\lambda=0}^\infty \int_{\lambda'=0}^{\sqrt{m}\sigma'} \Pr\left(d^Z_{max}(m)\geq \lambda \sqrt{m}\sigma' + \lambda'\right) d\lambda' \nonumber \\
    &\leq \sqrt{m}\sigma' \left(\sum_{\lambda\geq 0} \Pr(d^Z_{max}(m)\geq \lambda \sqrt{m}\sigma')\right)=O\left (\sqrt{m}\sigma'\right),  \label{eq:max-distance-general-upper-bound}\end{align}
    where the second equality is due to separating the integration into the intervals of width $\sqrt{m}\sigma'$
    and the inequality is due to setting $\lambda'=0$ for each interval.
	~\eqref{eq:max-distance-general-upper-bound} proves the first item of Claim \ref{claim:max-length}.

    Next, write the $m_d$ of the statement as $m^X_d$, to distinguish it from the similarly defined $m^Z_d$, which is the first step for which $\norm{Z(m)-Z(0)}\geq d$. As remarked above, we have $m^X_d\geq m^Z_d$. Note that for $m\geq m^Z_d$,  we have $d^Z_{max}(m)\geq d^Z_{max}(m_d)\geq d$. Therefore, by Markov's inequality,
\begin{align} \label{eq:markov-inq-distances-steps}
\E\left(d^Z_{max}(2\E(m^Z_d))\right)\geq\E\left(d^Z_{max}(2\E(m^Z_d))\mid m^Z_d<2\E(m^Z_d)\right)\cdot \Pr\left(m^Z_d<2\E(m^Z_d)\right)
\geq d \cdot \frac{1}{2}.
\end{align}
Now using~\eqref{eq:max-distance-general-upper-bound} with $m=2\E(m^Z_d)$, we have $\E(d^Z_{max}(2\E(m^Z_d)))=O(\sqrt{\E(m^Z_d)}\sigma')$ and hence, by~\eqref{eq:markov-inq-distances-steps},

\begin{align*}\E(m^X_d)\geq \E(m^Z_d)=\Omega\left(\frac{d^2}{\sigma'^2}\right),
\end{align*}
which proves the second item of Claim \ref{claim:max-length}.

The last item in the claim is a lower bound on $T_d=\E(T(m^X_d))$, the expected time that $X$ needs to reach distance $d$, assuming that $X$ is intermittent. To obtain it, we observe that $m^X_d$ is the hitting step of the set of nodes at distance $d$ or more in the torus. Hence, by Claim~\ref{claim:timeVSmoves}, we have $T_d=\E(m^X_d)\cdot \tau\geq \E(m^Z_d)\cdot \tau=\Omega(\frac{d^2}{\sigma'^2}\tau)$, which is exactly as needed.

Finally, observe that
\begin{align}\label{eq:variance1}
     \sigma'^2 = \int_{0}^{L} p'(\ell)\ell^2d\ell \leq L\int_{0}^{L} p'(\ell)\ell d\ell = L\tau'\leq L\tau,
\end{align}
     where the last inequality is justified by the fact that the projection reduces distances. This completes the proof of Claim \ref{claim:max-length}.
\end{proof}

Towards proving the lower bound for $\mu\in (2,3]$, we first
establish the following.

\begin{claim}\label{claim:distance_large_mu}
 Let $X^{\mu}$ be an intermittent L\'evy walk process on the torus $\torus$, for some fixed $\mu\in (2,3]$. The expected time required to reach a distance of
 $d\geq 1$ from the starting point is:
  \begin{align*}
      T_d = \begin{cases}
          \Omega(d^{\mu-1}) \text{ if } \mu \in (2,3)\\
          \Omega\left(\frac{d^2}{\log d}\right) \text{ if } \mu = 3
      \end{cases}.
  \end{align*}
\end{claim}

The proof of the claim in the intermittent setting is identical to that of Claim~13 in the Supplementary Material of \cite{GuinardKorman2021}, hence we omit it from this manuscript. Notably, the argument relies on Theorem~\ref{thm:projection}, even though that theorem is stated for $\mathbb{R}$ while the present claim concerns the torus $\torus$.
This link is made precise by Claim~\ref{claim:max-length}, which provides a lower bound on the time required to reach distance~$d$ in terms of the mean step length and the projected variance of a single step length.

\begin{theorem}\label{th: lower_bound_large_mu_surface}
Let $\mu\in(2,3]$ and let $\Delta_B$ be the largest face area of the bounding box $B=\texttt{Box}(S)$ containing a compact target $S$ with $\delta \in [\frac{1}{2},1]$ such that $x_3 = \Delta_B^{\delta}$ and $x_2 = \Delta_B^{1-\delta}$. Write $\mu=2+\epsilon$. The detection time of the intermittent L\'evy walk $X^\mu$ with respect to a target $S$ is
\begin{align*}
t_{detect}^{X^\mu}(S)=
 \begin{cases}\Omega\left(n\Delta_B^{\varepsilon(1-\delta) - 1}\right) \text{ if } \mu \in (2,3) \\
\Omega\left(\frac{n}{\Delta_B^{\delta}\log \Delta_B}\right) \text{ if } \mu = 3\end{cases}.
\end{align*}
\end{theorem}

\begin{proof}
     The theorem follows by combining the lower bound $\Omega\left(\frac{n}{\Delta_B^{2-\delta}}T_{\Delta_B^{1-\delta}}\right)$ given by Theorem \ref{th: lower_bound_delta} with the lower bound on $T_d$ given by Claim \ref{claim:distance_large_mu}. Specifically, for the case where $\mu \in (2,3)$, we obtain:
\begin{align*}
\Omega\left(\frac{n}{\Delta_B^{2-\delta}}T_{\Delta_B^{1-\delta}}\right) = \Omega\left(\frac{n}{\Delta_B^{2-\delta}}\Delta_B^{(1-\delta)(\epsilon + 1)}\right) = \Omega(n\Delta_B^{\varepsilon(1-\delta) - 1}),
\end{align*}
which matches the desired result. On the other hand, for $\mu = 3$, the expression becomes:
\begin{align*}
\Omega\left(\frac{n}{\Delta_B^{2-\delta}}T_{\Delta_B^{1-\delta}}\right) = \Omega\left(\frac{n \Delta_B^{2(1-\delta)}}{\Delta_B^{2-\delta}(1-\delta)\log \Delta_B}\right) = \Omega\left(\frac{n}{\Delta_B^{\delta}\log \Delta_B}\right).
\end{align*}
This completes the derivation for the second case and concludes the proof.
\end{proof}

Since the elongation of a disc or a ball is $\delta=1/2$, the above theorem implies the following.

\begin{corollary}
    \label{cor: lower_bound_large_mu_surface}
Let $\mu\in(2,3]$ and write $\epsilon=\mu-2$.  Consider either a ball or a disc target $S$ of surface area $\Delta$. The detection time of the intermittent L\'evy walk $X^\mu$ with respect to $S$ is
\begin{align*}
t_{detect}^{X^\mu}(S)=
 \begin{cases}\Omega\left(n\Delta^{\varepsilon/2 - 1}\right) \text{ if } \mu \in (2,3) \\
\Omega\left(\frac{n}{\Delta^{1/2}\log \Delta}\right) \text{ if } \mu = 3\end{cases}.
\end{align*}
\end{corollary}

\subsection{The \texorpdfstring{$\delta$}{delta} of a line target}\label{sec:line}
When the raw target is a line segment of length \(L\), the corresponding target
region is a cylindrical shape with radius equal to the detection radius
\(d \ge 1\). For this class of targets, the area of the bounding box
\(B = \texttt{Box}(S)\) satisfies \(\Delta_B = x_3 x_2\), where the largest side is
\(x_3 = L + 2d\) and the second largest side is \(x_2 = 2d\). By the definition of \(\delta\),
it follows that
\[
\Delta_B^\delta = x_3 = \frac{\Delta_B}{2d}.
\]
Taking the logarithm with base \(\Delta_B\) on both sides yields
\[
\delta = 1 - \frac{\log(2d)}{\log(\Delta_B)}.
\]
Since \(\Delta_B \ge (2d)^2\), it follows that the parameter \(\delta\) is always
bounded below by \(\tfrac{1}{2}\).

\section{The detection time of the Cauchy walk}\label{sec:detection-time-cauchy}
In this section, we aim to prove Theorem \ref{thm:upper-bound-cauchy-box}, which states that the expected time required for the Cauchy walk to find a  target $S$ is $\bigO{n \log^3 n / \Delta_P}$. Here, $\Delta_P = \Delta_P(S)$ denotes the {\em projected surface area} of the target---defined as the maximum of the areas of the projections of $S$ onto each of the six faces of the minimum surface area bounding box $B = \texttt{Box}(S)$.
We begin with the following lemma, whose proof is identical to the proof of Lemma 19 in \cite{GuinardKorman2021}, and is therefore omitted from this manuscript.

\begin{lemma}\label{detect_upper_bound}
    Consider a random walk process $Z$ on $\mathbb{R}^3$ and its projection $X$ on the torus $\mathbb{T}_n$ and denote by $Z^{z_0}$ the process $Z$ starting at $Z(0)=z_0$. Let $S \subset \torus$, then for any integer $m_0$ we have
\begin{align} \label{eq_detect_upper_bound}
    \mathbf{E}(m_{detect}^X(S)) = \bigO{m_0 \cdot \frac{ \sup_{z_0 \in S} \sum_{m=0}^{m_0} \Pr(Z^{z_0}(m) \in S) }{ \sum_{m=m_0}^{2m_0} \Pr(Z(m) \in S) }}.
\end{align}
\end{lemma}

\subsection{Technical lemmas and claims}

Next, we present two lemmas, analogous to Lemmas 20 and 21 in \cite{GuinardKorman2021}, whose proofs require only minor modifications. The full proofs are deferred to Section \ref{sec:auxilary claims}.

\begin{lemma} \label{lower_bound_prob}
    For any constant $\alpha > 0$, there exists a constant $c>0$ such that for any integer $m \in [1, \alpha \lmax]$, and any $x\in \R^3$, with $\norm{x} \leq m$,
    \begin{align}
        p^{Z(m)}(x) \geq \frac{c}{m^3}.
    \end{align}
\end{lemma}

\begin{lemma} \label{ball_upper_bound_log}
    For any constant $\alpha > 0$, there exists a constant $c' > 0$ such that, for any integer $m \in [2, \alpha \lmax]$, and for any $x \in \mathbb{R}^3$, we have
    \begin{align}
        \Pr(Z(m) \in B(x)) \leq \frac{c' \log^3 m}{m^3}.
    \end{align}
\end{lemma}

\subsection{Optimal hitting time}

Now we can use all the claims and lemmas above to prove the following theorem.

\begin{theorem} \label{thm:upper-bound-cauchy-box}
    Consider the Cauchy walk $X^\text{cauchy}$ on the torus. The hitting time of the walk with respect to a compact target $S$ of projected surface area $\Delta_P$ is
    \begin{align}
        \tdetect = \bigO{\frac{n\log^3 n}{\Delta_P}}.
    \end{align}
\end{theorem}

\begin{proof}
    By definition of {\em projected surface area}, there exists a projection $\pi$ onto one face of the bounding box $B=\texttt{Box}(S)$ such that $\texttt{Area}(\pi(S)) = \Delta_P$.
    Moreover, following \cite{vershynin2018high}, the $3$-packing number of any two-dimensional set with area $\Delta_P$ is at least $c\Delta_P$ where $c$ is a universal constant. Here, the $3$-packing number refers to the maximum cardinality of a subset of points where each pair is separated by a distance of at least $3$. Let $S^{\text{proj}} \subseteq \pi(S)$ be a set of points realizing the packing number. For each $w \in S^{\text{proj}}$, we define its pre-image under the projection, $\pi^{-1}(w)$, as the point in $S$ closest to the projection plane:
    \begin{align*}
    \pi^{-1}(w) = \arg\min \{d(w, w') : w' \in S, \pi(w') = w\}.
    \end{align*}
    For each $w' \in \pi^{-1}(S^{\text{proj}})$, we select a point $x \in S$ such that $\|x - w'\| = 1$ and the distance from $x$ to the surface of $S$ is at least $1$. The existence of such a point follows directly from the definition of target, given the assumption $d \geq 1$ for the detection radius. Let $S^* \subset S$ denote the set of these chosen points. Since projections are non-expansive (i.e., they do not increase distances), the mapping from $S^{\text{proj}}$ to $S^*$ is injective, and thus $|S^*| = |S^{\text{proj}}| = c\Delta_P$. Observe that the union of balls $B(S^*)$ is contained in $S$. Consequently, $\E\big(t_{\mathrm{detect}}^{X}(S)\big) \leq \E\big(t_{\mathrm{detect}}^{X}(B(S^\star))\big)$, therefore we focus on $B(S^\star)$.\\

    Recall from Claim \ref{claim:timeVSmoves} that $\E\big(t_{\mathrm{detect}}^{X}(B(S^\star))\big)
    = \E\big(m_{\mathrm{detect}}^{X}(B(S^\star))\big)\, \tau$, where $\tau$ is the expected step length of the random process. Thus, we can prove our theorem by giving an upper bound to ~\eqref{eq_detect_upper_bound} (Lemma \ref{detect_upper_bound}) applied to $B(S^\star)$ with $m_0 = \sqrt[3]{n}$. We begin by lower bounding the denominator, then we proceed by upper bounding the numerator.

    \paragraph{Lower bound for the denominator in~\eqref{eq_detect_upper_bound}.} Note that any $x \in B(S^\star)$ trivially satisfies $\norm{x} \leq m$ for any $m \geq m_0 + 1$, because $m_0$ was chosen to be slightly more than the diameter of the torus. Therefore, we can apply Lemma \ref{lower_bound_prob} to get a lower bound on the denominator.
    \begin{align*}
        \sum_{m=m_0+1}^{2m_0}\Pr(Z(m)\in B(S^\star) ) &= \sum_{m=m_0+1}^{2m_0} \int_{x\in B(S^\star)}p^Z_m(x)dv\\
        &\geq \sum_{m=m_0+1}^{2m_0} \frac{c}{m^3}\lvert B(S^\star)\rvert =\Omega\left( \frac{\Delta_P}{n^{2/3}}\right).
    \end{align*}

    \paragraph{Upper bound for the numerator in~\eqref{eq_detect_upper_bound}.}

    Next, we provide an upper bound to the numerator of the r.h.s of Eq. \ref{eq_detect_upper_bound} which is the expected number of returns to $B(S^\star)$ within the first $m_0$ steps, conditioning on the fact that $Z(0)=z$, for some $z\in B(S^*)$. Let us denote this process by $Z^{z}$ (note that $Z=Z^0$). Then,
    \begin{align*}
    \sum_{m=0}^{m_0}\Pr(Z^{z}(m)\in B(S^\star)) &\leq 2+\sum_{m=2}^{m_0} \Pr(Z^{z}(m)\in B(S^\star)). \label{eq:returns-except-3}
    \end{align*}
    Clearly, the probability density function $p^{Z^{z}(m)}$ of $Z^{z}(m)$ is obtained by a translation from $p^{Z(m)}$. Thus, by Corollary \ref{cor:monotonicity}, we have for any $y\in \mathbb{R}^3$:
    \begin{align*}
            p^{Z^z(m)}(y) = \bigO{\frac{1}{\norm{y-z}^3}}.
    \end{align*}
    In particular, for $y$ such that $\norm{y-z} \geq 2$,
    \begin{align*}
        \Pr(Z^z(m) \in B(y)) = \bigO{\frac{1}{(\norm{y-z}-1)^3}}
    \end{align*}
    Next, denote by $x_z \in S^\star$ the point whose ball contains $z$. Now, let $r>0$ be a value that we define later and let us denote by $S^\star_{\leq r}$ the subset of $S^\star$ whose points are at distance at most $r$ from $x_z$. Analogously, define $S^\star_{>r} = S^\star \setminus S^\star_{\leq r}$. We proceed with the following decomposition:
    \begin{align*}
    \Pr(Z^{z}(m)\in B(S^\star)) = \Pr\left(Z^{z}(m)\in B(S^\star_{>r})\right) + \Pr\left(Z^{z}(m)\in B(S^\star_{\leq r})\right).
    \end{align*}

    Note that $|\{ w \in \pi(S^*): \|w-\pi(x_z)\| \leq r \}| \leq 4r^2$, thus by the non-expansiveness of the projections we have $|S^\star_{\leq r}| \leq 4r^2$. Hence, using Lemma \ref{ball_upper_bound_log} we get:
    \begin{align*}
        \Pr\left(Z^{z}(m)\in B(S^\star_{\leq r})\right) = \bigO{\frac{r^2c'\log^3m}{m^3}}.
    \end{align*}
    Next, we aim to upper bound $\Pr\left(Z^{z}(m)\in B(S^\star_{>r})\right)$. By the triangle inequality, for any $x \in S^\star_{>r}$, we have $\norm{x-z} \geq \norm{x-x_z} -1 > 1$. Hence, we get:
    \begin{align*}
        \Pr\left(Z^{z}(m)\in B(S^\star_{>r})\right) &= \sum_{x \in S^*_{>r}}{\bigO{\frac{1}{(\norm{x-z}-1)^3}}} \\
        &= \sum_{x \in S^*_{>r}}{\bigO{\frac{1}{(\norm{x-x_z}-2)^3}}}\\
        &= \sum_{x \in S^*_{>r}}{\bigO{\frac{1}{(\norm{\pi(x)-\pi(x_z)}-2)^3}}}\\
        &= \sum_{k \in \mathbb{Z}: k\geq r+2 } \bigO{\frac{1}{k^2}} \\
        &= \bigO{ \frac{1}{r} }.
    \end{align*}
    where the third equality follows from the non-expansiveness of the projections. Therefore, by setting $r = m / \log m$, we get:
    \begin{align*}
        \Pr(Z^{z}(m)\in B(S^\star)) = \bigO { \frac{\log m}{m}}.
    \end{align*}
    Hence, the numerator of ~\eqref{eq_detect_upper_bound} can be upper bounded as follows:
    \begin{align*}
        \sum_{m=0}^{m_0}\Pr(Z^{z}(m)\in B(S^\star)) &= 2+\sum_{m=2}^{m_0} \bigO { \frac{\log m}{m}}\\
        &= \bigO { \log^2 n }.
    \end{align*}
    Altogether, combining the bounds for the numerator, those for the denominator, and Lemma \ref{detect_upper_bound} (with $m_0 = \sqrt[3]{n}$), we get the following result:

    \begin{align}
            \E(\mdetect) = \bigO { \frac{n \log^2 n }{\Delta_P} }.
    \end{align}
     Finally, using Claim \ref{claim:timeVSmoves}, we multiply this bound by the expected step-length $\tau = \bigO{\log n}$ to get the desired upper bound for the detection time, namely, \[\tdetect = \bigO{\frac{n \log^3 n }{\Delta_P}}.\]
\end{proof}

\subsection{Proofs of technical lemmas and claims}\label{sec:auxilary claims}
\subsubsection{Proof of Lemma \ref{detect_upper_bound}}
Since the proof of this lemma is the same as that of Lemma 19 in \cite{GuinardKorman2021}, we omit it here.
\subsubsection{Proof of Lemmas \ref{lower_bound_prob} and \ref{ball_upper_bound_log}}

\begin{observation}\label{obs:tail-law}
 The probability of taking a step of length at most $\ell> 0$ is $a\ell$ if $\ell \leq 1$ and $a(2-\frac{1}{\ell})$ if $\ell >1$. For integers $\lmax\geq\ell_2 \geq \ell_1\geq 1$, the probability to choose a length in $[\ell_1,\ell_2]$ is $a( \frac{1}{\ell_1}-\frac{1}{\ell_2})$.
\end{observation}

The next claim quantifies the probability that the Cauchy process goes to a distance of at least $d$ after $m$ steps.

\begin{claim}\label{claim:X-far}For any integer $m\geq 2$ and any real $d\in [1,\frac{\lmax}{3}]$  we have,
\begin{align*}
\Pr\left(\exists s\leq m \mbox{~s.t.~} \norm{ Z(s) }\geq d
\right)\geq 1-e^{-cm/d},
\end{align*}
for some constant $c>0$.
 In particular this lower bound is at least
 \begin{itemize}
    \item $1-O(m^{-3})$ if $d=c'\frac{m}{\log m}$ with $c'>0$ a small enough constant,
    \item $\Omega(1)$ if $d=c'm$ for any constant $c'>0$ with $c'm\leq \lmax/3$.
\end{itemize}
\end{claim}

As the argument is identical to the proof of the analogous claim in the two-dimensional case (Claim 24 in the Supplementary Material of \cite{GuinardKorman2021}), we omit it.

\begin{claim}\label{claim:X-m-less-than-cm}
\begin{itemize}
    \item For any constant $c>0$, there exists a constant $\delta>0$ such that, for any two integers $1\leq s\leq m$, we have $\Pr(\norm{ Z(s)} \leq cm )\geq \delta$.
    \item For any constant $0<\delta<1$, there exists a (large enough) constant $c>0$ such that, for any two integers $1\leq s\leq m$, we have $\Pr(\norm{ Z(s)} \leq cm )\geq \delta$.
\end{itemize}
\end{claim}
\begin{proof}
    Fix an integer $m \ge 1$ and let $c''$ be a constant to be specified later. Define $\mathcal{A}$ as the event that each of the first $m$ steps has length at most $\ell = c'' m$. Then, for any integer $s \le m$ and any constant $c > 0$, we have

\begin{align}
    \label{eq:norm-X-cond-small-steps}\Pr(\norm{ Z(s)} \leq cm ) \geq \Pr(\mathcal{A}) \cdot\Pr(\norm{ Z(s)} \leq cm  \mid\mathcal{A}).
\end{align}
We now analyze each term appearing on the right-hand side of ~\eqref{eq:norm-X-cond-small-steps} individually and show that:
\begin{itemize}
    \item For the first item of Claim~\ref{claim:X-m-less-than-cm}, we choose $c''>0$ such that both factors become constants, so their product is bounded below by some constant $\delta$.
    \item For the second item of Claim~\ref{claim:X-m-less-than-cm}, where the bound $\delta$ is prescribed, we will show that both terms can be made at least $\sqrt{\delta}$ by choosing $c$ and $c''$ appropriately.
\end{itemize}
 Proceeding with the first term in the r.h.s~of ~\eqref{eq:norm-X-cond-small-steps}, by Observation \ref{obs:tail-law}, we have:
\begin{align*}
\Pr(\mathcal{A})=\begin{cases}(ac''m)^m \text{ if } c''m\leq 1 \\
(2a)^m(1-\frac{1}{2c''m})^m \text{ if } c''m\in [1,\lmax]\\
1 \text{ if } c''m\geq \lmax \end{cases}. \end{align*}
For $1\leq m\leq \frac{1}{c''}$, we have $(ac''m)^m \geq (ac''m)^{\frac{1}{c''}}$ as $ac''m\leq c''m \leq 1$, and $ (ac''m)^{\frac{1}{c''}}\geq (ac'')^{\frac{1}{c''}}$ as $m\geq 1$. For the second item,  note that the function $(1-\frac{\alpha}{x})^x=e^{x\log(1-\frac{\alpha}{x})}$ is increasing in $x\geq \alpha$ and thus, for $x\geq 2\alpha$, we have $(1-\frac{\alpha}{x})^x\geq 2^{-2\alpha}$. Applying this with $\alpha=\frac{1}{2c''}$, we have, $(1-\frac{1}{2c''m})^m\geq 2^{-\frac{1}{c''}}$, for $m\geq \frac{1}{c''}$. Overall, using $2a\geq 1$, we get
\begin{align*}
\Pr(\mathcal{A})\geq \begin{cases}(\frac{c''}{2})^{\frac{1}{c''}} \text{ if } c''m\leq 1 \\
2^{-\frac{1}{c''}} \text{ if } c''m\in [1,\lmax]\\
1 \text{ if } c''m\geq \lmax\end{cases}. \end{align*}
Hence,
\begin{itemize}
    \item $\Pr(\mathcal{A})=\Omega(1)$ for any given $c''>0$.
    \item Moreover, for the second item of Claim~\ref{claim:X-m-less-than-cm}, where $0<\delta<1$ is prescribed, we may choose $c''$ sufficiently large (in particular, $c''\ge 1$, so that $c''m \ge 1$) to guarantee that $\Pr(\mathcal{A}) \ge 2^{-1/c''} \ge \sqrt{\delta}$.
\end{itemize} We are now ready to lower bound the second factor in ~\eqref{eq:norm-X-cond-small-steps}, namely
$\Pr(\|Z(s)\|\le cm \mid \mathcal{A})$. We introduce the following notation: for a random variable $X$,
let $X^{\mathcal{A}}$ denote $X$ conditioned on the event $\mathcal{A}$.
Our first objective is to show that
\begin{align}\label{eq:Z-Chebyshev}
    \Pr(\|Z^{\mathcal{A}}(s)\| \le cm)
    \ge 1 - \frac{\text{27}s\,\E(\|V^{\mathcal{B}}\|^2)}{c^2 m^2},
\end{align}
where $V^{\mathcal{B}} = (V_1^{\mathcal{B}}, V_2^{\mathcal{B}}, V_3^{\mathcal{B}})$ is a one-step increment of the walk on $\R^3$,
conditioned on the event $\mathcal{B}$ that its length is at most $c''m$.

~\eqref{eq:Z-Chebyshev} will be established by applying Chebyshev's inequality on each of the projections on the axes and using a union bound argument. Specifically, decomposing the walk $Z$ on the three axes, by writing $Z= (Z_1,Z_2,Z_3)$, we first use a union bound to obtain:
\begin{align*}
\Pr(\norm{ Z^{\mathcal{A}}(s)} > cm  ) &\leq  \Pr( \exists i=1,2,3 \mbox{~s.t.~} \lvert Z_i^{\mathcal{A}}(s)\rvert  > cm/3)\\
&  \leq \Pr(\lvert Z_1^{\mathcal{A}}(s)\rvert  > cm/3)+\Pr(\lvert Z_2^{\mathcal{A}}(s)\rvert  > cm/3) +\Pr(\lvert Z_3^{\mathcal{A}}(s)\rvert  > cm/3)\\
&\leq 3 \Pr( \lvert Z_1^{\mathcal{A}}(s)\rvert  > cm/3),
\end{align*}
where we used the symmetry to deduce that $Z_1$, $Z_2$, and $Z_3$ share the same distribution.
Hence,
\begin{align*}
\Pr(\norm{ Z^{\mathcal{A}}(s)} \leq cm ) \geq 1-3\Pr(\lvert Z_1^{\mathcal{A}}(s)\rvert  > cm/3).
\end{align*}
Next, we aim to lower bound the right-hand side.
Using that $\E[Z_1^{\mathcal{A}}(s)] = 0$ for all $s$,
Chebyshev's inequality yields:
\begin{align*}
\Pr(\lvert Z_1^{\mathcal{A}}(s) \rvert > cm/3 )\leq \frac{9\Var(Z_1^{\mathcal{A}}(s) )}{c^2m^2}.
\end{align*}
Since $Z_1^{\mathcal{A}}(s)$ is the sum of $s$ independent steps that follow the same law as $V_1^{\mathcal{B}}$, we have:
\begin{align*} \Var(Z_1^{\mathcal{A}}(s))= s\Var(V_1^{\mathcal{B}}).\end{align*}
As the expectation of $V_1^{\mathcal{B}}$ is zero, we have $\Var(V_1^{\mathcal{B}})=\E((V_1^{\mathcal{B}})^2)$. Furthermore, since $\lvert V_1^{\mathcal{B}}\rvert \leq \norm{V^{\mathcal{B}}}$, we obtain:
\begin{align*}\Var(Z_1^{\mathcal{A}}(s))\leq s\E(\norm{V^{\mathcal{B}}}^2),
\end{align*}
which concludes the proof of ~\eqref{eq:Z-Chebyshev}. Next, we estimate $\E(\|V^{\mathcal{B}}\|^{2})$.
If $c''m \le 1$, then under the conditioning on $\mathcal{A}$, the step length is chosen uniformly at random from $[0, c''m]$.
 Thus, its second moment is \begin{align}\label{eq:norm-V-m-small}
    \E(\norm{V^{\mathcal{B}}}^2)=\int_{0}^{c''m}\ell^2 \frac{d\ell}{c''m}=\frac{(c''m)^2}{3}.
\end{align}
On the other~hand, if $c''m\geq 1$, then $V^{\mathcal{B}}$ is a Cauchy walk with cut off $\ell_{max}=c''m$. Hence, its second moment is
\begin{align}
   \E(\norm{V^{\mathcal{B}}}^2)&= a'\int_0^1 \ell^2d\ell +a'\int_1^{c''m} \ell^2\ell^{-2}d\ell \\
     &\leq a'\int_0^{c''m} 1d\ell=a'c''m\leq c''m.\label{eq:norm-V-m-big}
\end{align}
Overall, by~\eqref{eq:Z-Chebyshev},~\eqref{eq:norm-V-m-small} and \eqref{eq:norm-V-m-big} we find that, for $s\leq m$,
\begin{align*}
    \Pr(\norm{ Z^{\mathcal{A}}(s)} \leq cm  )\geq \begin{cases}1-\frac{27sc''^2}{3c^2} \text{ if } c''m\leq 1\\
    1-\frac{27sc''}{c^2m} \text{ if } c''m\geq 1
    \end{cases}\\
    \geq \begin{cases}1-\frac{27c''}{3c^2} \text{ if } c''m\leq 1\\
    1-\frac{27c''}{c^2} \text{ if } c''m\geq 1
    \end{cases}.
\end{align*}
We then conclude the proof of Claim \ref{claim:X-m-less-than-cm} by observing the following.
\begin{itemize}
    \item For the first item of Claim \ref{claim:X-m-less-than-cm}, we have proved that $\Pr(\mathcal{A})=\Omega(1)$ for any constant $c''>0$. Hence, we may now choose $c''$ small enough so that $\Pr(\norm{ Z^{\mathcal{A}}(s)} \leq cm  )=\Omega(1)$.
    \item For the second item of Claim \ref{claim:X-m-less-than-cm}, we have already chosen $c''$ to be large (in order to have $\Pr(\mathcal{A})\geq \sqrt{\delta}$), but we are free to choose $c$ large enough so that $\Pr(\norm{ Z^{\mathcal{A}}(s)} \leq cm  )\geq \sqrt{\delta}$.
\end{itemize}
\end{proof}
To conclude, we give the proofs of Lemmas \ref{lower_bound_prob} and \ref{ball_upper_bound_log}.

\paragraph*{Proof of Lemma \ref{lower_bound_prob}}
Given the claims established above, the proof of Lemma~\ref{lower_bound_prob} follows the same steps as the proof of Lemma~20 in~\cite{GuinardKorman2021}. Therefore, we omit it.

\paragraph*{Proof of Lemma \ref{ball_upper_bound_log}}

Let $\alpha>0$ and $m\in [2,\alpha \lmax]$. By the monotonicity property stated in Corollary~\ref{cor:monotonicity}, it suffices to prove the result for $x=0$. Indeed, for any $x\in\R^{3}$, the sets $B(0)\setminus B(x)$ and $B(x)\setminus B(0)$ have the same volume $A$ and

\begin{align*}
\Pr\left(Z(m)\in B(x)\setminus B(0)\right)
&\leq A \max_{y\in B(x)\setminus B(0)}
    \left\{\text{$p^{Z(m)}(y)$}\right\} \\
&\leq A \min_{y\in B(0)\setminus B(x)}
    \text{}\left\{\text{$p^{Z(m)}(y)$}\right\} \\
&\leq \Pr\left(Z(m)\in B(0)\setminus B(x)\right).
\end{align*}

where the second inequality relies on the monotonicity property together with the observation that any point in
$B(x)\setminus B(0)$ is more than unit distance from the origin, and thus farther from $0$ than every point in $B(0)\setminus B(x)$.
 This shows that $\Pr(Z(m)\in B(x))\leq \Pr(Z(m)\in B(0))$, hence it is
 sufficient to prove the required upper bound for $x=0$. Intuitively, we argue that with high probability there is some step $s \le m$ at which $Z(s)$ becomes
“distant’’ (at least $cm/\log m$ away).
Once this is conditioned on, the probability of being in $B(0)$ at step $m$ is small, thanks to the
monotonicity property (Corollary~\ref{cor:monotonicity}). Formally, consider a (small) positive constant $c$, and let $\mathcal{A}$ be the event that there is some $s\leq m$ for which $\norm{Z(s)}\geq cm/\log m$. Let $B(0)$ be the ball of radius $1$ with center $0$. Write
\begin{align}
\Pr(Z(m)\in B(0))&= \Pr(Z(m)\in B(0) \cap \mathcal{A})+\Pr(Z(m)\in B(0) \cap \neg \mathcal{A}) \nonumber \\
&\leq \Pr(Z(m)\in B(0)\mid \mathcal{A})+\Pr(\neg \mathcal{A}), \label{eq:Z-in-M-cond-distant}
\end{align}
By the first item of Claim \ref{claim:X-far}, taking $c$ to be sufficiently small, we have
\begin{align*}
\Pr(\neg \mathcal{A})=O(m^{-3}).
\end{align*}
To rewrite the remaining term in ~\eqref{eq:Z-in-M-cond-distant}, we introduce the notation $Z^{x}$
for the Cauchy process on $\R^{3}$ with cutoff $\ell_{\max}$ starting from $Z(0)=x$.
Because our process starts at $0$, we have $Z = Z^{0}$.
Observe further that $Z^{x}$ has the same distribution as $Z^{0}$ translated by $x$.
Using this notation, we obtain the next bound, with the second inequality following from the Markov property:

\begin{align*}
  \Pr(Z^0(m)\in B(0)\mid \mathcal{A})  &\leq \max_{s\leq m}\Pr(Z^0(m)\in B(0)\mid \norm{Z^0(s)}\geq cm/\log m) \\
  &\leq \max_{s\leq m}\sup_{\norm{x}\geq cm/\log m} \Pr(Z^x(m-s)\in B(0))  \\
 &= \max_{s\leq m}\sup_{\norm{x}\geq cm/\log m} \Pr(Z^x(s)\in B(0)) \\
 &= \max_{s\leq m}\sup_{\norm{x}\geq cm/\log m} \Pr(Z^0(s)\in B(-x) )\\
 &= \max_{s_\leq m}\sup_{\norm{x}\geq cm/\log m}  \Pr(Z(s)\in B(x))
 \end{align*}
 Now use Corollary~\ref{cor:monotonicity} that gives $p^{Z(m)}(x) \leq \frac{3}{4\pi\norm{ x}^3}$. Hence, for any $x\in \R^3$ with $\norm{x}>1$, we have
 \begin{align*}
\Pr(Z(m)\in B(x))=\int_{B(x)} p^{Z(m)}(y) dy \leq \int_{B(x)} \frac{3}{4\pi(\norm{ x}-1)^3}dy= \frac{1}{(\norm{x}-1)^3}. \end{align*}
Let $m(c)$ be the largest integer $m>0$ such that $cm/\log m\leq 2$. For $m>m(c)$, we have
\begin{align*}
    \Pr\left (Z(s)\in B(x)\right ) & \leq \max_{s\leq m} \frac{1}{(\text{$cm/\log m$}-1)^{3}}=\frac{1}{(\text{$cm/\log m$}-1)^{3}}
\end{align*}
Overall, we find that, for $m>m(c)$
\begin{align*}
\Pr\left (Z(m)\in B(0)\right ) \leq \frac{1}{(cm/\log m-1)^{3}} + \frac{c'}{m^{3}},
\end{align*}
which we can bound by $\frac{c_2\log^3 m}{m^3}$ for some constant $c_2>0$. Since $m(c)$ is a constant, there is some other constant $c_3>0$ for which, for any $m\in [2,m(c)]$, we have $\Pr(Z(m)\in B(0))\leq \frac{c_3\log^3 m}{m^3}$. We then obtain, for any $m\geq 2$,
\begin{align*}\
\Pr(Z(m)\in B(0)) \leq \frac{\max\{c_2,c_3\} \log^3 m}{m^3},
\end{align*}
which concludes the proof.

\section{Approximately convex targets}\label{sec:approximately_convex}
In this section, we introduce the class of {\em approximately convex} targets and show that (1) the Cauchy walk is almost optimal for this class and (2) it contains all convex targets.

\begin{definition}
A compact target $S$ is approximately convex if
\[
\frac{\Delta_P}{\Delta_B} \geq \frac{1}{36}.
\]
\end{definition}

Observe that (1) this ratio is always at most 1 by definition, and (2) if the condition holds,
it remains valid upon scaling the target by any factor $r > 0$. Moreover, the class of valid targets grows larger as the right-hand side constant decreases.
Although this constant could be relaxed further, the remainder of this section establishes
that its current value already captures a broad class of targets.
Moreover, Section~\ref{sec:counterexample} presents a counterexample of a target that is not approximately convex, valid for any constant.

\subsection{Optimality of the Cauchy walk}
In this section, we demonstrate that for approximately convex targets, the Cauchy walk achieves optimal performance up to polylogarithmic factors. To show this, we first note that, by the definition of approximate convexity, both the universal lower bound (Theorem \ref{thm:universal-lower-box}) and the upper bound for $\mu = 2$ (Theorem \ref{thm:upper-bound-cauchy-box}) can be equivalently expressed in terms of either parameter --- up to constant factors. Consequently, the Cauchy walk identifies approximately convex targets in nearly optimal time, regardless of whether $\Delta_B$ or $\Delta_P$ serves as the governing parameter. In particular, the expected detection time differs from the universal lower bound by at most a polylogarithmic factor of $\mathcal{O}(\log^3 n)$.

\subsection{Convex targets are approximately convex}
\label{sec:conex}This subsection demonstrates that convex targets are inherently approximately convex. Recall that we always assume that targets are piecewise smooth  sets with well-defined surface area and projected surface area.

\begin{theorem}\label{thm:convex_apx_convex}
    Every convex target is approximately convex.
\end{theorem}
\begin{proof}
Consider a convex target $S$ of surface area $\Delta$. To establish Theorem~\ref{thm:convex_apx_convex}, it suffices to show that
\[
\frac{\Delta}{\Delta_B} \geq \frac{1}{6},
\qquad
\frac{\Delta}{\Delta_P} \le 6,
\]
where $\Delta$, $\Delta_B$, and $\Delta_P$ denote the surface area, face surface area, and largest projection area of $S$, respectively.

\begin{lemma}\label{lm:convex_b}
    $\frac{\Delta}{\Delta_B}\geq \frac{1}{6}$.
\end{lemma}
\begin{proof}

    Let $x,y \in S$ be points that maximize $\norm{x-y}$. Then the diameter of $S$ is defined as $D_1:=\norm{x-y}$. (For simplicity, we assume compactness of $S$ and therefore such $x$ and $y$ exist; otherwise, we can take $x$ and $y$, whose distance from each other is arbitrarily close to the diameter of $S$). This is realized along some unit vector $u_1 = (x-y)/\norm{x-y}$. Next, let $D_2$ be the diameter of the projection of $S$ in the orthogonal plane $u_1^\perp$ and let $u_2$ be its corresponding unit vector (in the direction of the vector realizing $D_2$). Consider a bounding box $B'$ of $S$ whose largest face $Q'$ is parallel to the plane spanned by $u_1$ and $u_2$. Let $\texttt{Area}(\pi_{Q'}(S))$ denote the area of the projection of $S$ onto $Q'$. Since $S$ is convex, its projection onto $Q'$ is a planar convex set containing two orthogonal lines of length $D_1$ and $D_2$, and therefore its area is $\texttt{Area}(\pi_{Q'}(S)) \geq \frac{1}{2}D_1 D_2$. Next, observe that since $Q'$ is the largest face of the bounding box $B'$, we have:

    \[
    \Delta_{B'} =\texttt{Area}(Q') = D_1 D_2.
    \]
    Combining the two bounds yields

    \[
    \texttt{Area}(\pi_{Q'}(S)) \geq \frac{1}{2}\Delta_{B'}.
    \]
    Finally, since $\texttt{Area}(\pi_{Q'}(S))$ is the area of a projection of $S$, the
    surface area of $S$ satisfies
    \[
    \Delta \ge \texttt{Area}(\pi_{Q'}(S)),
    \]
    hence $\Delta \geq \frac{1}{2}\Delta_{B'}$. Observe that, by the definition, $\Delta_B$ is the largest face area of the smallest bounding box and for any target we have
    $\Delta_{B'} \geq \frac{1}{3}\Delta_{B}$.
    Consequently, $\Delta \geq \frac{1}{6}\Delta_{B}$, which completes the proof.
\end{proof}

\begin{lemma}\label{lm:convex_p}
    $\frac{\Delta}{\Delta_P} \leq 6$.
\end{lemma}

\begin{proof}
     Recall that we assume that the boundary $\partial S$ of $S$ is  piecewise smooth.  Let $u \in \mathbb{R}^3$ be a unit vector, and let
    $\pi_{u^\perp}$ denote the orthogonal projection onto the plane orthogonal to $u$.

    Then the area of the orthogonal projection of $S$ onto $u^\perp$ is given by
    \[
    \operatorname{Area}\!\bigl(\pi_{u^\perp}(S)\bigr)
    = \frac{1}{2} \int_{\partial S} \lvert n(x)^{\top} u \rvert \, dA(x),
    \]
    where $n(x)$ denotes the  unit vector that is orthogonal to the surface $\partial S$ at $x$, and
    $dA(x)$ is the surface area element on $\partial S$ (See \cite{vouk1948projected} for a reference). By definition, $\Delta_P$ is the area of the largest projection of $S$ onto the six faces of its bounding box $B$. Since the box is composed of three pairs of parallel faces aligned with an orthonormal basis $\{e_1, e_2, e_3\}$, we have:
    \[
    \Delta_P \geq \frac{1}{3}\sum_{i=1}^3 \operatorname{Area}(\pi_{e_i^\perp}(S)) = \frac{1}{6} \int_{\partial S} \sum_{i=1}^3 |n(x)^\top e_i| \, dA(x)
    ,
    \]
 Note that $\sum_i |n(x)^\mathsf{T}e_i|$ is the $L_1$ norm of the unit normal $n(x)$. Given that $1 \leq \|n(x)\|_1 \leq \sqrt{3}$ for any unit vector in $\mathbb{R}^3$, integrating these bounds over the surface $\partial S$ yields $\Delta_P \geq \Delta/6$, which concludes the proof.

\end{proof}

Finally, combining Lemmas \ref{lm:convex_b} and \ref{lm:convex_p} with the definition of approximate convexity completes the proof of Theorem \ref{thm:convex_apx_convex}.
\end{proof}

\subsection{Supporting evidence for \texorpdfstring{$\Delta_P$}{Delta P} as a main parameter governing detection time}\label{sec:counterexample}

We have seen that for (approximately) convex targets, the parameters $\Delta$, $\Delta_B$, and $\Delta_P$ are all proportional, making the universal lower bound and the upper bound for the Cauchy strategy compatible. To further understand which parameter could best represent the running time, we next illustrate that neither $\Delta$ nor $\Delta_B$ can be used in the formula in Theorem~\ref{thm:upper-bound-cauchy-box} instead of $\Delta_P$.

\begin{itemize}
    \item {\bf Surface area $\Delta$.}
Consider a construction parameterized by an integer
$r$: define a target $S$ that fits within a large ball of radius
$r$, and consists of $\Theta(r^3)$ disjoint balls of constant radius. This target has total surface area
$\Theta(r^3)$, and if Theorem \ref{thm:upper-bound-cauchy-box} were to apply with this parameter, its detection time would be $\bigO{\frac{n \log^3 n}{r^3}}$. However, detecting this target necessarily entails detecting the bounding box $\texttt{Box}(S)$, which has surface area
$\Theta(r^2)$ and requires at least $\Omega({\frac{n}{r^2}})$ time, by the universal lower bound. When
$r\gg \log^3 n$, this leads to a contradiction, implying that in Theorem \ref{thm:upper-bound-cauchy-box}, the parameter $\Delta_P$ cannot be directly replaced with the surface area $\Delta$ of a general target.\\

\item {\bf Bounded surface area $\Delta_B$.} Consider a target $S$ consisting of two orthogonal line segments of length $L$ and
unit thickness. Define the inflated target as $S^{(r)} = \{rx \: \text{ with } \: x \in S \}$, then the surface area of the inflated target scales as $\sim r^2L$, implying that the maximal projection area satisfies $\Delta_P^{(r)}$$\approx c\, r^2 L$. In contrast, the surface area of the bounding box scales as $\Delta_B^{(r)}\approx c' \, r^2 L^2$.
Consequently,
\[
\frac{\Delta_P^{(r)}}{\Delta_B^{(r)}}
\text{$\displaystyle = \Theta\left(\frac{1}{L}\right)$}.
\]
Therefore, for sufficiently large $L$, this ratio falls below $1/36$, and the target fails to be approximately convex.

\end{itemize}

\section{Diffusive L\'evy walks searching a line target}\label{sec:simpleRW}
\subsection{L\'evy walks on the discrete torus}
In this section, we consider a discrete 3D torus $\mathbb{T}$ containing $n$ nodes, and a target $S$ that is a path of length $L$ parallel to one of the axes.  Recall that the searcher starts at a node chosen uniformly at random,
hence, we may assume without loss of generality that $S$ lies on the $x$-axis. The searcher moves over this graph according to a L\'evy walk $X$ of parameter $\mu$ following the general construction of \cite{clementi2021search}; a precise definition of this process is provided in Methods Section. We let $t_{\mathrm{detect}}^{X}(S)$ represent the expected time until the search process $X$ finds the target $S$ for the first time.

The goal of this section is to prove the following:

\begin{theorem}\label{thm:diff_lines}
     Let $X$ be a L\'evy walk with parameter $\mu > 2$ on the discrete three-dimensional cubic torus graph with $n$ nodes and let the target $S$ be a path of length $L$ parallel to one of the axes. Then,
    $$
    t_{\mathrm{detect}}^{X}(S) = \bigO{\frac{n\log{n}}{L}}.
    $$
 \end{theorem}

 It’s reasonable to assume that $X$ is a lazy random walk, where at each step, the probability of staying in the same position is $1/2$. Indeed, if we show that the lazy version of the process detects the target in time $\bigO{\frac{n \log n}{L}}$, the detection time for the non-lazy version will also be bounded by the same time.

\subsection{Upper bounds on  hitting (and mixing) times of L\'evy walks}

To prove the theorem, we first establish upper bounds on the hitting and mixing steps for a L\'evy walk with $\mu > 2$ on either the discrete 1D cycle or the 2D torus. While such bounds have been qualitatively noted in the literature (e.g. \cite{chiarini2021constructing}), we were unable to find a formal proof; we therefore provide a rigorous derivation below. Specifically, we prove that these values are bounded from above by the corresponding hitting and mixing steps of a simple random walk.

\begin{theorem}\label{th:2d_hit}
A L\'evy walk with parameter $\mu > 1$ starting from a fixed node on a one-dimensional cycle with $m$ nodes has expected hitting step $\bigO{m^2}$, whereas on a two-dimensional torus with $m^2$ nodes its expected hitting step is $\bigO{m^2 \log m}$.
\end{theorem}
The proof of this theorem is derived from several structural properties of the L\'evy walk on either the cycle or the two-dimensional torus. We categorize all processes satisfying these conditions as {isotropic walks}.

\begin{definition}\label{def:isotropic-walk}
    We say that a Markov chain on a torus of any dimension is an {isotropic walk} if it satisfies the following properties:
    \begin{enumerate}
        \item The process is translation invariant, meaning that the behavior of the walk looks the same from any starting point. Formally,
        \[\Prob(X(m+1)=z'+r\mid X(m)=z+r)=\Prob(X(m+1)=z'\mid X(m)=z)\] for any $z,z'$ and any translation $r$.
        \item The process is symmetric, i.e., the probability of moving from $u$ to $v$ is the same as moving from $v$ to $u$.
        \item The transition probability from any node to each of its adjacent neighbors is the same value $c$, where $c\geq c_0>0$ for some constant $c_0$ independent of the torus size.
    \end{enumerate}
\end{definition}
Theorem~\ref{th:2d_hit} follows from Theorem~\ref{th:2d_hit_gen} below, by applying it to the L\'evy walk on the two-dimensional torus and on the cycle, since both are isotropic walks. Indeed, the first two conditions of Definition~\ref{def:isotropic-walk}
hold by construction, while the third is satisfied because $c$, which is proportional to the probability of a step of unit length $\left(\sum_{l=1}^{l_{\max}} l^{-\mu}\right)^{-1}$,
remains lower-bounded by the same constant factor for all torus sizes. Formally,
\[
\left(\sum_{\ell=1}^{\ell_{\max}}\ell^{-\mu}\right)^{-1} \geq
\left(\sum_{\ell=1}^{\infty}\ell^{-\mu}\right)^{-1}
=
\frac{1}{\zeta(\mu)}.
\]

\begin{theorem}\label{th:2d_hit_gen}
 An isotropic walk starting from a fixed node on a one-dimensional cycle with $m$ nodes has expected hitting step $\bigO{m^2}$, whereas on a two-dimensional torus with $m^2$ nodes its expected hitting step is $\bigO{m^2 \log m}$.
\end{theorem}

\begin{proof}

    We show that the expected hitting step of an isotropic walk is bounded above (up to constants) by the expected hitting time of the simple random walk on the one or two dimensional discrete torus, respectively. Since the latter is known to satisfy the claimed  upper bounds (e.g., \cite{levin2017markov}), the same bounds follow for isotropic walks.

    \paragraph*{Hitting step on the two-dimensional torus}

    We can interpret the isotropic walk on the two-dimensional torus as a random walk on a complete weighted graph. Here, the edges represent all possible transitions, with weights assigned according to the corresponding transition probabilities of the original walk. Denote this graph $G$ and observe that it is transitive and undirected due to the definition of isotropic walk, hence the expected hitting time between node $a$ and $b$ is given by the product of the number of nodes in the graph and the effective resistance between those two nodes (this follows from Proposition 10.10 and Proposition 10.7 in \cite{levin2017markov}). In symbols,
    \begin{align}\label{eq:ht_G}
        \E_{a}[\mhit(b)] = \frac{m^2}{2}\eres{a}{b}{G}.
    \end{align}
    where $\mathbb{E}_a[\mhit(b)]$ denotes the expected hitting step of node $b$ when the process starts from node $a$, and $m^2$ is the total number of nodes in the two-dimensional torus. Observe that, by Rayleigh’s Monotonicity Law (Corollary 9.13 in \cite{levin2017markov}), deleting edges from $G$ cannot decrease the effective resistance. Hence, letting $G'$ be the graph obtained from $G$ by removing all edges but those between neighboring nodes on the two-dimensional torus, the right-hand side of ~\eqref{eq:ht_G} is at most $\frac{m^2}{2}\eres{a}{b}{G'}$. By isotropy, every edge of $G'$ has the same weight $0<c\leq 1$, rescaling all edge weights by a factor $1/c$ yields the graph $c^{-1}G'$, i.e., $G'$ with all edge weights equal to 1; this graph is exactly the two-dimensional discrete torus and note that its effective resistance $\eres{a}{b}{c^{-1}G'}$ is $c\eres{a}{b}{G'}$. Under this rescaling, the expected hitting step in $G$ is upper bounded by $\frac{m^2}{2c}\,\eres{a}{b}{c^{-1}G'}$. But $c$ is lower bounded by a constant and the hitting step of the simple walk on the two dimensional torus is $\Theta(m^2\eres{a}{b}{c^{-1}G'})$ (See Eq. (10.27) in \cite{levin2017markov}), so we have that
    \begin{align*}
        \E_{a}[\mhit(b)] = \bigO{\E_{a}^\frac{G'}{c}[\mhit(b)]} = \bigO{m^2\log m},
    \end{align*}
    where $\mathbb{E}^{\frac{G'}{c}}_{a}\!\left[\mhit(b)\right]$ denotes the expected hitting step of the node $b$ for the simple random walk on the two-dimensional torus, started from $a$; and the second equality follows from Proposition~10.21 in \cite{levin2017markov}.

    \paragraph*{Hitting step on the cycle}
    We proceed exactly as in the two-dimensional case; the only difference is that we now start from an isotropic walk on the cycle and we build the corresponding auxiliary graphs from its one-step transitions. Let $C$ be the weighted graph formed from a cycle in which each node is connected to every other node. For each pair of nodes, we assign to their edge the weight given by the one-step probability that the walk moves between them. Since the process is symmetric, this weight is the same in both directions, so $C$ is an undirected weighted graph. Now define $C'$ from $C$ by deleting all edges except those joining nearest neighbors on the cycle (i.e., nodes at  distance $1$ on the cycle graph). By translation invariance, $C'$ is a cycle in which every remaining edge has the same weight $0<c\leq 1$. If we rescale all edge weights by the factor $c^{-1}$, the resulting walk coincides with the simple random walk on the cycle. It is well known that for the simple random walk on a cycle with $m$ nodes, the maximal expected hitting step is $\bigO{m^2}$ (see Example~10.27 in \cite{levin2017markov} for a formal proof).
This completes the proof of Theorem~\ref{th:2d_hit_gen}.
\end{proof}

\subsection{Proof of Theorem~\ref{thm:diff_lines}}

Following the same steps as in the proof of Claim~\ref{claim:timeVSmoves}, one sees that the claim also applies in the discrete setting, where $\tau$ denotes the expected step length with respect to the Manhattan distance. For $\mu > 2$, this expectation is $\bigO{1}$ (see Fact~\ref{fact:exp+var}). Therefore, to prove Theorem~\ref{thm:diff_lines}, it suffices to prove:
\begin{align}\label{eq:discr_m_ub}
\max_a\mathbb{E}_a[\mhit^{X}(S)] = \bigO{\frac{n\log n}{L}},
\end{align}
where $\E_a\left[\cdot\right]$ denotes the expected value when $X$ starts from node $a$. Indeed,
\[
\E\left[\mhit^X(S)\right] = \E_\pi\left[\mhit^X(S)\right] \leq \max_{a}\E_a\left[\mhit^X(S)\right],
\]
where $\E_\pi[\cdot]$ denotes expectation under the uniform initial distribution $\pi$, and the first equality follows from the assumption that the initial node is chosen uniformly at random.

We are interested in the projections of the 3D walk $X$ on both the $x$-axis, and on the $yz$-plane. In particular, let $m_{\text{mix}}(\cdot)$ be the mixing step of the 3D random walk $X$ and $\mhit^{(yz)}(\cdot)$ be the hitting step of the 2D random walk $X_{yz}$.

\begin{theorem}\label{thm:SRW_line}
Consider the three-dimensional cubic torus with $n$ nodes. Let the target $S$ be a path of length $L$ parallel to one of the axes, and let $X$ be an irreducible, aperiodic and isotropic Markov chain on the torus.
Assume that $X$ satisfies the following:
\begin{itemize}
    \item
$X$ has uniform stationary distribution,
\item $\mmix(1/n^{2})=\bigO{n^{2/3}\log n}$, and \item
$\E_a\left[\mhit^{(yz)}(S_{yz})\right]=\bigO{n^{2/3}\log n}$ for every $a$ on the $yz$-plane.
\end{itemize}
Then
\[
\max_{a}\E_a\!\left[m_{\mathrm{hit}}^{X}(S)\right]=\bigO{\frac{n\log n}{L}}.
\]
\end{theorem}

\begin{proof}
    Let $T$ be a positive integer to be defined later and set $\tau_0:=0$. For $i\ge 1$, let $\tau_i$ be the first step after $\tau_{i-1}+T$ steps, at which the projected walk $X_{yz}$ visits the projected target set $S_{yz}$. Formally,
    \[
    \tau_i \;:=\; \inf\bigl\{\, m\ge \tau_{i-1}+T \;:\; X_{yz}(m)\in S_{yz}\,\bigr\}.
    \]
    Likewise, let $J$ be the smallest $i\geq 1 $ such that $X(\tau_i)\in S$, i.e., the smallest index at which the 3D walk hits $S$. Since the detection step is the first $m$ with $X(m)\in S$, we have the bound
    \[
    m^X_{\mathrm{hit}}(S)\ \leq \tau_J.
    \]
    that holds for every starting node. Moreover, by telescoping,
    \[
    \tau_J=\sum_{i=1}^{J}(\tau_i-\tau_{i-1})
          =\sum_{i=1}^{\infty}(\tau_i-\tau_{i-1})\mathbf 1_{\{J\ge i\}}.
    \]
    Taking expectations with respect to a fixed starting node $X(0)=a$,
    \begin{align}\label{eq:line_det}
    \mathbb E_a[m_{\mathrm{hit}}^{X}(S)]
    \le \sum_{i=1}^{\infty}\mathbb E_a\!\left[(\tau_i-\tau_{i-1})\mathbf 1_{\{J\ge i\}}\right].
    \end{align}
    For each $i$, apply the tower property with respect to $X(0),X(1),\dots, X({\tau_{i-1}+T})$, i.e., the initial position $X(0)$, together with the first $\tau_{i-1}+T$ steps:
    \begin{align}\label{eq:tower_decomp}
    \mathbb E_a\!\left[(\tau_i-\tau_{i-1})\mathbf 1_{\{J\ge i\}}\right]
    =
    \mathbb E_a\!\left[\mathbf 1_{\{J\ge i\}}\,
    \mathbb E\!\left[\tau_i-\tau_{i-1}\mid X(0),\dots, X({\tau_{i-1}+T})\right]\right],
    \end{align}
where the equality holds because $\mathbf 1_{\{J\ge i\}}$ is a function of $X(0),\dots, X({\tau_{i-1}+T})$.

    \begin{claim}
       $\E_a[m_{\mathrm{hit}}^{X}(S)] \leq \left(T + \bigO{n^{2/3}\log{n}}\right)\mathbb{E}_a\left[J\right].$
    \end{claim}
    \begin{proof}
        Let $\mathbb{E}_{a}\!\left[m_{\mathrm{hit}}^{(yz)}(S_{yz})\right]$ be the expected value of $m_{\mathrm{hit}}^{(yz)}(S_{yz})$ when $a$ is the initial location of the searcher in the projected chain. Observe
        that after the deterministic waiting period of $T$ steps after $\tau_{i-1}$ (and before $\tau_i$), the additional expected steps required to
        hit $S_{yz}$ depends only on the projected process started from position $X_{yz}(\tau_{i-1}+T)$.
        Hence, the inner conditional expectation in ~\eqref{eq:tower_decomp} can be rewritten as
    \begin{align*}
    \mathbb{E}\!\left[\tau_i-\tau_{i-1}\,\middle|\, X(0),\dots, X(\tau_{i-1}+T)\right]
    = T + \mathbb{E}_{X_{yz}(\tau_{i-1}+T)}\!\left[m_{\mathrm{hit}}^{(yz)}(S_{yz})\right].
    \end{align*}
    By assumption, the expected hitting step is upper bounded by $\bigO{n^{2/3}\log{n}}$, which is independent of $i$. Thus, we can rewrite ~\eqref{eq:tower_decomp} as follows
     \begin{align}\label{eq:tower_ub}
        \mathbb E_a\!\left[(\tau_i-\tau_{i-1})\mathbf 1_{\{J\ge i\}}\right] \leq \left(T + \bigO{n^{2/3}\log{n}}\right)\mathbb{E}_a\!\left[1_{\{J\ge i\}}\right].
    \end{align}
   Substituting ~\eqref{eq:tower_ub} in ~\eqref{eq:line_det}, we obtain
    \begin{align}\label{eq:ub_walk}
        \mathbb E_a[m_{\mathrm{hit}}^{X}(S)] &\leq \left(T + \bigO{n^{2/3}\log{n}}\right)\sum_{i\geq 1}\mathbb{E}_a\!\left[1_{\{J\ge i\}}\right] \notag \\
        &= \left(T + \bigO{n^{2/3}\log{n}}\right)\mathbb{E}_a\left[J\right].
    \end{align}
    This completes the proof of the claim.
\end{proof}

\begin{claim}
    $\mathbb{E}_a[J]=\Theta\!\left(\frac{n^{1/3}}{L}\right)$
\end{claim}
    \begin{proof}

    By assumption, the mixing step of the process on the whole torus is $\mmix(1/n^2) \leq Cn^{2/3}\log n$ for some positive constant $C$ and sufficiently large $n$. After $\mmix(1/n^2)$ steps the distribution of the process on the torus is approximately uniform. In particular, for any node $w$, the probability that $X$ is on $w$ is
\begin{align*}
\left[ \frac{1}{n} - \frac{1}{n^2}, \frac{1}{n} + \frac{1}{n^2} \right],
\end{align*}
Moreover, for every $m \geq \mmix(1/n^2)$ and every two nodes $u$ and $w$ that are on the $x$-axis:
\begin{align*}
    \Prob\!\big(X_x(\tau_{i-1}+m) = w \mid X_x(\tau_{i-1}) = u\big) = \Theta\left(\frac{1}{n^{1/3}}\right).
\end{align*}
However, this does not by itself imply that $\Prob\!\big(X_x(\tau_{i}) = w \mid X_x(\tau_{i-1}) = u\big) = \Theta\left(\frac{1}{n^{1/3}}\right)$. Indeed, $\tau_i$ is defined in terms of the trajectory of the $yz$-projection. If the $x$- and $yz$-projections are correlated, conditioning on the event that determines $\tau_i$ may provide additional information about $X_x(\tau_i)$, and therefore bias its distribution away from uniformity. We now show formally that, under our assumptions, $X_x(\tau_i)$ remains approximately uniform. Set $T = Cn^{2/3}\log n \geq \mmix(1/n^2)$, let $u, v$ be two nodes on the $x$-axis and $r$ a translation along the $x$-axis, then
\begin{align}
    & \Prob(X_x(\tau_i)=v\mid X_x(\tau_{i-1}) = u) \notag \\
    &= \sum_w\Prob(X_x(\tau_i)=v\mid X(\tau_{i-1}+T) = w)\Prob(X(\tau_{i-1}+T) = w\mid X_x(\tau_{i-1}) = u ) \notag \\
    &= \sum_w\Prob(X_x(\tau_i)=v+r\mid X(\tau_{i-1}+T) = w+r)\Theta \Big( \Prob(X(\tau_{i-1}+T) = w+r\mid X_x(\tau_{i-1}) = u )\Big) \notag \\
    &= \Theta\Big( \Prob(X_x(\tau_i)=v+r\mid X_x(\tau_{i-1}) = u) \Big), \notag
\end{align}
where the second equality follows from the translation invariance of isotropic processes (Definition~\ref{def:isotropic-walk}) and the uniform stationary distribution. Thus, since all $n^{1/3}$ possible translations of $v$ cover the entire $x$-axis and $\sum_v  \Prob(X_x(\tau_i)=v\mid X_x(\tau_{i-1}) = u) = 1$, we have:
\begin{align}\label{eq:proba_xproj}
\Prob(X_x(\tau_i)=v\mid X_x(\tau_{i-1}) = u) = \Theta\Big(\frac{1}{n^{1/3}}\Big).
\end{align}

Observe that~\eqref{eq:proba_xproj} implies that if we know the starting point and that $X_x(\tau_{i-1})$ did not hit the target, then the probability of hitting it at step $\tau_{i}$ is $\Theta(L/n^{1/3})$. That is,
\begin{align*}
    & \Prob\!\left(
        X_x(\tau_i)\in S_x
        \,\middle|\,
        \bigcap_{j=1}^{i-1} X_x(\tau_j)\notin S_x,\, X(0)=a
    \right) \\
    &= \sum_{u}
    \Prob\!\left(
        X_x(\tau_i)\in S_x
        \,\middle|\,
        X_x(\tau_{i-1})=u
    \right)
    \Prob\!\left(
        X_x(\tau_{i-1})=u
        \,\middle|\,
        \bigcap_{j=1}^{i-1} X_x(\tau_j)\notin S_x,\, X(0)=a
    \right) \\
    &= \Theta\!\left(\frac{L}{n^{1/3}}\right) \sum_{u}
    \Prob\!\left(
        X_x(\tau_{i-1})=u
        \,\middle|\,
        \bigcap_{j=1}^{i-1} X_x(\tau_j)\notin S_x,\, X(0)=a
    \right)\\
    &= \Theta\!\left(\frac{L}{n^{1/3}}\right),
\end{align*}
    where the second equality follows from the fact that $S_x$ has cardinality $L$ together with~\eqref{eq:proba_xproj}, and the third equality holds because the sum is taken over all possible events and therefore equals~1.

 Finally, we can prove that the expected value of $J$ is $\mathcal{O}\left(n^{1/3}/L\right)$. In symbols,

\begin{align*}
    \Prob(J \ge k \mid X(0) =a)
    &= \Prob\!\left(\bigcap_{i=1}^{k-1} X_x(\tau_i)\notin S_x\,\middle|X(0)=a\right) \\
    &= \prod_{i = 1}^{k-1}
    \Prob\!\left(
        X_x(\tau_i)\notin S_x
        \,\middle|\,
        \bigcap_{j=1}^{i-1} X_x(\tau_j)\notin S_x,\, X(0)=a
    \right) \\
    &= \prod_{i = 1}^{k-1}\left[
        1-\Prob\!\left(
            X_x(\tau_i)\in S_x
            \,\middle|\,
            \bigcap_{j=1}^{i-1} X_x(\tau_j)\notin S_x,\, X(0)=a
        \right)
    \right] \\
    &= \prod_{i = 1}^{k-1}\left[1-\Theta\!\left(\frac{L}{n^{1/3}}\right)\right] \\
    &= \left[1-\Theta\!\left(\frac{L}{n^{1/3}}\right)\right]^{k-1}.
\end{align*}
Consequently, $J$ has a geometric tail, and hence $\mathbb{E}_a[J]=\Theta\!\left(\frac{n^{1/3}}{L}\right)$.
\end{proof}
Altogether, substituting this bound and the definition of $T$ into ~\eqref{eq:ub_walk} yields
    $\mathbb{E}_a\left[m_{\mathrm{hit}}^{X}(S)\right]=\mathcal{O}\!\left(\frac{n\log n}{L}\right)$.
    This completes the proof of Theorem~\ref{thm:SRW_line}.\end{proof}

Observe that the L\'evy walk on the three-dimensional torus is irreducible and aperiodic, and admits the uniform distribution as its stationary measure ($X$ satisfies Proposition 1.20 in \cite{levin2017markov} with $\pi=1/n$). Therefore, it suffices to show that $\mmix(1/n^{2})$ and $\max_a\E_a\left[\mhit^{(yz)}(S_{yz})\right]$ are both $\bigO{n^{2/3}\log n}$ in order to apply Theorem~\ref{thm:SRW_line} to the L\'evy walk.

\begin{lemma}\label{th:2d_hit_gen_proj}
    $X_x$ has expected hitting step $\bigO{n^{2/3}}$ and $X_{yz}$ has expected hitting step $\bigO{n^{2/3}\log n}$ for every starting node.
\end{lemma}

Lemma~\ref{th:2d_hit_gen_proj} follows from the observation that $X_x$ and $X_{yz}$
are isotropic walks on a cycle with $n^{1/3}$ nodes and a two-dimensional torus
with $n^{2/3}$ nodes, respectively. Hence, the bounds on the hitting steps follow by applying Theorem~\ref{th:2d_hit_gen} to the projected L\'evy walks. Finally, we upper bound $\mmix(1/n^{2})$, the {\em mixing step} on the torus, i.e., the minimum number of steps needed for the total variation distance to the stationary distribution to drop below $1/n^2$.

\begin{lemma}\label{th:3d_mix}
    The mixing step of L\'evy walk on the three-dimensional torus with $\mu>2$ is $\mmix(1/n^{2})=\bigO{n^{2/3}\log n}$.
\end{lemma}
\begin{proof}
The probability that the chain $X$ with $\mu>2$ makes a step of length $1$ is
\begin{align*}
p_1 = \Big( \sum_{\ell=1}^{\lmax}\ell^{-\mu}\Big)^{-1} \geq \Big( \sum_{\ell=1}^{\infty}\ell^{-2}\Big)^{-1} = \frac{6}{\pi^2}.
\end{align*}
Thus, the probability of moving to a neighboring node is at least $\frac{1}{2}p_1\frac{1}{6} = p_1/12$, where $1/2$ comes from the laziness of $X$ and $1/6$ come from the fact of having $6$ neighbors for each node. Similarly, the probability of moving to a node at distance $1$ for a lazy random walk $Q$ is exactly $1/12$. Thus, for every couple of nodes $u$ and $v$ at distance $1$, we have
\begin{align}
    \Prob_X(v,u) \geq p_1 \Prob_Q(v,u),
\end{align}
where $\Prob_X(v,u)$ and $\Prob_Q(v,u)$ denote the one-step transition probabilities of $X$ and $Q$, respectively. Now, let $f$ be any function on the nodes of the torus and consider the Dirichlet form of the two chains defined as in Section 13.2 of \cite{levin2017markov}:
\[
\dirform{}{f}
:=
\frac{1}{2n}
\sum_{u,v}
\Prob(u,v)\bigl(f(u)-f(v)\bigr)^2.
\]
 Note that since all the terms are positive, we have
\begin{align}
    \dirform{X}{f} &\geq \frac{1}{2n}\sum_u \sum_{v: d(u,v)=1} \Prob_X(u,v)\bigl(f(u)-f(v)\bigr)^2 \notag \\
    & \geq \frac{p_1}{2n}\sum_u \sum_{v: d(u,v)=1} \Prob_Q(u,v)\bigl(f(u)-f(v)\bigr)^2 = p_1 \dirform{Q}{f}. \notag
\end{align}
which implies that the spectral gap (See Section 12.2 and Lemma 13.7 in \cite{levin2017markov} for a formal definition)
of $X$ is at least $p_1$ times the spectral gap $\gamma_Q$ of $Q$. That is,
\begin{align*}
    \gamma_X \geq p_1 \gamma_Q.
\end{align*}
We can now use the standard bound on the mixing time in terms of the spectral gap (Theorem 12.4 in \cite{levin2017markov}) to obtain
\begin{align}\label{eq:relaxation}
    \mmix(1/n^{2})\leq \frac{1}{\gamma_X} \log{(n^{3})} \leq \frac{1}{p_1\gamma_Q}\log{(n^{3})}
\end{align}
Now, in order to conclude the proof, we show that the spectral gap of the lazy random walk on the three-dimensional torus is of order $n^{-2/3}$.

Consider first the lazy simple random walk $R$ on the line with $n^{1/3}$ nodes, which stays still with probability $1/2$ and moves to either neighbor with probability $1/4$. As shown in Section 12.3.1 of \cite{levin2017markov}, its spectral gap is $\Theta(n^{-2/3})$. Now observe that the lazy simple random walk $Q$ on the three-dimensional torus can be viewed as a product chain of three copies of $R$, i.e., at every step, one of the three coordinates is chosen uniformly and updated according to $R$, while the other two coordinates remain unchanged. Therefore, by Corollary~12.13 in \cite{levin2017markov},

$$
\gamma_Q
=
\frac13\gamma_R
=
\Theta(n^{-2/3}).
$$
Finally, plugging the bounds on $p_1$ and $\gamma_Q$ into~\eqref{eq:relaxation} yields

$$
\mmix(1/n^{2}) = \bigO{n^{2/3}\log{(n)}}.
$$
\end{proof}

Altogether, Lemmas~\ref{th:3d_mix} and~\ref{th:2d_hit_gen_proj} enable us to apply Theorem~\ref{thm:SRW_line} to the L\'evy walk on the torus, thereby completing the proof of~\eqref{eq:discr_m_ub}, thus completing the proof of Theorem~\ref{thm:diff_lines}.

\section{Model}\label{sec:model}
\subsection{L\'evy walk process on the torus}

The proof of the following fact is given in the Supplementary Material of~\cite{GuinardKorman2021} (see Claim~2 therein).
\begin{fact}\label{fact:exp+var}
Let $\mu\in (1,3]$ and consider the L\'evy walk $Z^{\mu}$ (or $X^{\mu}$) with maximal step length $\ell_{\max}$.
The expected step length (and consequently, the expected duration of a single step) $\tau=E(\norm{V})$ scales as
\begin{align}\label{obs:tau}
\tau =
\begin{cases}
\Theta(\ell_{\max}^{\,2-\mu}), & \text{if } \mu \in (1,2),\\[4pt]
\Theta(\log \ell_{\max}), & \text{if } \mu = 2,\\[4pt]
\Theta(1), & \text{if } \mu \in (2,3],
\end{cases}
\end{align}
while the variance $\sigma^2$ and the second moment $M$ of the step length satisfy
\begin{align}\label{obs:var}
\sigma^2 = \Theta(M) =
\begin{cases}
\Theta(\ell_{\max}^{\,3-\mu}), & \text{if } \mu \in (1,3),\\[4pt]
\Theta(\log \ell_{\max}), & \text{if } \mu = 3.
\end{cases}
\end{align}
\end{fact}
\noindent Recall that in the case of the 3D torus, the maximal step length satisfies $\ell_{\max} = \sqrt[3]{n}/2$.

\subsection{Temporal scaling and step-time normalization}

We assume that the speed of the process is 1, so
 that traveling a distance $\ell$ between two points takes $\ell$ units of time.
We further assume that the scanning phase between successive ballistic displacements has constant duration. In fact, without loss of generality, we may assume that it takes zero time, i.e., $\tau_{\mathrm{scan}} = 0$. Indeed, by Fact~\ref{fact:exp+var},
the expected time required to traverse a single ballistic segment is $\tau = \Omega(1)$, whereas $\tau_{\mathrm{scan}} = \mathcal{O}(1)$.  Consequently, the expected time to complete a full step (ballistic displacement plus scanning) satisfies
\[
\tau + \tau_{\mathrm{scan}} = \Theta(\tau).
\]
 Let $T(m)$ denote the total time elapsed after $m$ steps of the walk:
\[
T(m) = \sum_{s=1}^{m} \|V(s)\|,
\]
where $V(s) = (V_1(s),V_2(s),V_3(s))$ is the relocation vector chosen at step $s$, and
\[
\|V(s)\| = \sqrt{V_1(s)^2 + V_2(s)^2 + V_3(s)^2}
\]
denotes its $\ell_2$-length. Let $m_{\mathrm{detect}}^X(S)$ be the (random) number of steps taken by $X$ until the target $S$ is first detected (equivalently, the first time $m$ such that $X(m)\in S$). By definition, the expected detection time is
\[
t_{\mathrm{detect}}^X(S)
= \mathbb{E}\!\left[T\!\left(m_{\mathrm{detect}}^X(S)\right)\right].
\]

\section{Empirical-trajectory analysis}\label{empirical}

\subsection{Dataset and condition selection}

We used the manually curated consensus data in the Leukocyte Tracking Database (LTDB) \cite{pizzagalli2018leukocyte}. Each acquisition in the dataset consists of a time-ordered sequence of frames, where each frame is a three-dimensional image volume acquired at one time point. For every annotated leukocyte visible in a frame, the dataset records its track identifier and three-dimensional centroid coordinates in micrometers; the acquisition interval gives the elapsed time between consecutive frames. Observations carrying the same identifier across frames form the raw track of one cell. The selected conditions comprise NK cells and neutrophils following influenza vaccination \cite{farsakoglu2019influenza,pizzagalli2019characterization} and T cells in a humanized-mouse HIV model \cite{murooka2012hiv}.

We chose these as illustrative conditions because they provide a compact set of virus-related immune contexts imaged in the same anatomical site, namely, the popliteal lymph node, while spanning innate (NK cells and neutrophils) and adaptive (T cells) leukocyte motility. To avoid pooling biologically or temporally distinct observations, we grouped the tracks by experimental condition and acquisition interval. Among condition groups represented by more than one recorded video, we selected the sampling condition that yielded the largest number of complete interior directed step vectors in the trajectories obtained after preprocessing. These conditions are intended as illustrative examples rather than as a controlled comparison among cell types.

\subsection{Preprocessing}

The frame-sampled raw tracks cannot be used directly as trajectories in our analysis: consecutive positions in a raw track are observations at fixed acquisition times, whereas the model requires a sequence of directed step vectors delimited by reorientations; a single directed step vector can therefore span several frames. We converted the raw tracks into trajectories using the following preprocessing procedure. For each track segment, we computed the frame-to-frame displacement $\mathbf{u}_i=\mathbf{x}_{i+1}-\mathbf{x}_i$ between consecutive cell locations. Following the classification of $90^\circ$ reorientations as leukocyte large-turns in~\cite{georgantzoglou2022two}, we used $90^\circ$ as the threshold for constructing directed step vectors. Consecutive frame-to-frame displacements were assigned to the same step vector as long as the three-dimensional turning angle between them remained at most $90^\circ$. A displacement whose turning angle exceeded $90^\circ$ began a new step. The vector assigned to each step was the sum of its constituent frame-to-frame displacements.

\subsection{Model fitting}

The ordered sequence of the resulting directed step vectors constituted the preprocessed trajectory used for fitting and replay. The first and last vectors of each trajectory were excluded from fitting because the observation window may censor their full extent, but they were retained as part of the trajectory for finite-budget replay.

For the magnitudes $\ell$ of complete directed displacement vectors, we fitted the following scaled form of the step-length density used in the main text:
\[
f(\ell\mid\mu,c,L_{\max})=
\begin{cases}
a_c, & 0\leq \ell\leq c,\\[3pt]
a_c(\ell/c)^{-\mu}, & c<\ell\leq L_{\max},\\[3pt]
0, & \ell>L_{\max},
\end{cases}
\]
where
\[
a_c=\left[c+\int_c^{L_{\max}}(\ell/c)^{-\mu}\,d\ell\right]^{-1}
\]
is the normalizing constant. In the theoretical model, lengths are expressed in normalized units and the crossover is fixed at $c=1$, as in the step-length density defined in the main text. For the empirical analysis, $c$ was first specified in physical units and the step vectors were subsequently normalized by this scale. We used the $2\,\mu\mathrm m\,\mathrm{min}^{-1}$ speed threshold employed by Pizzagalli et al.~to calculate a leukocyte arrest coefficient \cite{pizzagalli2019characterization}. Converting the threshold according to each acquisition interval gave $c=0.5\,\mu\mathrm m$ for T cells and neutrophils sampled every 15 s and $c=1.0\,\mu\mathrm m$ for NK cells sampled every 30 s. Thus, this analysis uses the published arrest threshold as a common value across the three cell groups, for which we found no condition-specific estimates of the crossover. Because the fitted exponents are sensitive to the choice of $c$, the reported values should be viewed as preliminary. Within each condition, $L_{\max}$ was the largest magnitude among the complete directed step vectors, and $\mu\in[1.001,10]$ was estimated by maximum likelihood using the exact normalization above.

\subsection{Finite endpoint-detection replay analysis (Fig.~6 in main text)}

Step vectors and their magnitudes were normalized by the condition-specific $c$. For each condition, simulations used a periodic cube of side $W'=2L'_{\max}$ and detection radius $d'=1$. Ball, disk, and line targets were selected so that projected surface areas $\Delta_P$ were identical.

For the empirical construction, the direction of each step vector in a preprocessed trajectory was sampled independently and uniformly on the sphere, while preserving its magnitude, its position in the sequence, and the trajectory's endpoint count and path-length budget. Because the raw tracks were obtained from slab-shaped imaging volumes, the directions inferred from them can reflect anisotropy imposed by the acquisition geometry. For the fitted construction, step-vector magnitudes were drawn independently from the fitted condition-level density and their directions were sampled uniformly on the sphere. For every preprocessed trajectory, replay construction, target shape, and retained $\Delta_P$, we performed 10,000 replays from uniformly distributed initial positions. Detection was checked after every step-vector endpoint. The shaded regions in Fig.~6 of the main text are pointwise 95\% percentile-bootstrap confidence intervals obtained from 1,000 resamples of complete source raw tracks.

\bibliographystyle{unsrtnat}
\bibliography{references}

\end{document}